\documentclass[11pt,a4paper]{article}

\usepackage[top=1in, bottom=1in, left=1in, right=1in]{geometry}

\usepackage{enumitem}
\usepackage{amsthm}
\usepackage{amsfonts}

\newtheorem{theorem}{Theorem}[section]
\newtheorem{lemma}[theorem]{Lemma}
\newtheorem{corollary}[theorem]{Corollary}
\newtheorem{proposition}[theorem]{Proposition}
\newtheorem{definition}[theorem]{Definition}
\newtheorem{problem}[theorem]{Problem}
\newtheorem{observation}[theorem]{Observation}
\newtheorem{remark}[theorem]{Remark}
\newtheorem{assumption}[theorem]{Assumption}

\usepackage[lined, algoruled, linesnumbered]{algorithm2e}
\usepackage{color,soul}
\definecolor{cblue}{rgb}{0.36, 0.54, 0.66}
\definecolor{cred}{rgb}{1, 0, 0}

\usepackage{graphicx}

\usepackage{lineno}

\usepackage{hyperref,xcolor}
\definecolor{winered}{rgb}{0.5,0,0}
\hypersetup
{
    pdfborder={0 0 0},
    colorlinks=true,
    linkcolor={winered},
    urlcolor={winered},
    filecolor={winered},
    citecolor={winered},
    linktoc=all,
}

\usepackage{comment}

\author{Evangelos Kosinas\thanks{University of Ioannina and Archimedes, Athena RC, Greece. E-mail: \texttt{ekosinas@cs.uoi.gr}. This work has been partially supported by project MIS 5154714 of the National Recovery and Resilience Plan Greece 2.0 funded by the European Union under the NextGenerationEU Program.}}

\begin{document}

\title{An Optimal $3$-Fault-Tolerant Connectivity Oracle}

\maketitle

\begin{abstract}
We present an optimal oracle for answering connectivity queries in undirected graphs in the presence of at most three vertex failures. Specifically, we show that we can process a graph $G$ in $O(n+m)$ time, in order to build a data structure that occupies $O(n)$ space, which can be used in order to answer queries of the form \emph{``given a set $F$ of at most three vertices, and two vertices $x$ and $y$ not in $F$, are $x$ and $y$ connected in $G\setminus F$?''} in constant time, where $n$ and $m$ denote the number of vertices and edges, respectively, of $G$. The idea is to rely on the DFS-based framework introduced by Kosinas [ESA'23], for handling connectivity queries in the presence of multiple vertex failures. Our technical contribution is to show how to appropriately extend the toolkit of the DFS-based parameters, in order to optimally handle up to three vertex failures. Our approach has the interesting property that it does not rely on a compact representation of vertex cuts, and has the potential to provide optimal solutions for more vertex failures. Furthermore, we show that the DFS-based framework can be easily extended in order to answer vertex-cut queries, and the number of connected components in the presence of multiple vertex failures. In the case of three vertex failures, we can answer such queries in $O(\log n)$ time.
\end{abstract}

\clearpage

\small

\tableofcontents

\clearpage

\normalsize

\section{Introduction}

\subsection{Problem definition}
This paper is motivated by the following problem. Let $G$ be an undirected graph, and let $d_{\star}$ be a positive integer. The goal is to efficiently build a data structure, that occupies as little space as possible, so that we can efficiently answer queries of the following form:\\

``\emph{Given a set of vertices $F$, with $|F|\leq d_{\star}$, that have failed to work, and two vertices $x$ and $y$ not in $F$, are $x$ and $y$ connected in $G\setminus F$?}".\\

This problem has received significant attention in the last few years \cite{DBLP:journals/siamcomp/DuanP20,DBLP:conf/icalp/PilipczukSSTV22,DBLP:conf/focs/LongS22,DBLP:conf/esa/Kosinas23,DBLP:conf/icalp/LongW24}, and it can be characterized as an instance of \emph{``emergency planning"}~\cite{DBLP:conf/focs/PatrascuT07}, or as a data structure problem in the \emph{fault-tolerant}, or \emph{sensitivity} setting. 
The potential usefulness of such data structures in real-world applications is obvious: since networks in real life are prone to failures, it may be worthwhile to spend some time in order to preprocess the graph, so that we can deal efficiently with (temporary) malfunctions.

What makes this particular problem so interesting, and justifies the variety of the approaches that have been taken so far, is the difficulty in simultaneously optimizing all the parameters of efficiency: i.e., the preprocessing time, the space usage, the time to answer the queries, and (in some cases) the time to separately handle the set $F$ of failed vertices. Thus, the solutions that exist in the literature provide various trade-offs, and none of them outcompetes the others in every measure of efficiency (see Table~\ref{table:bounds1}).

However, for the cases where $d_{\star}\in\{1,2\}$, one can provide an optimal solution by relying on some older works. Specifically, for $d_{\star}=1$, one can use a simple DFS-based approach~\cite{DBLP:journals/siamcomp/Tarjan72}, and for $d_{\star}=2$ one can use the SPQR trees~\cite{DBLP:journals/algorithmica/BattistaT96,DBLP:journals/siamcomp/HopcroftT73,DBLP:conf/gd/GutwengerM00}, that represent compactly all $2$-vertex cuts of a biconnected graph. These cases admit of an \emph{optimal} solution in the sense that the data structures can be constructed in $O(n+m)$ time, they occupy $O(n)$ space, and they can answer any connectivity query in $O(1)$ time, where $n$ and $m$ denote the number of vertices and edges, respectively, of $G$. That these preprocessing and query times are optimal is obvious. The claim of the optimality of the $O(n)$ space usage is justified by the $\Omega(\min\{m,d_{\star}n\})$ lower bound on the bits of space usage that are provably necessary for some classes of graphs \cite{DBLP:journals/siamcomp/DuanP20}. 

Two related kinds of queries, that do not involve a particular pair of vertices, but enquire for the global impact of the failed vertices on the connectivity of $G$, are: 

\begin{enumerate}
\item{\emph{``Is $G\setminus F$ connected?"} (vertex-cut query)}
\item{\emph{``What is the number of connected components of $G\setminus F$?"}}
\end{enumerate}

A conditional lower bound for vertex-cut queries was given in \cite{DBLP:conf/focs/LongS22}, and only very recently \cite{arxivMerav} have provided an efficient solution for them. The optimal oracles mentioned above for $d_{\star}\in\{1,2\}$ can also answer queries $1$ and $2$. Specifically, a DFS-based approach can easily answer both such queries in constant time when $|F|=1$, and an $O(n)$-time preprocessing on the SPQR-based data structure can answer at least the first kind of queries in $O(\log n)$ time (as mentioned in \cite{arxivMerav}).

In this paper, it is sufficient to assume that the input graph $G$ is connected. This is because, if one has an oracle (for any of the problems that we have mentioned), that works for connected graphs, then one can simply initialize an instance of the oracle on every connected component of the graph, and thus have in total an oracle providing the same bounds as the original. Furthermore, we can use the sparsifier of Nagamochi and Ibaraki~\cite{DBLP:journals/algorithmica/NagamochiI92}, in order to replace every ``$m$" in the bounds that we state with ``$\min\{m,d_{\star}n\}$" (with the exception of the preprocessing time, in which ``$m$" must necessarily appear as an additive factor).

\subsection{Our contribution}
Our main contribution is an optimal oracle for the case $d_{\star}=3$. This can be stated as the following:

\begin{theorem}
\label{theorem:main}
We can process an undirected graph $G$ with $n$ vertices and $m$ edges in $O(n+m)$ time, in order to build a data structure with $O(n)$ size, which can answer queries of the form ``given a set $F$ of at most three vertices, and two vertices $x$ and $y$ not in $F$, are $x$ and $y$ connected in $G\setminus F$?'' in constant time. 
\end{theorem}

We note that Theorem~\ref{theorem:main} constitutes a genuine theoretical contribution, since the existence of an optimal oracle for $d_{\star}>2$ was an open problem prior to our work (see our discussion in Section~\ref{section:previouswork}). Furthermore, the solution that we provide has the following notable characteristics:

\begin{enumerate}[label={(\arabic*)}]
\item{It encompasses the cases $d_{\star}\in\{1,2\}$, and has the potential to be extended to an optimal solution for $d_{\star}>3$.}
\item{The underlying data structure does not rely on a computation or compact representation of vertex cuts. Instead, we work directly on a DFS tree, and this enables us to report efficiently the number of connected components upon removal of a set of vertices.}
\item{The most complicated data structures that we use are: the optimal DSU data structure of Gabow and Tarjan~\cite{DBLP:journals/jcss/GabowT85} (that works on a predetermined tree of unions) on the RAM model of computation, an optimal level-ancestors query data structure~\cite{DBLP:journals/tcs/BenderF04}, an optimal RMQ data structure (e.g., \cite{DBLP:conf/cpm/FischerH06}), and an optimal oracle for NCA queries (e.g., \cite{DBLP:journals/siamcomp/BuchsbaumGKRTW08}). Thus, in practice one could use alternative (but still very efficient) implementations of those data structures, in order to have a working oracle operating in near-optimal efficiency.}
\end{enumerate}

In order to establish Theorem~\ref{theorem:main}, we rely on the DFS-based framework introduced by Kosinas~\cite{DBLP:conf/esa/Kosinas23}, for designing an oracle for any fixed $d_{\star}$. The general idea in \cite{DBLP:conf/esa/Kosinas23} is to use a DFS tree $T$ of the graph, and determine the connectivity relation in $G\setminus F$ between some specific subtrees of $T\setminus F$, which are called \emph{internal components} (see Section~\ref{section:internalAndHanging}). This connectivity relation is sufficient in order to be able to answer efficiently any connectivity query. The input $d_{\star}$ for constructing the oracle is vital in order to compute a set of DFS-based parameters that are used in order to handle up to $d_{\star}$ failures. What differs for us in the case $d_{\star}=3$, is that we choose an alternative set of DFS-based parameters (coupled with, and motivated by, a much more involved case analysis for the internal components), which can be computed in linear time, occupy $O(n)$ space, and, if handled carefully, can be used in order to establish the connectivity between the internal components in constant time. (After that, every connectivity query on $G\setminus F$ can be handled precisely as in \cite{DBLP:conf/esa/Kosinas23}.) 

One of the main attributes of the oracle from \cite{DBLP:conf/esa/Kosinas23}, that makes it particularly attractive for our purposes, is its simplicity, both in itself, and in relation to all the other approaches that have been taken so far for general $d_{\star}$. In fact, we show that we can easily extend the underlying data structure, so that it can answer efficiently both vertex-cut queries, and queries for the number of connected components of $G\setminus F$. Specifically, we have the following (see Section~\ref{section:numberOfComponents}):

\begin{theorem}
\label{theorem2}
Let $G$ be an undirected graph with $n$ vertices and $m$ edges, and let $d_{\star}$ be a positive integer. We can process $G$ in $O(d_{\star}m\log{n})$ time, in order to build a data structure with size $O(d_{\star}m\log{n})$, which can answer queries of the form ``given a set $F$ of at most $d_{\star}$ vertices, what is the number of the connected components of $G\setminus F$?", in $O((d^4+ 2^dd^2)\log{n})$ time, where $d=|F|$.
\end{theorem}

Notice that the oracle in Theorem~\ref{theorem2} can also answer the corresponding vertex-cut query within the same time bound, because, if we are given a set of vertices $F$, then $G\setminus F$ is connected if and only if $G\setminus F$ consists of only one connected component. (Recall our convention that $G$ is connected.)  

Our method for augmenting the data structure from \cite{DBLP:conf/esa/Kosinas23} so that we can provide Theorem~\ref{theorem2}, implies the following improved bound in the case where $G$ is $d_{\star}$-connected (i.e., in the case where $G$ is so ``well-connected", that we have to remove at least $d_{\star}$ vertices in order to disconnect it; see Section~\ref{section:numberOfComponents}):

\begin{corollary}
\label{corollary1}
Suppose that $G$ is $d_{\star}$-connected. Then, the oracle described in Theorem~\ref{theorem2} can answer queries of the form ``given a set $F$ with $d_{\star}$ vertices, what is the number of the connected components of $G\setminus F$?", in $O(d_{\star}^4\log{n})$ time.
\end{corollary}

Again, notice that the oracle in Corollary~\ref{corollary1} can also answer the corresponding vertex-cut query within the same time bound. We note that the problem of efficiently answering $d_{\star}$-vertex-cut queries for $d_{\star}$-connected graphs was left as an open problem in \cite{DBLP:conf/icalp/PettieY21}, and it was only very recently addressed in \cite{arxivMerav}. (For a comparison between our bounds and that of \cite{arxivMerav}, see Table~\ref{table:bounds4} and the discussion below the table.) Furthermore, in the case where $G$ is $\kappa$-connected, for some $\kappa\leq d_{\star}$, we can still have an improved time bound for answering the queries, but this is a little bit more involved:

\begin{corollary}
\label{corollary:componentsKappa}
Suppose that $G$ is $\kappa$-connected. Then, the oracle described in Theorem~\ref{theorem2} can answer queries of the form ``given a set $F$ of vertices, with $|F|=d\leq d_{\star}$, what is the number of connected components of $G\setminus F$?", in $O(d^4\log{n} + (d-\kappa+1)({d-1\choose\kappa-1}(\kappa-1)+\dots+{d-1\choose d-1}(d-1))\log{n})$ time.
\end{corollary}

For a proof of Corollary~\ref{corollary:componentsKappa}, see Section~\ref{section:numberOfComponents}. To appreciate this time bound, let us suppose, for example, that $G$ is $\kappa$-connected (for some $\kappa<d_{\star}$), and that $F$ is a set of $\kappa+1$ vertices. Then, we can determine the number of connected components of $G\setminus F$ in $O(\kappa^4\log{n})$ time. 

Our technique for answering the queries for the number of connected components can be adapted to the case where $d_{\star}=3$, so that we have the following:

\begin{corollary}
\label{corollary:components3}
The data structure described in Theorem~\ref{theorem:main} can answer queries of the form ``given a set $F$ of at most three vertices, what is the number of connected components of $G\setminus F$?'' in $O(\log n)$ time. 
\end{corollary}

Corollary~\ref{corollary:components3} is almost an immediate consequence of the analysis in Section~\ref{section:numberOfComponents}, and of the fact that our optimal oracle for the case $d_{\star}=3$ determines the connection between the internal components in constant time. However, we have to be a little bit careful due to the assumptions~\ref{assumption1} and \ref{assumption2} that we make in Section~\ref{section:oraclefor3}. We explain the necessary details that establish Corollary~\ref{corollary:components3} in Section~\ref{paragraph:reporting}.

\subsection{Previous work}
\label{section:previouswork}
In Table~\ref{table:bounds1} we see the bounds provided by the best known solutions for general $d_{\star}$, and in Table~\ref{table:bounds2} we see how they compare with our optimal oracle for the case $d_{\star}=3$. (In Table~\ref{table:bounds2} some entries are removed, because they do not provide any advantage in their bounds over the rest when $d_{\star}$ is a fixed constant.)

\begin{table*}[h]
 \centering
 \hspace*{-0.35cm}
\renewcommand{\arraystretch}{1.3}
\begin{tabular}{ |c|c|c|c|c| } 
 \hline
 {} & Preprocessing & Space & Update & Query\\ \hline\hline
 {Duan and Pettie~\cite{DBLP:journals/siamcomp/DuanP20}} & $O(mn\log n)$ & $O(d_{\star}m\log n)$ & $O(d^3\log^3{n})$ & $O(d)$ \\ \hline
 {Long and Saranurak~\cite{DBLP:conf/focs/LongS22}} & $\hat{O}(m)+\tilde{O}(d_{\star}m)$ & $O(m\log^{*}{n})$ & $\hat{O}(d^2)$ & $O(d)$\\ \hline
 {Pilipczuk et al.~\cite{DBLP:conf/icalp/PilipczukSSTV22}} & $O(2^{2^{O(d_{\star})}}mn^2)$ & $O(2^{2^{O(d_{\star})}}m)$ & $-$ & $O(2^{2^{O(d_{\star})}})$\\ \hline
 {Kosinas~\cite{DBLP:conf/esa/Kosinas23}} & $O(d_{\star}m\log n)$ & $O(d_{\star}m\log n)$ & $O(d^4\log n)$ & $O(d)$\\ \hline
 {Long and Wang~\cite{DBLP:conf/icalp/LongW24}} & $\hat{O}(m)+O(d_{\star}m\log^3{n})$ & $O(m\log^3{n})$ & $O(d^2\log^7{n})$ & $O(d)$\\ \hline

\end{tabular}

 \caption{The best known bounds for a deterministic oracle for answering connectivity queries under up to $d_{\star}$ vertex failures. Notice that there are various trade-offs that make every one of those solutions have an advantage over the rest in some respects. The $\widetilde{O}$ symbol hides polylogarithmic factors, and $\widehat{O}$ hides subpolynomial factors that are worse than polylogarithmic. The function $\log^{*}{n}$ that appears in the space usage of the Long and Saranurak oracle is described in their paper as one that ``can be substituted with any slowly growing function". All these oracles, except the one by Pilipczuk et al., support an update phase, in which the set $F$ of failed vertices is processed, so that all connectivity queries under the failures from $F$ can be answered in $O(d)$ time, where $d=|F|$. Thus, in order to answer the first connectivity query, at least $\mathit{Update}$ time must have been expended. \label{table:bounds1}}
\end{table*} 

\begin{table*}[h]
 \centering
 \hspace*{-0.6cm}
\renewcommand{\arraystretch}{1.3}
\begin{tabular}{ |c|c|c|c|c| } 
 \hline
 {} & Preprocessing & Space & Update & Query\\ \hline\hline
 {Long and Saranurak~\cite{DBLP:conf/focs/LongS22}} & $\hat{O}(m)+\tilde{O}(m)$ & $O(m\log^{*}{n})$ & $\hat{O}(1)$ & $O(1)$\\ \hline
 {Pilipczuk et al.~\cite{DBLP:conf/icalp/PilipczukSSTV22}} & $O(mn^2)$ & $O(m)$ & $-$ & $O(1)$\\ \hline
 {Kosinas~\cite{DBLP:conf/esa/Kosinas23}} & $O(m\log n)$ & $O(m\log n)$ & $O(\log n)$ & $O(1)$\\ \hline
 {Block- and SPQR-trees \cite{DBLP:journals/algorithmica/BattistaT96} (for $d_{\star}\in\{1,2\}$)} & $O(m+n)$ & $O(n)$ & $-$ & $O(1)$\\ \hline
 {\textbf{This paper} (for $d_{\star}=3$)}  & $O(m+n)$ & $O(n)$ & $-$ & $O(1)$\\ \hline
  
\end{tabular}

 \caption{The best known bounds when $d_{\star}$ is a fixed constant. We get the first three entries precisely from Table~\ref{table:bounds1}, after removing ``$d_{\star}$" and ``$d$" from the bounds, and then deleting the entries that do not have any advantange (in the asymptotic complexity) over the rest. For the case $d_{\star}=3$ in particular, one can substitute ``$m$" with ``$n$" everywhere, except in the preprocessing time, where ``$m$" must appear as an additive factor. (Because we can consider as part of the preprocessing the sparsification of Nagamochi and Ibaraki~\cite{DBLP:journals/algorithmica/NagamochiI92}, that works in time $O(m+n)$, and produces a graph with $n$ vertices and $O(d_{\star}n)$ edges, which maintains the connectivity relation up to $d_{\star}$ failures.) \label{table:bounds2}}
\end{table*} 

We must note that some authors (e.g., \cite{DBLP:journals/siamcomp/DuanP20, DBLP:conf/focs/LongS22}) claim that a data structure from \cite{DBLP:conf/focs/KanevskyTBC91} implies a near-optimal oracle for $d_{\star}=3$. (Specifically, that the oracle occupies space $O(n)$, can answer queries in $O(1)$ time, and can be constructed in near-linear time.) However, no attempt was ever made (as far as we know) to substantiate this claim, and therefore we cannot take it for granted. The data structure from \cite{DBLP:conf/focs/KanevskyTBC91} was designed in order to solve the problem of answering $4$-connectivity queries. Thus, after a near-linear-time preprocessing, one can use this data structure in order to answer queries of the form \emph{``are $x$ and $y$ $4$-connected?"}. A $4$-connectivity query for $x$ and $y$ asks whether $x$ and $y$ remain connected if \emph{any} set of at most three vertices (and/or edges) is removed. Thus, we first notice that this solves a different problem than our own, and, in general, the two problems are relatively independent. This is because, if we know that $x$ and $y$ are \emph{not} $4$-connected, then this in itself tells us nothing about whether a particular set of at most three vertices disconnects $x$ and $y$ upon removal. So the question arises, how does the data structure from \cite{DBLP:conf/focs/KanevskyTBC91} determine the $4$-connectivity? And the answer to that, according to \cite{DBLP:conf/focs/KanevskyTBC91}, is that the data structure stores a collection of \emph{some} $3$-vertex cuts, which are sufficient in order to determine that relation. However, in order to solve our problem with such a data structure, one would need to have a representation of \emph{all} $3$-vertex cuts of the graph, and of how they separate the vertex set. Can the data structure from \cite{DBLP:conf/focs/KanevskyTBC91} support this functionality? This is not discussed in \cite{DBLP:conf/focs/KanevskyTBC91}, and one should provide an explanation as to how this data structure can provide an oracle for answering connectivity queries under any set of at most three vertex failures. 
In any case, even if \cite{DBLP:conf/focs/KanevskyTBC91} can be made to do that, the construction time is not linear, and so our own oracle remains the first known optimal solution for $d_{\star}=3$.
 

\begin{table*}[h]
 \centering
 \hspace*{-1.2cm}
\renewcommand{\arraystretch}{1.3}
\begin{tabular}{ |c|c|c|c|c| } 
 \hline
 {} & Preprocessing & Space & Query & Graph Type\\ \hline\hline
 {Jiang et al.~\cite{arxivMerav}} & $O(m)+\widetilde{O}(n^{1+\delta})$ & $\widetilde{O}(n)$ & $\widetilde{O}(2^d)$ & all\\ \hline

 {$>>$} & $O(m)+\widetilde{O}(d_{\star}^2n)+\widetilde{O}((d_{\star}n)^{1+\delta})$ & $\widetilde{O}(d_{\star}n)$ & $\widetilde{O}(d_{\star}^2)$ & $d_{\star}$-conn\\ \hline
  {\textbf{This paper}} & $O(d_{\star}m\log{n})$ & $O(d_{\star}m\log{n})$ & $O((d^4+2^dd^2)\log{n})$ & all\\ \hline
  {$>>$} & $O(d_{\star}m\log{n})$ & $O(d_{\star}m\log{n})$ & $O(d_{\star}^4\log{n})$ & $d_{\star}$-conn\\ \hline
  
\end{tabular}

 \caption{Best known bounds for vertex-cut oracles. These take as input $G$ and $d_{\star}$, and can answer, for every vertex set $F$ with $|F|=d\leq d_{\star}$, whether $G\setminus F$ is connected. Our own oracles have the stronger property that they report the number of the connected components of $G\setminus F$. In the first line, the oracle of \cite{arxivMerav} demands that $d_{\star}=O(\log{n})$ (and so $d_{\star}$ is absorbed by the $\widetilde{O}$ expression), whereas our result in the third line works for any $d_{\star}$. Being able to lift the restriction $d_{\star}=O(\log{n})$ gives an advantage to us, as we discuss in the main text. 
The ``$+\delta$'' that appears in the exponent in the preprocessing time of the oracle of \cite{arxivMerav} can be any fixed $\delta>0$, but the choice of $\delta$ influences the $\mathtt{polylog{n}}$ factors hidden behind the expressions for the space and the query time bounds. $\delta$ can also be replaced with $o(1)$, but this is going to introduce an $n^{o(1)}$ factor in both the space and the query time. Notice the better time bounds for the queries in both oracles when we are given as information that the graph is $d_{\star}$-connected. Also, notice that the time bounds for the queries depend on $|F|$, except when the graph is $d_{\star}$-connected, in which case the queries are non-trivial only for $F$ with $|F|=d_{\star}$. Using the sparsifier of Nagamochi and Ibaraki~\cite{DBLP:journals/algorithmica/NagamochiI92}, we can replace the ``$m$" in our bounds with $\min\{m,d_{\star}n\}$.  \label{table:bounds4}}
\end{table*} 

\subsection{Concurrent work} 
The problem of designing an efficient oracle for vertex-cut queries was only very recently addressed in \cite{arxivMerav}, although the question was posed (and conditional lower bounds were given) in \cite{DBLP:conf/focs/LongS22}, and a particular interest for the case of $\kappa$-connected graphs was expressed in \cite{DBLP:conf/icalp/PettieY21}. In Table~\ref{table:bounds4} we see the bounds provided by \cite{arxivMerav}, and by our simple extension of the framework of \cite{DBLP:conf/esa/Kosinas23}. We note that our own oracle answers a stronger type of queries: it reports the number of connected components after removing a set of vertices. (We show that, in the framework of \cite{DBLP:conf/esa/Kosinas23}, these two queries are essentially equivalent.)

In order to get a better appreciation of our result, it is worth comparing it in some detail with that of \cite{arxivMerav}\footnote{Although the paper of \cite{arxivMerav} does not appear to have been peer-reviewed, the result is highly credible.}. First, as noted above, our oracle reports the number of connected components after removing a set of vertices $F$ (and not just whether $F$ is a vertex cut). As far as we know, we are the first to provide a non-trivial (and very efficient) oracle for this problem. Second, our bounds involve only one ``$\log{n}$'' factor, whereas the bounds from \cite{arxivMerav} hide many more. Third, although \cite{arxivMerav} provide a better query time for the case where the graph is $d_{\star}$-connected, we provide a better query time throughout the entire regime where the graph is $\kappa$-connected, for any $\kappa\leq d_{\star}$. (For the precise time bound, see Corollary~\ref{corollary:componentsKappa}.) This is a much stronger result than the one asked for by \cite{DBLP:conf/icalp/PettieY21}, and it has the interesting property that it does not rely on a computation or compact representation of vertex cuts, but it works directly on a DFS tree. Finally, in both our oracle and that of \cite{arxivMerav}, there appears a $2^d$ factor in the query time, where $d=|F|$, which seems to make both oracles useless when $d=\Omega(\log{n})$. However, this is a worst-case bound for our own oracle. A single glance into the mechanics of our oracle (provided in Section~\ref{section:numberOfComponents}) will reveal that the exponential time occurs only if there is a large subset of $F$ that consists of vertices that are pairwise related as ancestor and descendant. Thus, if it happens that there are not large subsets of $F$ of related vertices (w.r.t. the DFS tree), then our oracle is potentially much better than applying brute force (e.g., BFS), and thus it makes sense to initialize it for values of $d_{\star}$ which are larger than $\Omega(\log{n})$.

\subsection{Related work}

\paragraph{Connectivity under edge failures.}
There is a similar problem for connectivity oracles under \emph{edge failures}.\footnote{In fact, this problem can be easily reduced to oracles for vertex failures, but such a reduction provides suboptimal solutions; by dealing directly with edge failures, one can provide more efficient oracles.} This problem was considered first by Patrascu and Thorup~\cite{DBLP:conf/focs/PatrascuT07}, where they provided an oracle with $O(m)$ space and near-optimal query time, but very high preprocessing time. Kosinas~\cite{DBLP:conf/soda/Kosinas24} provided an optimal solution for up to four edge failures, using a DFS-based approach. Duan and Pettie~\cite{DBLP:journals/siamcomp/DuanP20} provide a very efficient oracle for any number of edge failures, with near-linear preprocessing time (in expectation). For more references on that problem, see \cite{DBLP:journals/siamcomp/DuanP20}. 

\paragraph{Connectivity under mixed deletions/insertions.}
What happens if, instead of vertex failures only, we allow for intermixed failures and activations of vertices? This model was considered first by Henzinger and Neumann~\cite{DBLP:conf/esa/HenzingerN16}, where they provided an efficient reduction to any oracle for vertex failures. However, some data structures for vertex failures have the potential to be extended to the intermixed activations/deactivations version, so that they provide better bounds than a black-box reduction. This approach was taken by Long and Wang~\cite{DBLP:conf/icalp/LongW24} (who extended the data structure of Long and Saranurak~\cite{DBLP:conf/focs/LongS22}), and by Bingbing et al.~\cite{DBLP:conf/esa/HuK024} (who extended the data structure of Kosinas~\cite{DBLP:conf/esa/Kosinas23}). 

\paragraph{Dynamic subgraph connectivity.}
If the activations/deactivations of vertices are not restricted to batches of bounded size (which afterwards rebound back to the original state of the graph), but they are totally unrestricted, then we are dealing with the so-called \emph{dynamic subgraph connectivity} problem. This model was introduced first by Frigioni and Italiano~\cite{DBLP:journals/algorithmica/FrigioniI00} in the context of planar graph, and was later considered in various works~\cite{DBLP:conf/icalp/Duan10, DBLP:journals/siamcomp/ChanPR11, DBLP:conf/wads/DuanZ17}, for general graphs. The unrestricted nature of the updates forbids the existence of adequately efficient solutions, according to some conditional lower bounds \cite{DBLP:conf/stoc/HenzingerKNS15, DBLP:conf/stoc/JinX22}.

\paragraph{Labeling schemes.} There is a very interesting line of work \cite{DBLP:conf/podc/DoryP21, DBLP:conf/wdag/ParterP22a, DBLP:conf/podc/IzumiEWM23, DBLP:conf/stoc/ParterPP24, DBLP:conf/soda/LongPS25} that provides connectivity oracles for vertex and edge failures which can answer the queries given access only to some labels that have been assigned to the elements of interest (i.e., to the query vertices, and to the vertices or edges that have failed). Here the goal is to optimize the size of the labels that are attached to the elements of the graph, and the time to compute the labels (in a preprocessing phase). Very recently, such labeling schemes were presented for the first time for vertex cut queries, by Jiang, Parter, and Petruschka~\cite{arxivMerav}.

\paragraph{$\kappa$-connectivity oracles.}
There is a related problem of constructing an oracle that can efficiently answer $\kappa$-connectivity queries. I.e., given two vertices $x$ and $y$, the question is ``are $x$ and $y$ $\kappa$-connected"? The query can be extended as: ``if $x$ and $y$ are not $\kappa$-connected, what is their connectivity"? or ``what is a $\lambda$-vertex cut, with $\lambda<\kappa$, that separates $x$ and $y$"? For $\kappa\leq 4$, optimal oracles exist, w.r.t. the space usage and query time \cite{DBLP:journals/algorithmica/BattistaT96,DBLP:conf/focs/KanevskyTBC91}. A data structure for the case of $\kappa$-connected graphs was first given by Cohen et al.~\cite{DBLP:conf/stoc/CohenBKT93}, and a faster construction, informed by an elaborate analysis of the structure of minimum vertex cuts, was given by Pettie and Yin~\cite{DBLP:conf/icalp/PettieY21}. Various other solutions are known when there is no restriction on the connectivity of the graph (e.g., \cite{DBLP:conf/esa/Nutov22,DBLP:conf/stoc/PettieSY22,DBLP:journals/corr/abs-2411-02658}).

\paragraph{Directed connectivity.}
There are similar problems for reachability under vertex failures in directed graphs. The most efficient oracle known, for general reachability, is given by van den Brand and Saranurak~\cite{DBLP:conf/focs/BrandS19} (which provides the answer with high probability). For the case of single source reachability, dominator trees~\cite{DBLP:journals/siamcomp/BuchsbaumGKRTW08} and an oracle by Choudhary~\cite{DBLP:conf/icalp/Choudhary16}, provide optimal space and query time for one and two vertex failures, respectively. (However, the \emph{construction time} of the oracle in \cite{DBLP:conf/icalp/Choudhary16} is probably not optimal.) For the case of strong connectivity, an optimal solution is given by Georgiadis et al.~\cite{DBLP:journals/siamcomp/GeorgiadisIP20} for one vertex failure, which can also report the number of strongly connected components upon removal of a vertex, in constant time. For strong connectivity under multiple failures, an efficient oracle was given by Baswana et al.~\cite{DBLP:journals/algorithmica/BaswanaCR19}. The complexity landscape for the case of more vertex failures is still under exploration, and some lower bounds were given by Chakraborty and Choudhary~\cite{DBLP:conf/icalp/ChakrabortyC20}.

\subsection{Organization}
Throughout this paper we assume that we work on a connected undirected graph $G$, with $n$ vertices and $m$ edges.
Since our approach is to build on the DFS-based framework of Kosinas~\cite{DBLP:conf/esa/Kosinas23}, it is natural to spend some time in order to review the basics of that framework. Thus, we provide a brief (but self-contained) overview of this DFS-based approach in Section~\ref{section:theframework}, and we show how it can be easily extended in order to answer queries about the number of connected components upon removal of a set of vertices (in Section~\ref{section:numberOfComponents}). In Section~\ref{section:someconcepts} we introduce some of the DFS-based parameters that we will need for the oracle that handles up to three vertex failures. Then, as a warm-up, we show how we can easily handle the case of two vertex failures (in Section~\ref{section:f=2}), by relying on some of those DFS-based parameters (and by completely avoiding the use of block- and SPQR-trees). In Section~\ref{section:f=3}, we sketch the case analysis and the techniques that provide the oracle for $d_{\star}=3$, and the full details are provided in Sections~\ref{section:vwnotrelated} and \ref{section:vwrelated}. In Sections \ref{section:DFSConcepts} and \ref{section:efficient_computation} we provide definitions and algorithms for the efficient computation of the DFS-based parameters that we will need. We believe that these can find applications to various other (vertex-connectivity) problems.

\section{The DFS-based framework for connectivity oracles}
\label{section:theframework}

\subsection{Basic definitions and notation}
\label{section:basicDFS}
Let $T$ be a DFS tree of $G$, with start vertex $r$ \cite{DBLP:journals/siamcomp/Tarjan72}. We use $p(v)$ to denote the parent of every vertex $v\neq r$ in $T$ ($v$ is a child of $p(v)$). For any two vertices $u,v$, we let $T[u,v]$ denote the simple tree path from $u$ to $v$ on $T$, or the set of vertices on that path.  For any two vertices $u$ and $v$, if the tree path $T[r,u]$ uses $v$, then we say that $v$ is an ancestor of $u$ (equivalently, $u$ is a descendant of $v$). In particular, a vertex is considered to be an ancestor (and also a descendant) of itself. It is very useful to identify the vertices with their order of visit during the DFS, starting with $r\leftarrow 1$. Thus, if $v$ is an ancestor of $u$, we have $v<u$. For any vertex $v$, we let $T(v)$ denote the subtree rooted at $v$, and we let $\mathit{ND}(v)$ denote the number of descendants of $v$ (i.e., $\mathit{ND}(v)=|T(v)|$). Thus, we have $T(v)=\{v,v+1,\dots,v+\mathit{ND}(v)-1\}$, and therefore we can check the ancestry relation in constant time. If $c$ is a child of a vertex $v$, we define \emph{the next sibling} of $c$ as the lowest child of $v$ that is greater than $c$ (w.r.t. the DFS numbering).

A DFS tree $T$ has the following very convenient property that makes it particularly suitable for solving various connectivity problems: the endpoints of every non-tree edge of $G$ are related as ancestor and descendant on $T$ \cite{DBLP:journals/siamcomp/Tarjan72}, and so we call those edges \emph{back-edges}. Our whole approach is basically an exploitation of this property, which does not hold in general rooted spanning trees of $G$ (unless they are derived from a DFS traversal, and only then \cite{DBLP:journals/siamcomp/Tarjan72}). 
Whenever $(x,y)$ denotes a back-edge, we make the convention that $x>y$. We call $x$ and $y$ the higher and the lower, respectively, endpoint of $(x,y)$. 

Given a DFS tree $T$ of $G$ (together with a DFS numbering), we get another DFS tree $T'$ of $G$ (providing a different DFS numbering) if we rearrange the children list of every vertex. If $T'$ is derived from $T$ in this way, we call it \emph{a permutation of $T$}. A recent work by Kosinas~\cite{DBLP:conf/esa/Kosinas23} has demonstrated the usefulness in considering a collection of permutations of a base DFS tree, in order to provide an efficient (and efficiently constructed) oracle for connectivity queries under vertex failures. Notice that, although the DFS numbering may change, the ancestry relation is an invariant across all permutations of a base DFS tree. The usefulness in considering permutations of a base DFS tree will become apparent in Section~\ref{section:f=2}. (In the sequel, we will actually need only two permutations of the base DFS tree $T$, which we denote as $T_\mathit{lowInc}$ and $T_\mathit{highDec}$, and we define in Section~\ref{section:someconcepts}.)

For a vertex $v$, we will use $\mathit{depth}(v)$ to denote the depth of $v$ on the DFS tree. This is defined recursively as $\mathit{depth}(r):=0$, and $\mathit{depth}(v):=\mathit{depth}(p(v))+1$, for any vertex $v\neq r$. (Notice that the $\mathit{depth}$ of vertices is an invariant across all permutations of a base DFS tree.) We will use the $\mathit{depth}$ of vertices in order to initialize a level-ancestor oracle~\cite{DBLP:journals/tcs/BenderF04} on the DFS tree, which answers queries of the following form: $\mathtt{QueryLA}(v,d)\equiv$ ``return the ancestor of $v$ which has depth $d$". The oracle in \cite{DBLP:journals/tcs/BenderF04} can be initialized in $O(n)$ time, and can answer every level-ancestor query in constant time. We need those queries in order to be able to find in constant time the child of $v$ in the direction of $u$, for any two vertices $v$ and $u$ such that $v$ is a proper ancestor of $u$. This child is given precisely by the answer to $\mathtt{QueryLA}(u,\mathit{level}(v)+1)$. 

\subsection{Internal components and hanging subtrees} 
\label{section:internalAndHanging}
Given a set $F=\{f_1,\dots,f_k\}$ of vertices of $G$ that have failed to work, we have that $T\setminus F$ is possibly split into several connected components. Since every connected component $C$ of $T\setminus F$ is a  subtree of $T$, it makes sense to speak of the root $r_C$ of $C$. Now, following the terminology from \cite{DBLP:conf/esa/Kosinas23}, we distinguish two types of connected components of $T\setminus F$: internal components, and hanging subtrees. (See Figure~\ref{figure:internal}.) If $C$ is a connected component of $T\setminus F$ such that $r_C$ is an ancestor of at least one vertex from $F$, then $C$ is called an \emph{internal component}. Otherwise, $C$ is called \emph{a hanging subtree}. Notice that the root of a hanging subtree $H$ is a child of a failed vertex $f$, and so we say that $H$ is a hanging subtree of $f$. The important observation from \cite{DBLP:conf/esa/Kosinas23} is that, although there may exist $\Omega(n)$ hanging subtrees (irrespective of $k$), the number of internal components is at most $k$. Thus, in the case of three vertex-failures, we have at most three internal components.

\begin{figure}[h!]\centering
\includegraphics[trim={0cm 18.5cm 0 0.5cm}, clip=true, width=0.97\linewidth]{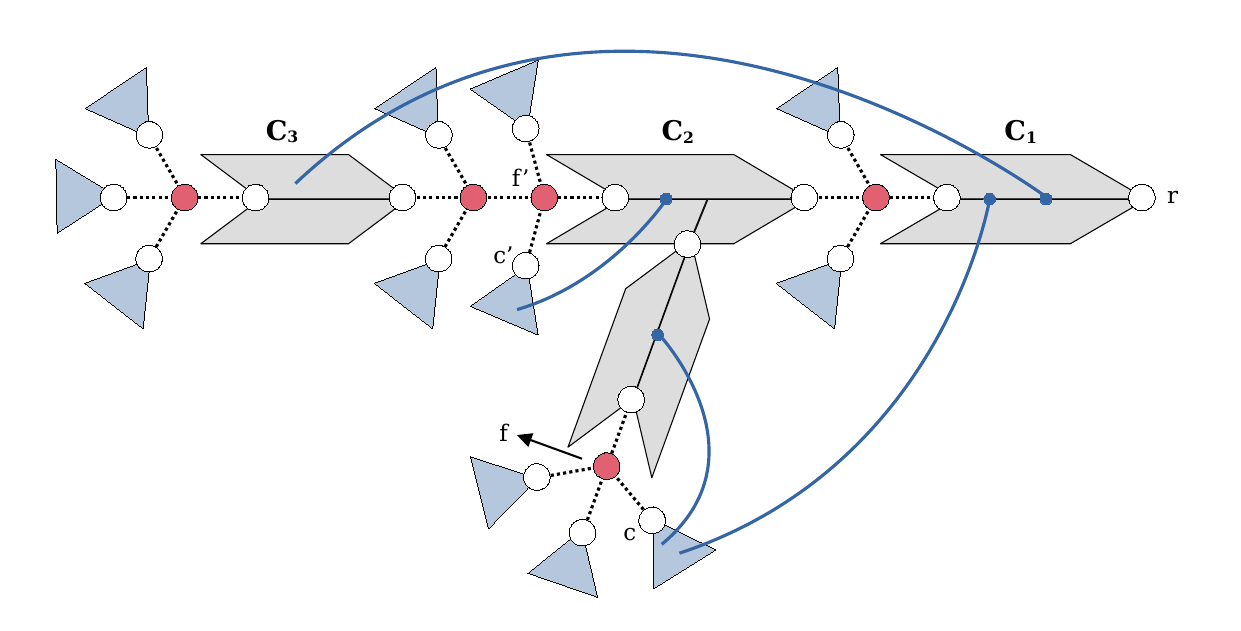}
\caption{\small{An example of internal components (in gray) and hanging subtrees (in blue) of a DFS tree rooted at $r$, given a set of vertices that have failed to work. (The failed vertices are coloured red.) We have that $T(c)$ is a hanging subtree of the failed vertex $f$, and $T(c')$ is a hanging subtree of $f'$. Notice that there two kinds of back-edges that remain in the graph: those that connect internal components, and those that connect hanging subtrees and internal components. (There is no direct connection between two hanging subtrees with a back-edge.) We have that $C_3$ remains connected with $C_1$ directly with a back-edge, and $C_1$ remains connected with $C_2$ through the mediation of the hanging subtree $T(c)$.}}\label{figure:internal}
\end{figure}  

The main challenge in determining the connectivity relation in $G\setminus F$ is to determine the connection between the internal components of $T\setminus F$. (For more on that, see Section~\ref{section:internal}.) This is because, for a hanging subtree with root $c$, it is sufficient to know at least $k$ back-edges of the form $(x,y)$ with $x\in T(c)$ and $y<p(c)$, with pairwise distinct lower endpoints. We call such a set of back-edges \emph{the surviving back-edges of $T(c)$}. Then, we have that $T(c)$ is connected with an internal component of $T\setminus F$ in $G\setminus F$, if and only if there is at least one surviving back-edge of $T(c)$ whose lower endpoint is not in $F$. Furthermore, in order for $T(c)$ to be connected with another hanging subtree $T(c')$ in $G\setminus F$, it is necessary that both of them remain connected with an internal component. Thus, by establishing the connection in $G\setminus F$ between the internal components of $T\setminus F$, and by having stored, for every vertex $c$, a set of (candidate) surviving back-edges for $T(c)$, we can now easily determine in $O(k)$ time whether two vertices remain connected in $G\setminus F$. (For the details, we refer to \cite{DBLP:conf/esa/Kosinas23}.)

\subsection{Computing the number of connected components of $G\setminus F$}
\label{section:numberOfComponents}
The distinction into internal components and hanging subtrees of $T\setminus F$ is the key to determine the number of connected components of $G\setminus F$. As we will discuss in Section~\ref{section:internal}, there is a graph $\mathcal{R}$, called the connectivity graph, whose nodes correspond to the internal components of $T\setminus F$, and whose connected components capture the connectivity relation between the internal components of $T\setminus F$ in $G\setminus F$. More precisely, two internal components are connected in $G\setminus F$ if and only if their corresponding nodes are connected in $\mathcal{R}$. Now, if we consider a hanging subtree of $T\setminus F$, there are two possibilities: either it is isolated in $G\setminus F$ (i.e., a connected component of $G\setminus F$), or it is connected with an internal component through a surviving back-edge. Thus, the number of connected components of $G\setminus F$ is: $\#$[connected components of $\mathcal{R}$]+$\#$[isolated hanging subtrees].

The data structure from \cite{DBLP:conf/esa/Kosinas23} can build the connectivity graph $\mathcal{R}$ in $O(d^4\log{n})$ time, where $d=|F|$. Thus, it remains to compute the number of isolated hanging subtrees of $T\setminus F$ in $G\setminus F$. To do this, we process the failed vertices from $F$ (in any order), and for each of them we determine the number of its children that induce isolated hanging subtrees. For every vertex $v\in F$, and every child $c$ of $v$, we have that $T(c)$ induces an isolated hanging subtree if and only if: $(1)$ no vertex from $F$ is a descendant of $c$, and $(2)$ every surviving back-edge $(x,y)$ of $T(c)$ has $y\in F$. (See Figure~\ref{figure:isolated}.) Notice that $(1)$ is necessary in order to ensure that $T(c)$ is indeed a hanging subtree, and then $(2)$ ensures that $T(c)$ is isolated.

\begin{figure}[h!]\centering
\includegraphics[trim={0cm 22cm 0 0.5cm}, clip=true, width=1\linewidth]{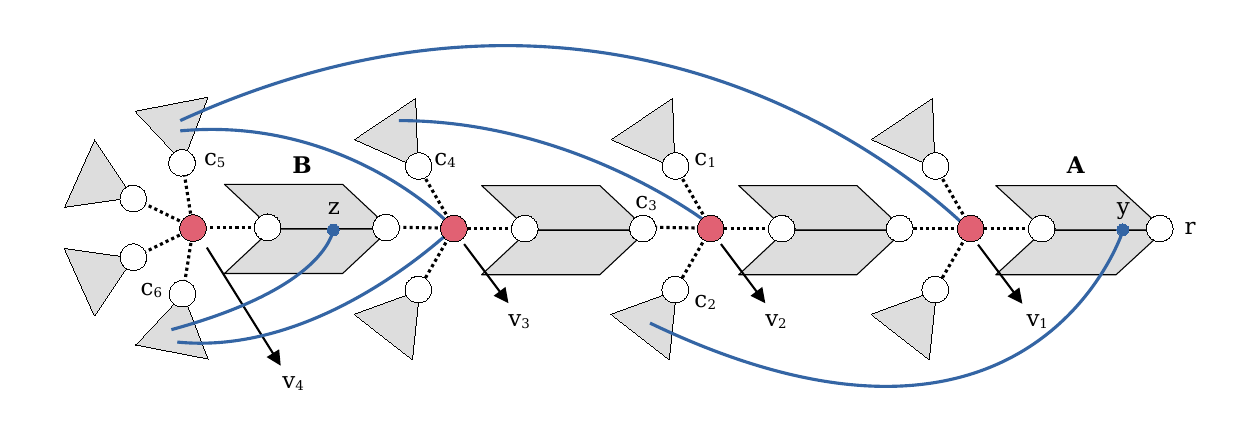}
\caption{\small{The vertices $\{v_1,v_2,v_3,v_4\}$ have failed to work. $T(c_1)$ is an isolated hanging subtree of $v_2$, because it has no surviving back-edges. On the other hand, $T(c_2)$ is not an isolated hanging subtree of $v_2$, because it has a surviving back-edge whose lower endpoint is not a failed vertex. $T(c_3)$ is not a hanging subtree of $v_2$, because there are descendants of $c_3$ that are failed vertices. $T(c_4)$ is an isolated hanging subtree of $v_3$, because the lower endpoint of its only surviving back-edge is $v_2$. $T(c_5)$ is an isolated hanging subtree of $v_4$, because the lower endpoints of its two surviving back-edges are $v_1$ and $v_3$. $T(c_6)$ is a non-isolated hanging subtree of $v_4$, because it has a surviving back-edge whose lower endpoint is not a failed vertex. Following the notation in the main text, and assuming that we store up to four surviving back-edges, we have the sorted lists $\mathcal{L}(c_1)=\{\bot,\bot,\bot,\bot\}$, $\mathcal{L}(c_2)=\{y,\bot,\bot,\bot\}$, $\mathcal{L}(c_3)=\{v_1,\bot,\bot,\bot\}$, $\mathcal{L}(c_4)=\{v_2,\bot,\bot,\bot,\}$, $\mathcal{L}(c_5)=\{v_1,v_3,\bot,\bot\}$, and $\mathcal{L}(c_6)=\{v_3,z,\bot,\bot\}$.}}\label{figure:isolated}
\end{figure}

To facilitate the search for the number of children of a vertex $v\in F$ that induce isolated hanging subtrees, we need to enrich the data structure from \cite{DBLP:conf/esa/Kosinas23} as follows. Recall that, for every vertex $c$, the data structure stores $d_{\star}$ surviving back-edges for $T(c)$. What we really care about are the lower endpoints of those back-edges. Thus, for every vertex $c$, we store the lower endpoints of the $d_{\star}$ surviving back-edges of $T(c)$ in a list $\mathcal{L}(c)$, sorted in increasing order (w.r.t. the DFS numbering). (Notice that some of the elements in $\mathcal{L}(c)$ may be $\bot$, and we make the convention that $\bot$ is an element greater than every vertex; see Figure~\ref{figure:isolated} for an illustration.) Then, we sort the children list of every vertex in increasing order w.r.t. the $\mathcal{L}$ lists of its children, where the comparison between the $\mathcal{L}$ lists is done lexicographically. We can use bucket-sort in order to have all children lists sorted thus in $O(d_{\star}n)$ time in total. So, for a vertex $v$, let $C(v)$ be its sorted children list. (Notice that the asymptotic complexity of the construction time and the space usage of the data structure stays the same.)

Now, given a vertex $v\in F$, we will show how to count the number of children of $v$ that induce isolated hanging subtrees. (Then, the total sum of the number of those children, for all vertices $v\in F$, gives the number of the isolated hanging subtrees of $T\setminus F$ in $G\setminus F$.) Let $v_1,\dots,v_k$ be the proper ancestors of $v$ in $F$, sorted in increasing order. Now let $c$ be a child of $v$ that induces a hanging subtree. Then, $T(c)$ is isolated if and only if $\mathcal{L}(c)\subseteq\{v_1,\dots,v_k\}$ (where we ignore all ``$\bot$" entries in $\mathcal{L}(c)$). Thus, for every subset $F'$ of $\{v_1,\dots,v_k\}$ (including the empty set), we search for (the endpoints of) the segment $\mathcal{S}$ of $C(v)$ that consists of all children $c$ of $v$ with $\mathcal{L}(c)=F'$. Notice that $\mathcal{S}$ is indeed a segment (i.e., consists of consecutive entries) of $C(v)$, and can be found in $O(|F'|\log n)$ time using binary search on $C(v)$. (This is because the comparisons are done lexicographically, but we only need to access the first $|F'|$ entries of the $\mathcal{L}$ lists of the vertices in $C(v)$.) Since $|F'|\leq k<|F|=d$, and this is done for every subset of $\{v_1,\dots,v_k\}$, we have that all these segments $\mathcal{S}$ can be found in $O(2^dd\log n)$ time in total for $v$, and thus in $O(2^dd^2\log n)$ time in total for all vertices in $F$. (This establishes the time bound for the query for the number of connected components of $G\setminus F$.) It is important to notice that, two \emph{distinct} subsets $F_1$ and $F_2$ of $\{v_1,\dots,v_k\}$, provide \emph{disjoint} segments $\mathcal{S}_1$ and $\mathcal{S}_2$.

Now we are almost done. Since we know the endpoints of $\mathcal{S}$, we also know its size, but we have to subtract from $|\mathcal{S}|$ the number of children of $v$ in $\mathcal{S}$ that do not induce hanging subtrees. To do that, we process all vertices $v'$ from $F$, and for every $v'$ that is a proper descendant of $v$, we ask for the child $c$ of $v$ that is an ancestor of $v'$. (We note that $c$ can be found in constant time using a level-ancestor query, as explained in Section~\ref{section:basicDFS}; we assume that we have initialized an oracle for such queries on the DFS tree that is given by the children lists $C(v)$.) Then, we can check in constant time whether $c\in \mathcal{S}$. If that is the case, then we mark $c$ as a child that does not induce a hanging subtree. After the processing of all those $v'$, we can subtract from $|\mathcal{S}|$ the number of vertices in $\mathcal{S}$ that do not induce hanging subtrees of $v$ (which are precisely all the vertices that we have marked). 
This discussion establishes Theorem~\ref{theorem2}. 


\paragraph{The case where $G$ is $d_{\star}$-connected.}
The computation of the number of connected components of $G\setminus F$ can be sped up if we know that $G$ is $\kappa$-connected, for some $\kappa\leq d_{\star}$. (Of course, if $\kappa>d_{\star}$, then $G\setminus F$ is connected, for any set of vertices $F$ with $|F|\leq d_{\star}$.) As a warm-up, let us consider the case where $\kappa=d_{\star}$. Then, $G\setminus F$ may be disconnected only if $|F|=d_{\star}$. (In every other case, i.e., when $|F|<d_{\star}$, we can immediately report that $G\setminus F$ consists of a single connected component.)

So let $F$ be a vertex set with $|F|=d_{\star}$. Notice that every vertex $v$ in $F$, except possibly one, has the property that there are less than $d_{\star}$ vertices in $F$ that are ancestors of $v$. Then, for such a vertex $v$, we have that no child $c$ of $v$ induces an isolated hanging subtree. This is precisely due to the fact that $G$ is $d_{\star}$-connected: otherwise, i.e., if there was a child $c$ of $v$ such that $T(c)$ induces an isolated hanging subtree, then the set of the ancestors of $v$ in $F$ would constitute a set with less than $d_{\star}$ vertices whose removal would disconnect $T(c)$ from the rest of the graph. Thus, if no vertex in $F$ has $d_{\star}$ ancestors from $F$, then $G\setminus F$ is connected. Otherwise, let $\tilde{v}$ be the unique vertex in $F$ with the property that all vertices from $F$ are ancestors of it. Then, only $\tilde{v}$ may have isolated hanging subtrees. But it is easy to find them with a single binary search in $C(\tilde{v})$: it is sufficient to search for those children $c$ of $\tilde{v}$ in $C(\tilde{v})$ with $\mathcal{L}(c)=F\setminus\{\tilde{v}\}$. (Because, since $G$ is $d_{\star}$-connected, every child $c$ of $\tilde{v}$ has at least $d_{\star}-1$ non-$\bot$ entries in $\mathcal{L}(c)$.) This binary search (for the endpoints of the corresponding segment of $C(\tilde{v})$) takes time $O(d_{\star}\log{n})$ (because the comparisons in this binary search are done lexicographically on the $\mathcal{L}$ lists, which have $d_{\star}$ entries). Thus, we can find the number of connected components of $G\setminus F$ in $O(d_{\star}^4\log{n})$ time (i.e., the query time is dominated by the computation of the connectivity graph $\mathcal{R}$). This establishes Corollary~\ref{corollary1}.

\paragraph{The case where $G$ is $\kappa$-connected (for $\kappa\leq d_{\star}$).}
For the general case where $G$ is $\kappa$-connected, with $\kappa\leq d_{\star}$, the query for the number of the connected components of $G\setminus F$ is non-trivial only if $\kappa\leq d\leq d_{\star}$, where $d=|F|$. Here we are guided by the same intuition as in the previous paragraph (for the case where $G$ is $d_{\star}$-connected).

So let $F$ be a set of failed vertices with $|F|=d$ and $\kappa\leq d\leq d_{\star}$, and let $F'$ be the subset of $F$ that consists of those vertices that have at least $\kappa$ ancestors from $F$. Then, it should be clear that only vertices from $F'$ can provide isolated hanging subtrees. To see this, consider a vertex $v\in F$ that has less than $\kappa$ ancestors from $F$ (where $v$ is included as an ancestor of itself from $F$). Then, every child $c$ of $v$ has less than $\kappa$ ancestors from $F$, and therefore, by removing all of them, we have that $T(c)$ is still connected with the remaining graph (since the graph is $\kappa$-connected). This implies the existence of a surviving back-edge for $T(c)$ -- which still exists even if we remove all of $F$ from $G$, and demonstrates that $T(c)$ remains connected with at least one internal component of $T\setminus F$ in $G\setminus F$.

Now, for every $v\in F'$, we let $\mathit{Anc}(v)$ denote the set of its proper ancestors from $F$, and let $N(v)=|\mathit{Anc}(v)|$. Then, for every $v\in F'$, and for every subset $A$ of $\mathit{Anc}(v)$ with at least $\kappa-1$ elements, we must perform a binary search, as explained in Section~\ref{section:numberOfComponents}, for the segment of $C(v)$ that consists of those vertices $c$ with $\mathcal{L}(c)=A$. Thus, we perform ${N(v)\choose \kappa-1}+\dots+{N(v)\choose N(v)}$ binary searches, which, in total, take time $O(({N(v)\choose \kappa-1}(\kappa-1)+\dots+{N(v)\choose N(v)}N(v))\log{n})$.

Thus, we can provide the following upper bound for the total time for finding the number of isolated hanging subtrees. Notice that $F'$ consists of at most $d-\kappa+1$ vertices, and that $N(v)\leq d-1$ for every $v\in F'$. Thus, the most pessimistic upper bound that we can get is $O((d-\kappa+1)({d-1\choose\kappa-1}(\kappa-1)+\dots+{d-1\choose d-1}(d-1))\log{n})$. Notice that this bound dominates the time that it takes to find, for every $v\in F'$, for every subset $A$ of $\mathit{Anc}(v)$, the elements in the segment of $C(v)$, that corresponds to $A$, that are ancestors of vertices from $F$ (since $v$ has at most $d-\kappa+1$ descendants from $F$).

To appreciate this time bound, let us suppose, for example, that $G$ is $\kappa$-connected, and that $F$ is a set of $\kappa+1$ vertices. Then, we can determine the number of isolated hanging subtrees of $T\setminus F$ in $O(\kappa^2\log n)$ time, and thus the number of connected components of $G\setminus F$ in time $O(\kappa^4\log{n})$. (Thus, the query time is still dominated by the time to compute the connectivity graph $\mathcal{R}$.)


\subsection{Determining the connection between the internal components}
\label{section:internal}
Upon receiving the set $F$ of failed vertices, the main problem is to determine the connectivity relation in $G\setminus F$ between the internal components of $T\setminus F$. Here it helps to distinguish two basic types of connection between internal components: either $(1)$ directly with a back-edge, or $(2)$ through the mediation of a hanging subtree. These two connections can only exist for two internal components that are related as ancestor and descendant. So let $C$ and $C'$ be two distinct internal components such that $r_C$ is a descendant of $r_{C'}$. Now, if $(1)$ there is a back-edge $(x,y)$ such that $x\in C$ and $y\in C'$, then we say that $C$ and $C'$ are connected \emph{directly with a back-edge}. And if $(2)$ there is a hanging subtree $H$ for which there exist a back-edge $(x,y)$ with $x\in H$ and $y\in C$, and a back-edge $(x',y')$ with $x'\in H$ and $y'\in C'$, then we say that $C$ and $C'$ are connected \emph{through the mediation of $H$}. (Notice that, in this case, we have that $r_H$ must be a descendant of $r_C$.) It is not difficult to verify the following:

\begin{lemma}[Implicit in \cite{DBLP:conf/esa/Kosinas23}]
\label{lemma:basic}
Let $F$ be a set of failed vertices, and let $C$ and $C'$ be two internal components of $T\setminus F$. Then $C$ and $C'$ are connected in $G\setminus F$ if and only if:  there is a sequence of internal components $C_1,\dots,C_t$, with $C_1=C$ and $C_t=C'$, such that $C_i$ and $C_{i+1}$ are connected either directly with a back-edge or through the mediation of a hanging subtree, for every $i\in\{1,\dots,t-1\}$.
\end{lemma} 

Lemma~\ref{lemma:basic} implies that, in order to determine the connectivity relation between the internal components, it is sufficient to check, for every pair of internal components $C$ and $C'$ (that are related as ancestor and descendant), whether they are connected directly with a back-edge, or through the mediation of a hanging subtree. If this is the case, then we add an artificial edge between (some representative nodes for) $C$ and $C'$. Thus, we have built a graph $\mathcal{R}$, which in \cite{DBLP:conf/esa/Kosinas23} was called \emph{the connectivity graph}, whose nodes are (representatives of) the internal components, and whose edges represent the existence of a connection either with a back-edge or through a hanging subtree. Then, it is sufficient to compute the connected components of $\mathcal{R}$, in order to get the connectivity relation between all internal components.

So the problem we have to solve is the following: given two internal components $C$ and $C'$ (that are related as ancestor and descendant), how can we determine efficiently whether $C$ and $C'$ are connected directly with a back-edge, or through the mediation of a hanging subtree? To determine the existence of a back-edge that directly connects $C$ and $C'$ will usually be a very straightforward task for us, with the exception of some cases where we will have to use an indirect, counting argument. (Recall that we will only have to deal with at most three internal components, but still our task is challenging.) The problem of establishing a connection through the mediation of a hanging subtree is, in general, much more difficult, because the number of hanging subtrees can be $\Omega(n)$, and thus it is forbidding to check for each of them, explicitly, whether it provides some desired back-edges. Thus, we resort to somehow considering large batches of hanging subtrees at once, and this is why we work with specific permutations of a base DFS tree. The idea is to rearrange the children lists, so that children with ``similar properties" (which capture properties of their induced subtrees) appear as consecutive siblings in the children lists, and so we can process large segments of those lists at once. 
We will provide a concrete instance of this idea in the following section.

\section{Designing an oracle for $d_{\star}=3$}
\label{section:oraclefor3}
Let $F$ be a set of at most three vertices that have failed to work. In order to determine the connectivity between two vertices $x$ and $y$ on $G\setminus F$, we use the DFS-based framework outlined in Section~\ref{section:theframework}. Thus, it is sufficient to establish the connectivity between the (at most three) internal components of $T\setminus F$, where $T$ is a DFS tree of $G$.

In order to reduce the number of cases that may appear, it is very convenient to make the following two assumptions.

\begin{assumption}
\label{assumption1}
The root of $T$ is not a failed vertex, and no two failed vertices are related as parent and child.
\end{assumption}

\begin{assumption}
\label{assumption2}
Every back-edge has the property that its higher endpoint is a leaf of $T$, which is not a failed vertex.
\end{assumption}

First, let us provide a brief rationale for Assumptions~\ref{assumption1} and \ref{assumption2}. We note that these two assumptions follow from the observation, that splitting the edges of the graph does not change the connectivity structure under vertex failures. To be specific, let $E'$ be a subset of the edges of $G$, and suppose that we substitute every edge $e=(x,y)\in E'$ with two edges $(x,z_e)$ and $(z_e,y)$, where $z_e$ is a new \emph{auxiliary} vertex that \emph{splits} the edge $e$. Let $G'$ be the resulting graph. (Notice that the vertex set of $G'$ consists of that of $G$, with the addition of some auxiliary vertices.) Then it is easy to see that, if we are given a set of vertices $F$ of $G$, and two vertices $x$ and $y$ of $G\setminus F$, then the connectivity of $x$ and $y$ in $G\setminus F$ is the same as that of $x$ and $y$ in $G'\setminus F$. Thus, we have the luxury to split some edges of $G$, if that is convenient for us.

Now, we first perform a depth-first search with start vertex $s$, and let $T_0$ be the resulting DFS tree. In order to ensure that no failed vertex can be the root of the DFS tree, we attach an artificial root $r$ as the parent of $s$. Let $T$ be the resulting DFS tree. Then, in order to ensure that no two failed vertices can be related as parent and child, we split every tree edge $(v,p(v))$, where $v\notin\{r,s\}$. Furthermore, in order to ensure that the higher endpoint of every back-edge is a leaf and cannot be a failed vertex, we consider every vertex $v$ of $G$, we split the back-edges whose higher endpoint is $v$, and we let the auxiliary vertices that split them become children of $v$. An example of the result of those operations is given in Figure~\ref{figure:DFS}. It is easy to see that we can perform those splittings and reorganize the adjacency lists in $O(n+m)$ time, and run again a depth-first search, so that the resulting DFS tree $T$ satisfies Assumptions~\ref{assumption1} and \ref{assumption2}.

\begin{figure}[h!]\centering
\includegraphics[trim={1.2cm 20cm 0 1cm}, clip=true, width=1.1\linewidth]{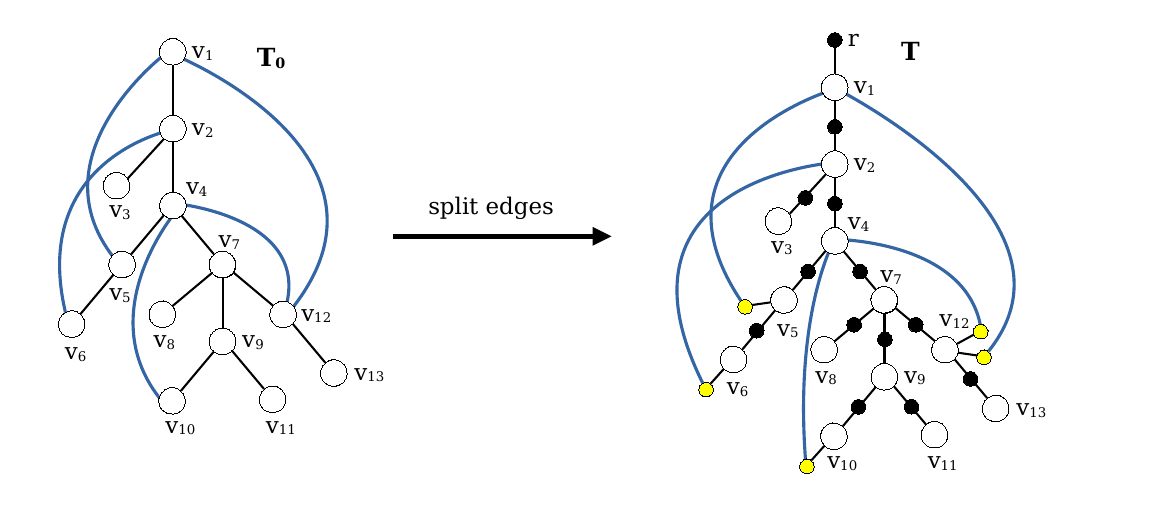}
\caption{\small{The result of adding an auxiliary root $r$ to the DFS tree $T_0$, and splitting the edges of the graph. The original vertices of the graph are coloured white. The artificial root, and the artificial vertices that split the tree-edges are coloured black. The children of the ordinary vertices that split the back-edges that stem from them are coloured yellow.}}\label{figure:DFS}
\end{figure}

Now, let us first note that the case where $F$ consists of a single vertex is trivial. Assumption~\ref{assumption1} implies that there is a single internal component in $T\setminus F$, and thus there is nothing to do with it: every connectivity query for two vertices $x$ and $y$ in $G\setminus F$ is answered with the help of the surviving back-edges (if needed), as explained in Section~\ref{section:internalAndHanging}. 

In the following section we introduce some elementary DFS-based concepts. These will help us in order to handle easily the case where $|F|=2$ (in Section~\ref{section:f=2}). This case provides an opportunity to demonstrate some of the techniques that will also be useful for handling the case where $|F|=3$. Finally, in Section~\ref{section:f=3}, we outline an analysis into cases, for determining the connection of the internal components of $T\setminus F$ when $|F|=3$. The full details for treating those cases are provided in Sections~\ref{section:vwnotrelated} and \ref{section:vwrelated}.

\subsection{Some concepts defined on a DFS tree}
\label{section:someconcepts}
Here we introduce some DFS-based concepts that will help the reader follow our discussion in the next two subsections. These are defined for all vertices (although they may be $\bot$ for some of them), and can be computed in linear time in total. For a more thorough treatment and motivation for those concepts, we refer to Section~\ref{section:DFSConcepts}.

The most fundamental concept that we rely on is that of \emph{the leaping back-edges over the parent of a vertex}. Specifically, for every vertex $v$, we let $B_p(v)$ denote the set of all back-edges $(x,y)$ such that $x\in T(v)$ and $y<p(v)$. The sets $B_p$ prove to be very useful for vertex-connectivity purposes, but we cannot afford to compute all of them explicitly, since there are graphs with $O(n)$ edges for which the total size of all $B_p$ sets (for some DFS trees) is $\Omega(n^2)$. Instead, we will rely on parameters that capture some properties of the back-edges in those sets that will be sufficient for our purposes.

First, we introduce some concepts that are defined w.r.t. the lower endpoints of the back-edges in the $B_p$ sets.
%
%
For a vertex $v$, we define $\mathit{low}(v):=\min\{y\mid \exists (x,y)\in B_p(v)\}$ and $\mathit{high}_p(v):=\max\{y\mid \exists (x,y)\in B_p(v)\}$. The $\mathit{low}$ points were used by Tarjan~\cite{DBLP:journals/siamcomp/Tarjan72} in order to solve various algorithmic problems of small connectivity, and the $\mathit{high}_p$ points were introduced by Georgiadis and Kosinas~\cite{DBLP:conf/isaac/GeorgiadisK20} in order to solve some algorithmic problems that relate to vertex-edge cuts.

Next, we introduce some concepts that are defined w.r.t. the higher endpoints of the back-edges in the $B_p$ sets.
For a vertex $v$, we define $\mathit{L}_p(v):=\min\{x\mid \exists (x,y)\in B_p(v)\}$ and $\mathit{R}_p(v):=\max\{x\mid \exists (x,y)\in B_p(v)\}$. These are called the \emph{leftmost} and the \emph{rightmost} point, respectively, of $v$. Furthermore, we also define $M_p(v):=\mathit{nca}\{L_p(v),R_p(v)\}$. Notice that $M_p(v)$ (as a vertex) is an invariant across all permutations of a base DFS tree, although $L_p(v)$ and $R_p(v)$ may vary.

We define the permutations $T_\mathit{lowInc}$ and $T_\mathit{highDec}$ of the base DFS tree $T$. The children lists in $T_\mathit{lowInc}$ are sorted in increasing order w.r.t. the $\mathit{low}$ points (where we consider $\bot$ to be the largest element), and the children lists in $T_\mathit{highDec}$ are sorted in decreasing order w.r.t. the $\mathit{high}_p$ points (where we consider $\bot$ to be the smallest element). Once the $\mathit{low}$ and $\mathit{high}_p$ points of all vertices are computed, we can easily construct $T_\mathit{lowInc}$ and $T_\mathit{highDec}$ in $O(n)$ time with bucket-sort.

\subsection{The case $|F|=2$}
\label{section:f=2}

Let $F=\{u,v\}$. Then there are two possibilities: either $u$ and $v$ are not related as ancestor and descendant, or one of them is an ancestor of the other. In the first case, Assumption~\ref{assumption1} implies that there is a single internal component in $T\setminus F$, and therefore this case is trivial.

So let us consider the second case, and let us assume w.l.o.g. that $u$ is an ancestor of $v$. Then, by Assumption~\ref{assumption1}, we have the situation depicted in Figure~\ref{figure:AB_MAIN}. Thus, there are two internal components, $A$ and $B$, and the goal is to determine whether they are connected, either directly with a back-edge, or through a hanging subtree of $v$. For this, we can use Lemma~\ref{lemma:lbelowwonT_high_MAIN}, which provides a very simple constant-time testable criterion for the connectivity of $A$ and $B$ in $G\setminus F$. This lemma also provides a good example of the techniques and arguments that we will employ in the most demanding cases as well. 

\begin{figure}[h!]\centering
\includegraphics[trim={0cm 23cm 0 0.5cm}, clip=true, width=1\linewidth]{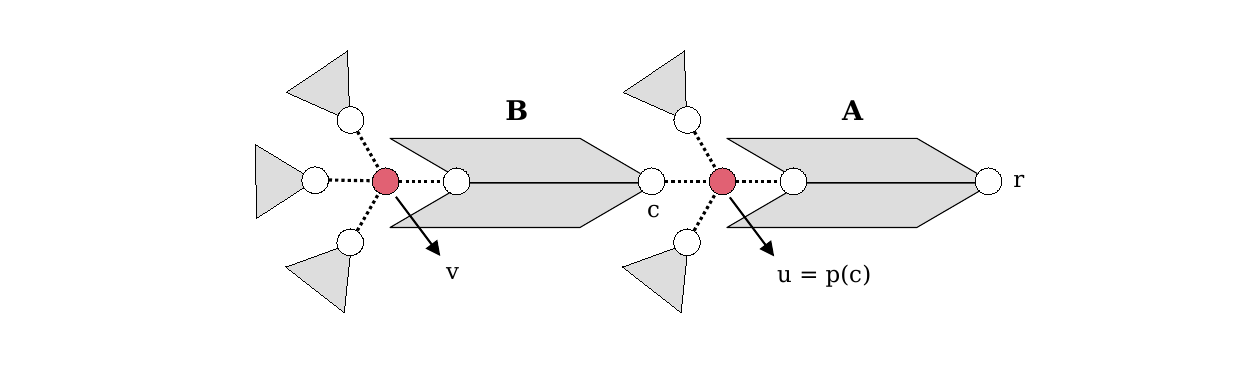}
\caption{\small{An illustration of the situation analyzed in Lemma~\ref{lemma:lbelowwonT_high_MAIN}. The set of failed vertices is $\{v,p(c)\}$, and the goal is to check whether the parts $A$ and $B$ remain connected. To do this, we rely on the leftmost and the rightmost points, $L_p(c)$ and $R_p(c)$, respectively, that provide back-edges from $T(c)$ to $T[p(p(c)),r]$.}}\label{figure:AB_MAIN}
\end{figure}

\begin{lemma}
\label{lemma:lbelowwonT_high_MAIN}
Let $c$ and $v$ be two vertices such that $c$ is a proper ancestor of $v$, and $p(c)\neq r$. Let $A$ be the set of vertices that are not descendants of $p(c)$, and let $B$ be the set of vertices that are descendants of $c$, but not of $v$. (See Figure~\ref{figure:AB_MAIN}.) Let $L_p(c)$ and $R_p(c)$ be the leftmost and the rightmost points of $c$ on $T_\mathit{highDec}$. Then, after removing $v$ and $p(c)$ from the graph, we have that $A$ is connected with $B$ if and only if one of the following is true:
\begin{enumerate}[label={(\arabic*)}]
\item{One of $L_p(c)$ and $R_p(c)$ is on $B$.}
\item{Both  $L_p(c)$ and $R_p(c)$ are proper descendants of $v$, and $\mathit{high}_p(d)\in B$, where $d$ is the child of $v$ that is an ancestor of $L_p(c)$.}
\end{enumerate} 
\end{lemma}
\begin{remark}
\normalfont
The idea behind the ``$\Rightarrow$" direction is the following. If $(1)$ is not true, then $B$ and $A$ are not connected directly with a back-edge. Therefore, there must exist a hanging subtree of $v$ that connects $B$ and $A$. Then, since we are working on $T_\mathit{highDec}$, the child $d$ of $v$ that is an ancestor of $L_p(c)$ must necessarily induce such a subtree. The ``$\Leftarrow$" direction is immediate. 
\end{remark}

\begin{proof}
First of all, notice that $A$ and $B$ are the two internal components that appear on $T$ when $v$ and $p(c)$ are removed from the graph. Also, notice that each of the conditions $(1)$ and $(2)$ implies that neither of $L_p(c)$ and $R_p(c)$ is $\bot$. Furthermore, if either of $L_p(c)$ and $R_p(c)$ is $\bot$, then there is no back-edge $(x,y)$ where $x$ is a descendant of $c$, and $y$ is a proper ancestor of $p(c)$, and therefore it is impossible that $A$ and $B$ remain connected in $G\setminus\{v,p(c)\}$. Thus, in the following we can assume that neither of $L_p(c)$ and $R_p(c)$ is $\bot$.

$(\Rightarrow)$ As noted previously, in order for $A$ and $B$ to be connected in $G\setminus\{v,p(c)\}$, there has to be a back-edge $(x,y)\in B_p(c)$, and so neither of $L_p(c)$ and $R_p(c)$ is $\bot$. Now let us assume that $(1)$ is not true. By definition, we have that both $L_p(c)$ and $R_p(c)$ are descendants of $c$. Thus, since $(1)$ is true, we have that both $L_p(c)$ and $R_p(c)$ must be descendants of $v$. 

First, we can observe that $B$ is not connected with $A$ directly with a back-edge. To see this, suppose the contrary. Then, there is a back-edge $(x,y)$ such that $x\in B$ and $y\in A$. This means that $x$ is a descendant of $c$, but not of $v$, and $y$ is a proper ancestor of $p(c)$. In particular, this implies that $(x,y)\in B_p(c)$, and therefore $L_p(c)\leq x\leq R_p(c)$. Since $x$ is not a descendant of $v$, we have that either $x<v$ or $v+\mathit{ND}(v)\leq x$. If $x<v$, then $L_p(c)\leq x$ implies that $L_p(c)<v$, which is impossible, since $L_p(c)$ is a descendant of $v$. And if $v+\mathit{ND}(v)\leq x$, then $x\leq R_p(c)$ implies that $v+\mathit{ND}(v)\leq R_p(c)$, which is also impossible, since $R_p(c)$ is a descendant of $v$. This shows that there is no back-edge that connects $B$ and $A$ directly. Thus, since $B$ remains connected with $A$ in $G\setminus\{v,p(c)\}$, this must be through hanging subtrees (of $v$).

So let $d'$ be a child of $v$ such that $T(d')$ is connected with both $B$ and $A$ through back-edges. Thus, there is a back-edge $(x,y)$ such that $x\in T(d')$ and $y\in B$. This implies that $y$ is a descendant of $c$, but not a descendant of $v$. Therefore, we have $(x,y)\in B_p(d')$ and $y\geq c$. Thus, $\mathit{high}_p(d')\geq c$. Furthermore, there is a back-edge $(x',y')$ with $x'\in T(d')$ and $y'\in A$. Since $d'$ is a descendant of $c$, this implies that $(x',y')\in B_p(c)$. Therefore, we have $L_p(c)\leq x'\leq R_p(c)$. Now, Assumption~\ref{assumption2} guarantees that neither of $L_p(c)$ and $R_p(c)$ is $v$. Thus, both of them are proper descendants of $v$. So let $d$ be the child of $v$ that is an ancestor of $L_p(c)$. Then, since $x'$ is a descendant of $d'$ with $L_p(c)\leq x'$, we infer that $d\leq d'$. Therefore, since we work on $T_\mathit{highDec}$, we have $\mathit{high}_p(d)\geq\mathit{high}_p(d')$, and therefore $\mathit{high}_p(d)\geq c$. Thus, since $\mathit{high}_p(d)$ is a proper ancestor of $p(d)=v$, we conclude that $\mathit{high}_p(d)\in B$.

$(\Leftarrow)$ First, let us assume $(1)$, and let $z$ be one among $L_p(c)$ and $R_p(c)$ that is in $B$. Then, we have that $(z,l_1(z))$ is a back-edge that connects $B$ and $A$. Now let us assume $(2)$. Notice that $(L_p(c),l_1(L_p(c)))$ is a back-edge such that $L_p(c)$ is a descendant of $d$, and $l_1(L_p(c))\in A$. Thus $T(d)$ is connected with $A$. Then, since $\mathit{high}_p(d)\in B$, we have that $T(d)$ is a hanging subtree of $v$ that connects $A$ and $B$ (because there is a back-edge $(x,\mathit{high}_p(d))\in B_p(d)$, which establishes the connection between $T(d)$ and $B$).
\end{proof}

\subsection{The case $|F|=3$}
\label{section:f=3}

In the case $|F|=3$, we initially determine the ancestry relation between the vertices from $F$. Thus, we have the following cases (see also Figure~\ref{figure:cases}):

\begin{enumerate}[label={(\arabic*)}] 
\item{No two vertices from $F$ are related as ancestor and descendant.}
\item{Two vertices from $F$ are related as ancestor and descendant, but the third is not related in such a way with the first two.}
\item{One vertex from $F$ is an ancestor of the other two, but the other two are not related as ancestor and descendant.}
\item{Any two vertices from $F$ are related as ancestor and descendant.}
\end{enumerate}
Notice that cases $(1)$-$(4)$ are mutually exclusive, and exhaust all possibilities for the ancestry relation between the vertices from $F$. 

\begin{figure}[h!]
\includegraphics[trim={0cm 11cm 0 0.5cm}, clip=true, width=1\linewidth]{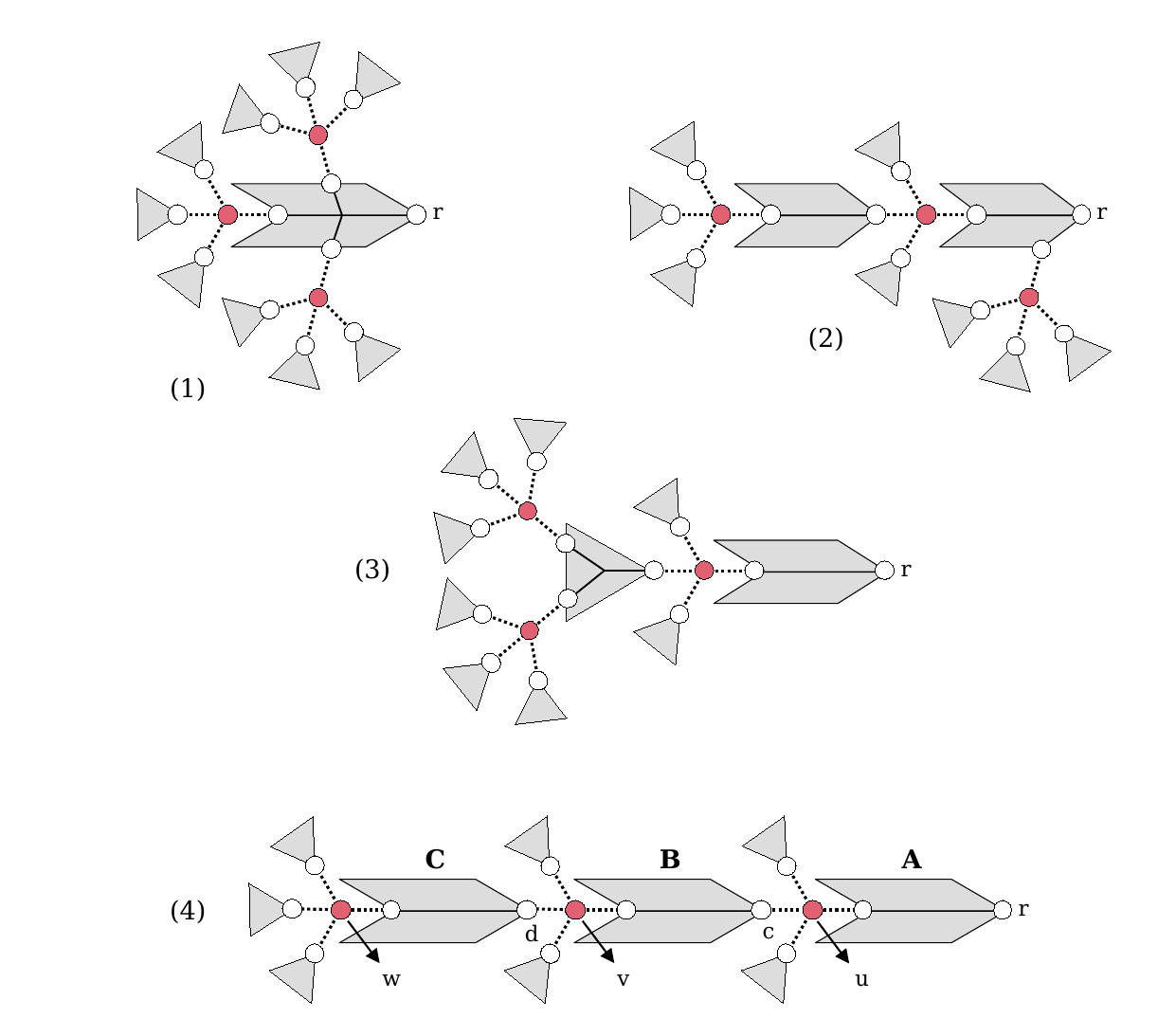}
\caption{\small{The four possibilities for the ancestry relation between the vertices from $F$, when $|F|=3$.}}\label{figure:cases}
\end{figure}

In case $(1)$ we have a single internal component in $T\setminus F$, and therefore this case is trivial. In case $(2)$ we have two internal components, and this case is handled precisely as the case where $|F|=2$ (i.e., we can determine the connectivity between the two internal components, using the parameters described in Lemma~\ref{lemma:lbelowwonT_high_MAIN}, which are associated with the child of the failed vertex that is an ancestor of the other failed vertex). 

Case $(3)$ is a little bit more involved, but it is handled with similar techniques as case $(2)$. Let $F=\{u,v,w\}$, and let us assume w.l.o.g. that $u$ is an ancestor of $v$ and $w$, but $v$ and $w$ are not related as ancestor and descendant. Here we distinguish between the case where $v$ and $w$ are descendants of different children of $u$, and the case where $v$ and $w$ are descendants of the same child of $u$. The first case is handled precisely with the same technique as case $(2)$ (see Section~\ref{section:vwdescendantsofcc'}). For the second case we have to be a little bit more careful, but the same technique essentially applies here too. (For the details, see Section~\ref{section:vwdescendantsofc}.) 

Case $(4)$ is much more involved. Let $F=\{u,v,w\}$, and let us assume w.l.o.g. that $u$ is an ancestor of $v$, and $v$ is an ancestor of $w$. Let also $c$ be the child of $u$ in the direction of $v$, and let $d$ be the child of $v$ in the direction of $w$, and let the three internal components $A$, $B$ and $C$, of $T\setminus F$ be as depicted in Figure~\ref{figure:cases}. Here we found it convenient to distinguish all different cases w.r.t. the location of $M_p(c)$. Thus, we have the six cases: $(i)$ $M_p(c)=\bot$, $(ii)$ $M_p(c)\in B$, $(iii)$ $M_p(c)=v$, $(iv)$ $M_p(c)$ is a descendant of a child $d'$ of $w$, $(v)$ $M_p(c)=w$, and $(vi)$ $M_p(c)\in C$. (Notice that these cases are mutually exclusive, and exhaust all possibilities for $M_p(c)$.) Before proceeding further, we would encourage the reader to try to solve some of those cases independently, before studying our own solutions.

In case $(i)$ we have that $A$ is isolated from $B$ and $C$, and it remains to determine whether $B$ remains connected with $C$ (either directly with a back-edge, or through a hanging subtree of $w$). We found that this case becomes easily manageable if we distinguish the four different cases for the location of $M_p(d)$: i.e., either $M_p(d)=\bot$, or $M_p(d)$ is a descendant of a child of $w$, or $M_p(d)=w$, or $M_p(d)\in C$. This case is the simplest one, and very instructive for the more demanding cases. (For the details, see Section~\ref{section:M(c)=bot}.)

In case $(ii)$ we have that $B$ is connected with $A$ directly with a back-edge, and it remains to determine whether $C$ is connected with either $B$ or $A$ (either directly with a back-edge, or through a hanging subtree of $w$). This is done similarly as in case $(i)$. (For the details, see Section~\ref{section:M(c)inB}.)

In case $(iii)$, there is a possibility that a child of $u$ induces a hanging subtree that connects $A$ and $B$. This is easy to check, and, if true, we can handle the rest as in case $(ii)$. Otherwise, we have that only $T(d)$ has the potential to provide back-edges that establish the connectivity of $A$ and $B$. Here we distinguish all the different cases for $\mathit{high}_p(d)$ and $\mathit{low}(d)$, and the most difficult case appears when $\mathit{high}_p(d)\in B$ and $\mathit{low}(d)\in A$. However, in this case $d$ is uniquely determined by $c$, and so we can gather enough information during the preprocessing phase in order to accommodate for this case. (For the details, see Section~\ref{section:Mp(c)=v}.)


In case $(iv)$, notice that $T(d')$ is the only hanging subtree that may connect $A$ with either $B$ or $C$. Here we found it very convenient to distinguish between the different cases for $M_p(d)$, as in case $(i)$. (Excluding the case $M_p(d)=\bot$, which here is impossible.) Then, if $M_p(d)$ is a descendant of a child of $w$, it must necessarily be a descendant of $d'$. Thus, this case is particularly easy. If $M_p(d)\in C$, then we can use either $L_p(d)$ and $R_p(d)$ in order to establish the connectivity between $C$ and $B$, or the leftmost and rightmost points of $d$ that reach the segment $T[p(v),c]$ (for the definition of those, see Section~\ref{section:leftmostReachSegment}). Thus, this case is still not difficult to manage. The case $M_p(d)=w$ presents a singular difficulty, and we have to resort to a counting argument to determine the existence of some back-edges. (For the details, see Section~\ref{section:Mp(c)Descw}.)


In case $(v)$ we distinguish the two cases for $M_p(d)$: either $M_p(d)=w$, or $M_p(d)\in C$. In the first case, we can appropriately use the leftmost and the rightmost points of $d$ that reach the segment $T[p(v),c]$ (for the definition of those, see Section~\ref{section:leftmostReachSegment}). In the second case, we follow essentially the same approach as for the case $M_p(d)\in C$ when $M_p(c)$ is a descendant of a child of $w$. (For the details, see Section~\ref{section:MpC=w}.)

Finally, in case $(vi)$, we have $M_p(d)\in C$ (as a consequence of $M_p(c)\in C$). An easy case appears if $M_p(d)=M_p(c)$, because then we can use the leftmost and rightmost points of $d$ that reach the segment $T[p(v),c]$. Otherwise, (after sorting out some easier cases), there appears again a singular difficulty, for which we have to resort to a counting argument to determine the existence of some back-edges. (For the details, see Section~\ref{section:MpCinC}.)

\paragraph{Reporting the number of connected component in the case where $d_{\star}=3$.}
\label{paragraph:reporting}
Our main contribution that provides the optimal oracle for the case $d_{\star}=3$ is to demonstrate how to determine the connectivity in $G\setminus F$ between the internal components of $T\setminus F$, for $|F|\leq 3$. Thus, we can use the idea described in Section~\ref{section:numberOfComponents}, in order to report the number of connected components of $G\setminus F$ as well. Recall that the idea in Section~\ref{section:numberOfComponents} is to simply count the number of the isolated hanging subtrees of $T\setminus F$, and add to their number that of the connected components of the connectivity graph $\mathcal{R}$ of the internal components of $T\setminus F$. 

However, in order to determine the connectivity between the internal components of $T\setminus F$, in the case where $d_{\star}=3$, we rely on Assumptions~\ref{assumption1} and \ref{assumption2}. The way these assumptions are ensured is to split the edges of the graph, and create a DFS tree with the desired properties. This construction is given by first using a DFS tree $T_0$ of $G$, on whose root we attach a new artificial root, and then we split the edges of $G$ appropriately, so that we get the final DFS tree $T$ on which we will work (see Figure~\ref{figure:DFS}). (In the following, we assume familiarity both with the contents of Section~\ref{section:numberOfComponents}, and with the construction of $T$ from $T_0$.)

Now there are two possible sources for miscounting. First, the number of the connected components of $\mathcal{R}$ might not be the same as what we would get from $T_0$. More precisely, the problem is that some internal components of $T\setminus F$ might not be internal components of $T_0\setminus F$, and thus we overestimate the number of internal components. The second possible source of error is that we may erroneously include in the number of the isolated hanging subtrees those that are induced by the children of $f$ that are the auxiliary vertices that split the back-edges that stem from $f$, for $f\in F$. 

It is not difficult to see that both these problems have an easy fix. First, notice that some connected components of $\mathcal{R}$ may correspond to singleton auxiliary vertices, and so it is sufficient to simply discard those. There are two possible ways to have such ``fake" internal components. First, if the root of $T_0$ is a failed vertex, then the root of $T$ is a singleton internal component of $T\setminus F$ that consists of the artificial root. (See Figure~\ref{figure:DFS}.) And second, if two failed vertices have the property that one of them is the parent of the other in $T_0$, then the vertex $z$ that splits the tree-edge that connects them in $T_0$ constitutes a singleton internal component of $T\setminus F$ that consists of $z$.

In order to avoid counting the number of isolated hanging subtrees that are induced by the auxiliary vertices that have split the back-edges, we can simply ignore them when we compute the segments of children of the failed vertices whose $\mathcal{L}$ sets consist only of failed vertices (the notation here refers to our procedure in Section~\ref{section:numberOfComponents} for counting the number of isolated hanging subtrees). To do that, we can simply have the auxiliary vertices that split back-edges put at the end of the children lists to which they belong, and then we make sure that the binary search (described in Section~\ref{section:numberOfComponents}) avoids those final segments of the children lists.

\section{Introduction to the toolkit of the DFS-based parameters}
\label{section:DFSConcepts}

\subsection{The sets of back-edges that leap over parents of vertices}
The most fundamental concept that we rely on is that of the \emph{leaping back-edges over the parent of a vertex}. Specifically, for every vertex $v$, we let $B_p(v)$ denote the set of all back-edges $(x,y)$ such that $x\in T(v)$ and $y<p(v)$.\footnote{We note that this definition is very similar to the set $B(v)$ of the back-edges that leap over a vertex $v$, which was used in previous works to deal with edge-connectivity problems \cite{DBLP:conf/esa/NadaraRSS21,DBLP:conf/esa/GeorgiadisIK21,DBLP:conf/soda/Kosinas24}. (I.e., the only difference in the definition of $B(v)$, is that, instead of ``$y<p(v)$", we have``$y<v$".)} The usefulness of the $B_p$ sets can be made apparent by the following simple example. Suppose that a vertex $v\neq r$ is removed from the graph. Then, the connected components of $T\setminus{v}$ are: the set $A$ of all vertices that are not descendants of $v$, and the subtrees of the form $T(c)$, for every child $c$ of $v$. (See Figure~\ref{figure:A}.) Then it is easy to see that, if $c$ and $c'$ are two distinct children of $v$, then $T(c)$ and $T(c')$ remain connected on $G\setminus{v}$ if and only if they remain connected with $A$ on $G\setminus{v}$. Thus, in order to establish the connectivity on $G\setminus{v}$, we only have to determine, for every child $c$ of $v$, whether $T(c)$ remains connected with $A$ on $G\setminus{v}$. But this is obviously equivalent to the condition $B_p(c)\neq\emptyset$, which can be easily checked if we have computed a certificate for the non-emptiness of $B_p(c)$ (e.g., its number of elements, or just an edge in it). 

\begin{figure}[h!]\centering
\includegraphics[trim={0cm 23cm 0 0cm}, clip=true, width=1\linewidth]{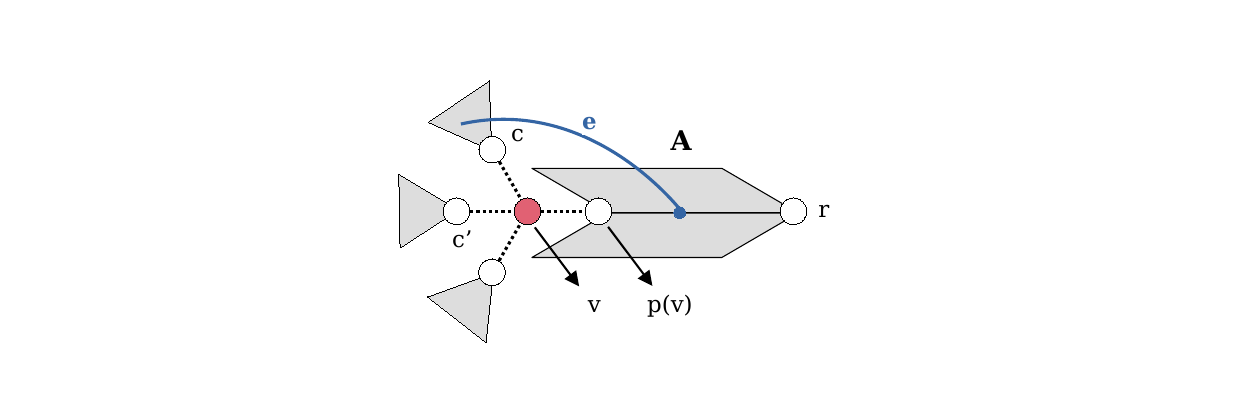}
\caption{\small{Here $v$ is a vertex that has failed to work. Thus, the connected components of $T\setminus{v}$ are: the set $A=T(r)\setminus T(v)$ of the non-descendants of $v$, and the subtrees of the form $T(c)$, where $c$ is a child of $v$. We see that $B_p(c)$ contains the back-edge $e$, and therefore $T(c)$ remains connected with $A$ on $G\setminus{v}$.}}\label{figure:A}
\end{figure}  

The sets $B_p$ prove to be very useful for vertex-connectivity purposes, but we cannot afford to compute all of them explicitly, since there are graphs with $O(n)$ edges for which the total size of all $B_p$ sets (for some DFS trees) is $\Omega(n^2)$. Instead, we will rely on parameters that capture some properties of the back-edges in those sets that will be sufficient for our purposes. For example, we will use the values $|B_p(v)|$ and $\mathit{sumY}(v)$, where $\mathit{sumY}(v)$ denotes the sum of the lower endpoints of the back-edges in $B_p(v)$. These two parameters can be easily computed for all vertices $v$, in linear time in total, by processing the vertices in a bottom-up fashion. (For $|B_p(v)|$, see, e.g. \cite{DBLP:conf/isaac/GeorgiadisK20}; $\mathit{sumY}(v)$ is computed similarly.) Moreover, we will need some concepts that capture some aspects of the distribution of the endpoints of the back-edges in the $B_p$ sets, and this is what we discuss next.

\subsection{The $\mathit{low}$ and $\mathit{high}$ points}
\label{section:dfs_highlow}
For a vertex $v$, it is useful to consider the lowest endpoints of the back-edges that stem from $v$. Thus, we let $l_1(v)$ denote the minimum lower endpoint of all back-edges whose higher endpoint is $v$. (If $B_p(v)=\emptyset$, then $l_1(v):=\bot$.) Then we inductively define $l_i(v)$, for $i>1$, as $l_i(v)=\min\{y\mid \mbox{ there is a back-edge } (v,y) \mbox{ and } y>l_{i-1}(v)\}$. Of course, it may be that $l_i(v)=\bot$, for some $i\geq 1$. In this case, we have $l_j(v)=\bot$, for every $j\geq i$.

Let $v$ be a vertex (distinct from the root or from a child of the root). If $B_p(v)\neq\emptyset$, we can consider the lowest lower endpoint of all back-edges in $B_p(v)$, denoted as $\mathit{low}(v)$. Thus, $\mathit{low}(v)=\min\{y\mid \exists (x,y)\in B_p(v)\}$. This concept was used by Tarjan~\cite{DBLP:journals/siamcomp/Tarjan72} in order to solve various problems of low connectivity (e.g., computing all articulation points in undirected graphs). The elegance in using the $\mathit{low}$ points consists in the simplicity in their computation, and in their usefulness in determining the existence of back-edges that connect subtrees with the remaining graph. For problems of higher connectivity, we shall need the generalization to the $\mathit{low}_i$ points, for $i\geq 1$, introduced by \cite{DBLP:conf/esa/Kosinas23}, which are defined inductively as follows. First, $\mathit{low}_1(v):= \mathit{low}(v)$. For $i>1$, we define $\mathit{low}_i(v):=\min\{y\mid \exists (x,y)\in B_p(v) \mbox{ and } y>\mathit{low}_{i-1}(v)\}$. In words, we may say that $\mathit{low}_i$ is the $i$-th lowest lower endpoint of all back-edges in $B_p(v)$. (Of course, it may be that $\mathit{low}_i(v)=\bot$, for some $i\geq 1$, in which case we have that $\mathit{low}_j(v)=\bot$ for every $j\geq i$.) In the sequel we will only use $\mathit{low}_1$ and $\mathit{low}_2$, and very often we will simply denote $\mathit{low}_1$ as $\mathit{low}$. In general, the usefulness of the $\mathit{low}_i$ points can be demonstrated by the example shown in Figure~\ref{figure:ABCD}. To put it briefly, the idea is that, if the $\mathit{low}_1$ point happens to be a failed vertex, then we can resort to the $\mathit{low}_2$ point in order to establish a particular kind of connection. For any positive integer $k$, the $\mathit{low}_1,\dots,\mathit{low}_k$ points of all vertices can be easily computed in $O(m+nk\log{k})$ time in total~\cite{DBLP:conf/esa/Kosinas23}. 

\begin{figure}[h!]\centering
\includegraphics[trim={0cm 14cm 0 0cm}, clip=true, width=1\linewidth]{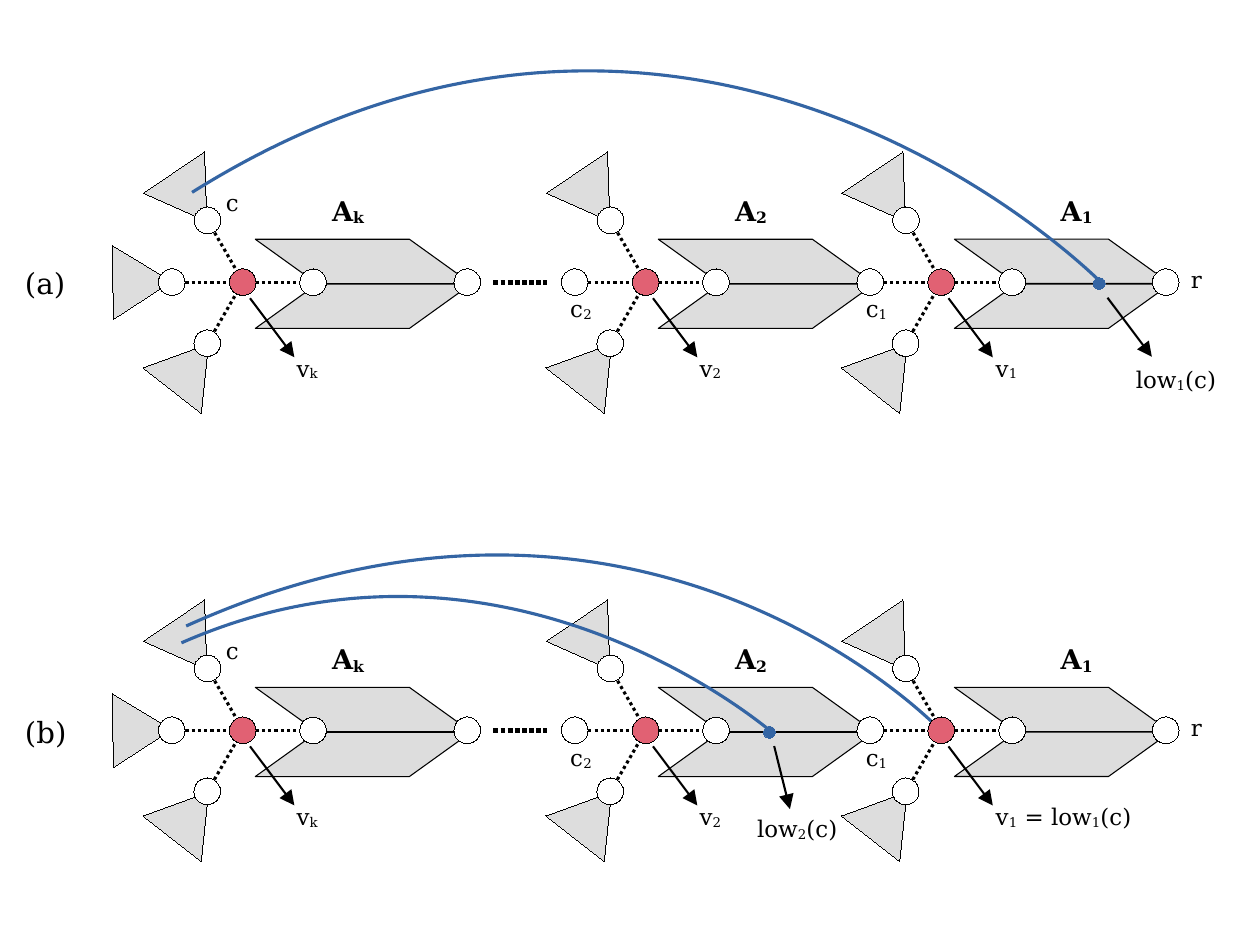}
\caption{\small{Here the vertices $v_1,\dots,v_k$ have failed to work. For every $i\in\{1,\dots,k-1\}$, we have that $v_i$ is an ancestor of $v_{i+1}$, and $c_i$ is the child of $v_i$ in the direction of $v_{i+1}$. $A_1$ denotes the set of vertices in $T(r)\setminus T(v_1)$, and $A_i$, for $i\in\{2,\dots,k\}$, denotes the set of vertices in $T(c_{i-1})\setminus T(v_i)$. Using the $\mathit{low}_i$ points, we can establish a connection between $T(c)$ and some of the parts $A_1,\dots,A_k$. For example, there is a back-edge that connects $T(c)$ directly with $A_1$ if and only if $\mathit{low}_1(c)\in A_1$, which is equivalent to $\mathit{low}_1(c)\in T[p(v_1),r]$ (since $\mathit{low}_1(c)$ is an ancestor of $c$). This situation is illustrated in $(a)$. More generally, we have that the lowest part $A_i$ which is connected directly with $T(c)$ with a back-edge is given by the the smallest $i\in\{1,\dots,k\}$ such that $\mathit{low}_1(c)\in A_i$. If it happens that $\mathit{low}_1(c)=v_i$, for some $i\in\{1,\dots,k-1\}$, then there is a back-edge that connects $T(c)$ directly with $A_{i+1}$ if and only if $\mathit{low}_2(c)\in A_{i+1}$. This situation is illustrated in $(b)$ (for $i=1$).}}\label{figure:ABCD}
\end{figure}


Similarly to the $\mathit{low}_i$ points, we define the $\mathit{high}_p$ points as follows.\footnote{We use the notation $\mathit{high}_p$ in order to distinguish those points from $\mathit{high}$, which are defined similarly \cite{DBLP:conf/soda/Kosinas24}, but are used for edge-connectivity problems.} First, $\mathit{high}_p(v)$ is the highest lower endpoint of all back-edges in $B_p(v)$. Thus, $\mathit{high}_p(v)=\max\{y\mid \exists (x,y)\in B_p(v)\}$. Then, the higher order $\mathit{high}_p$ points, denoted as $\mathit{high}^p_i$, for $i\geq 1$, are defined as follows. First, we let $\mathit{high}^p_1(v):=\mathit{high}_p(v)$, and for $i>1$ we have $\mathit{high}^p_i(v):=\max\{y\mid \exists (x,y)\in B_p(v) \mbox{ and } y<\mathit{high}^p_{i-1}(v)\}$. (Again, it may be that $\mathit{high}^p_i(v)=\bot$, for some $i\geq 1$, in which case we have that $\mathit{high}^p_j(v)=\bot$ for every $j\geq i$.) Notice that, for the case of a single vertex failure, the $\mathit{high}^p_1$ point can have the same functionality as $\mathit{low}_1$. However, for the case of more vertex failures, the functionality of $\mathit{high}^p_1$ is different, but it is a kind of dual to $\mathit{low}_1$. 
To be more specific, as in the case of the $\mathit{low}_i$ points, we need the higher order $\mathit{high}_p$ points in order to handle the case that some of them happen to be precisely failed vertices. (See Figure~\ref{figure:ABCD2}.) For the case of three vertex-failures, the first two $\mathit{high}_p$ points, $\mathit{high}^p_1$ and $\mathit{high}^p_2$, suffice for our purposes. In Section~\ref{section:computeHigh} we provide a linear-time algorithm for computing the $\mathit{high}^p_1$ and $\mathit{high}^p_2$ points of all vertices.

\begin{figure}[h!]\centering
\includegraphics[trim={0cm 14cm 0 0cm}, clip=true, width=1\linewidth]{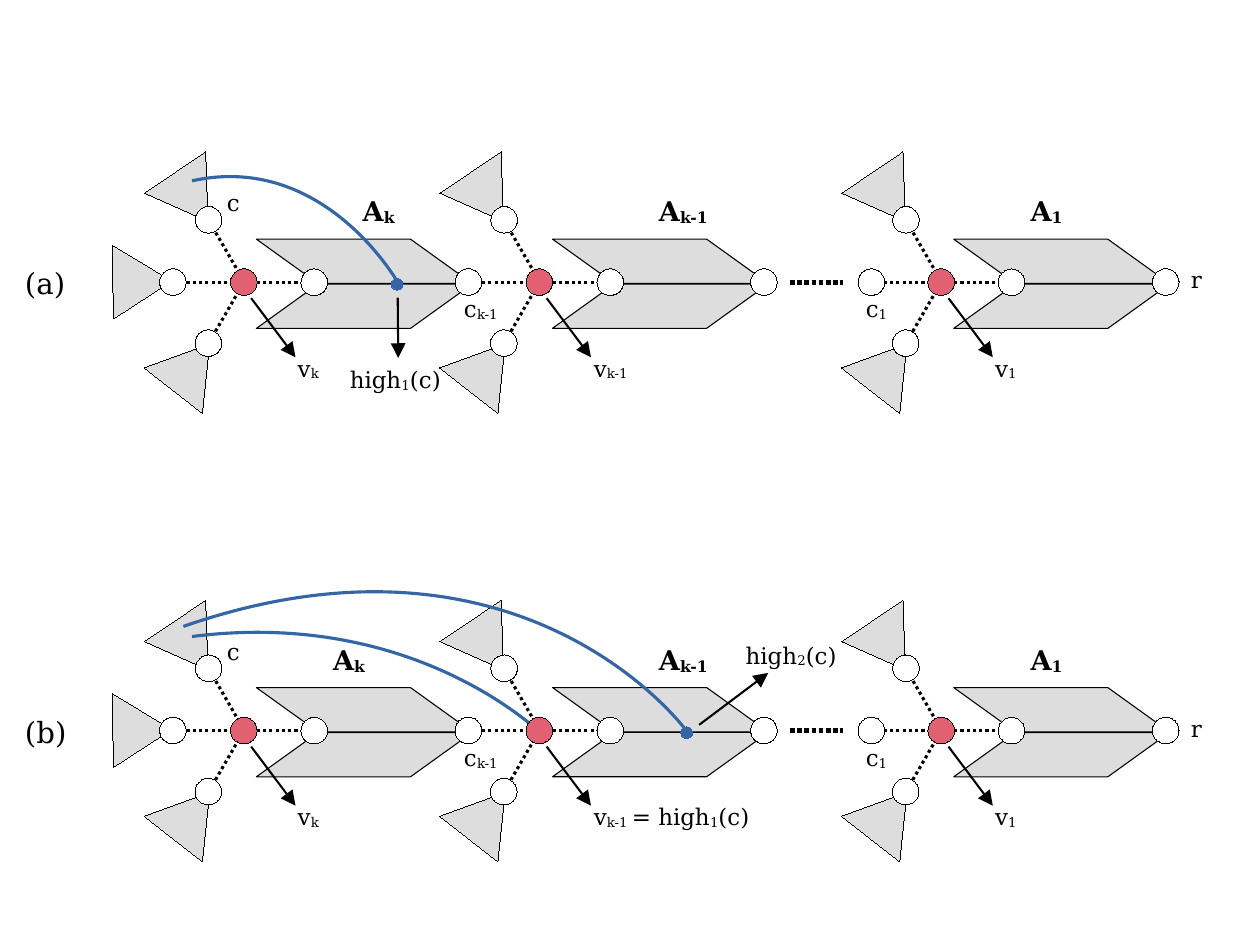}
\caption{\small{Here the vertices $v_1,\dots,v_k$ have failed to work. For every $i\in\{1,\dots,k-1\}$, we have that $v_i$ is an ancestor of $v_{i+1}$, and $c_i$ is the child of $v_i$ in the direction of $v_{i+1}$. $A_1$ denotes the set of vertices in $T(r)\setminus T(v_1)$, and $A_i$, for $i\in\{2,\dots,k\}$, denotes the set of vertices in $T(c_{i-1})\setminus T(v_i)$. Using the $\mathit{high}^p_i$ points, we can establish a connection between $T(c)$ and some of the parts $A_1,\dots,A_k$. For example, there is a back-edge that connects $T(c)$ directly with $A_k$ if and only if $\mathit{high}^p_1(c)\in A_k$, which is equivalent to $\mathit{high}^p_1(c)\in T[p(v_k),c_{k-1}]$ (since $\mathit{high}^p_1(c)$ is an ancestor of $c$). This situation is illustrated in $(a)$. More generally, we have that the highest $i$ such that $A_i$ is connected directly with $T(c)$ with a back-edge is given by the the smallest $i\in\{1,\dots,k\}$ such that $\mathit{high}^p_1(c)\in A_i$. If it happens that $\mathit{high}^p_1(c)=v_i$, for some $i\in\{1,\dots,k-1\}$, then there is a back-edge that connects $T(c)$ directly with $A_i$ if and only if $\mathit{high}^p_2(c)\in A_i$. This situation is illustrated in $(b)$ (for $i=k-1$).}}\label{figure:ABCD2}
\end{figure}

It is immediate from their definition, that the parameters $\mathit{low}_i$, for $i\geq 1$, and $\mathit{high}^p_i$, for $i\geq 1$, are an invariant across all permutations of a base DFS tree. More precisely, for any vertex $v$, any $i\geq 1$, and two DFS trees $T$ and $T'$ that are a permutation of each other, we have that the $\mathit{low}_i(v)$ points, computed in either $T$ or $T'$, are identical (if viewed as vertices), although they may be assigned different DFS numbers. The same is true for the $\mathit{high}^p_i(v)$ points, computed in either $T$ or $T'$.

In the sequel we will use the permutations $T_\mathit{lowInc}$ and $T_\mathit{highDec}$ of a base DFS tree $T$. The children lists in $T_\mathit{lowInc}$ are sorted in increasing order w.r.t. the $\mathit{low}$ points (where we consider $\bot$ to be the largest element), and the children lists in $T_\mathit{highDec}$ are sorted in decreasing order w.r.t. the $\mathit{high}_p$ points (where we consider $\bot$ to be the smallest element). Once the $\mathit{low}$ and $\mathit{high}_p$ points of all vertices are computed, we can easily construct $T_\mathit{lowInc}$ and $T_\mathit{highDec}$ in $O(n)$ time with bucket sort.

Finally, we will need the parameters $\mathit{numLow}(v):=|\{(x,y)\in B_p(v)\mid y=\mathit{low}(v)\}|$ and $\mathit{numHigh}(v):=|\{(x,y)\in B_p(v)\mid y=\mathit{high}_p(v)\}|$. In Section~\ref{section:numlow} we provide algorithms for computing the $\mathit{numLow}$ and $\mathit{numHigh}$ values of all vertices, in linear time in total.

\subsection{The leftmost, rightmost, and maximum points}
\label{section:leftmostrightmost}
Now we consider some concepts that are defined with respect to the higher endpoints of the back-edges in $B_p(v)$. Most of these concepts are not an invariant across all permutations of a DFS tree. Thus, one should be careful to specify what is the DFS tree that one uses. In the sequel, however, this will always be clear from the context, and thus we do not have to overload with special notation those parameters (i.e., in order to emphasize the DFS tree on which they are computed). Thus, from now on we simply assume that we work on a base DFS tree $T$. 

First, we consider the \emph{leftmost} and the \emph{rightmost} points of $v$, denoted as $L_p(v)$ and $R_p(v)$, respectively, which are defined as the minimum and the maximum, respectively, descendant of $v$, with the property that they provide a back-edge that leaps over the parent of $v$. More precisely, $L_p(v)=\min\{x\mid \exists (x,y)\in B_p(v)\}$ and $R_p(v)=\max\{x\mid \exists (x,y)\in B_p(v)\}$. The nearest common ancestor of $L_p(v)$ and $R_p(v)$ is denoted as $M_p(v)$, and it is called the \emph{maximum} point of $v$. In other words, $M_p(v)$ is the maximum descendant of $v$ which is an ancestor of every vertex which provides a back-edge that leaps over the parent of $v$. It is easy to see that, contrary to $L_p$ and $R_p$, the $M_p$ points are an invariant across all permutations of a DFS tree.

Notice that two vertices $v$ and $v'$ with $M_p(v)=M_p(v')$ are related as ancestor and descendant (since they have $M_p(v)=M_p(v')$ as a common descendant). For a vertex $v$ with $M_p(v)\neq\bot$, we define $\mathit{next}_{M_p}(v)$ as the greatest vertex which is lower than $v$ and has the same $M_p$ point as $v$. Then it is easy to see that the collection of all vertices that have a vertex $x$ as their $M_p$ point, is given by the sequence $x$, $\mathit{next}_{M_p}(x)$, $\mathit{next}_{M_p}(\mathit{next}_{M_p}(x))$, $\dots$, and this is a list of vertices in decreasing order.

Notice that, for a vertex $v$ with $B_p(v)\neq\emptyset$, we have $l_1(L_p(v))<p(v)$ and $l_1(R_p(v))<p(v)$. However, it is not necessary that $l_2(L_p(v))<p(v)$ or $l_2(R_p(v))<p(v)$ (and we may even have $l_2(L_p(v))=\bot$ or $l_2(R_p(v))=\bot$). In such a case, for our purposes we may need to have some ``back-up" leftmost and rightmost points, that can still establish the connection between a subtree induced by a descendant of $v$ and another part of the graph. More specifically, the problem that may appear is the same that forced us to extend the $\mathit{low}_1$ and $\mathit{high}^p_1$ points to higher order points. Thus, it may be e.g. that $z=l_1(L_p(v))$ is a failed vertex, and $l_2(L_p(v))=\bot$ or $l_2(L_p(v))\geq p(v)$. In this case, we will need the leftmost descendant of $v$ that provides a back-edge from $B_p(v)$ whose lower endpoint is not $z$. Thus, given two vertices $v$ and $z$, where $z$ is an ancestor of $v$, we define \emph{the leftmost point of $v$ that skips $z$} as the minimum descendant $x$ of $v$ for which there is a back-edge $(x,y)\in B_p(v)$ with $y\neq z$.

As in \cite{DBLP:conf/soda/Kosinas24}, we extend the concepts of the leftmost and the rightmost points to subtrees of vertices as follows. For a vertex $v$ and a descendant $d$ of $v$, we let $L(v,d)$ and $R(v,d)$ denote the minimum and the maximum, respectively, descendant of $d$ that provides a back-edge that leaps over the parent of $v$. More precisely, $L(v,d)=\min\{x\mid x\in T(d) \mbox{ and } \exists (x,y)\in B_p(v)\}$ and $R(v,d)=\max\{x\mid x\in T(d) \mbox{ and } \exists (x,y)\in B_p(v)\}$. Using a similar algorithm as in \cite{DBLP:conf/soda/Kosinas24}, we have the following:

\begin{proposition}
\label{proposition:L(v,d)}
Let $N$ be a positive integer, and let $L(v_i,d_i)$, for $i\in\{1,\dots,N\}$, be a set of queries for the leftmost point in $T(d_i)$ that provides a back-edge from $B_p(v_i)$, for $i\in\{1,\dots,N\}$. Then, we can answer all those queries in $O(N+n+m)$ time in total. Furthermore, the result holds if, instead of ``$L(v_i,d_i)$", we have ``$R(v_i,d_i)$", for every $i\in\{1,\dots,N\}$. 
\end{proposition}

The idea for establishing Proposition~\ref{proposition:L(v,d)} is to process the vertices $v_i$ in a bottom-up fashion. Then, we search in $T(d_i)$, in increasing order, for the first vertex $x$ with the property that there is a back-edge $(x,y)\in B_p(v_i)$. In order to achieve linear time, we have to be able to avoid at once segments of vertices $x$ that are provably no longer able to provide a back-edge with low enough lower endpoint (i.e., because their $l_1$ point is $\geq p(v_i)$). We can do that with the use of a DSU data structure~\cite{DBLP:journals/jcss/GabowT85} that unites consecutive (w.r.t. the DFS numbering) vertices. For the details, see \cite{DBLP:conf/soda/Kosinas24}. 

Here we further extend those concepts, to the leftmost and the rightmost points that reach a particular segment of $T$. Specifically, for three vertices $z$, $u$ and $v$, such that $z$ is a descendant of $u$, and $u$ is a descendant of $v$, we define \emph{the leftmost and the rightmost points that start from $T(z)$ and reach the segment $T[u,v]$}, denoted as $L(z,u,v)$ and $R(z,u,v)$, as the minimum and the maximum, respectively, descendant of $z$, from which stems a back-edge with lower endpoint on $T[u,v]$. Thus, if we let $S=\{x\mid x\in T(z) \mbox{ and there is a back-edge } (x,y) \mbox{ with } y\in T[u,v]\}$, then we have $L(z,u,v)=\min(S)$ and $R(z,u,v)=\max(S)$. Notice that, for two vertices $z$ and $v$ such that $z$ is a descendant of $v$ (and $v$ is not the root or a child of the root), we have $L(v,z)=L(z,p(p(v)),r)$ and $R(v,z)=R(z,p(p(v)),r)$. In Section~\ref{section:leftmostReachSegment} we show that, if we have a collection of queries for the leftmost and the rightmost points that reach a particular segment of $T$, then we can compute all of them in linear time in total, provided that they satisfy a kind of nestedness property. Notice that such queries resemble the 2D-range queries that were used in \cite{DBLP:conf/esa/Kosinas23}, but here we need to have precomputed the answers to them in the preprocessing phase, and this we can only do for specific batches of queries that satisfy some special properties. For more details, see Section~\ref{section:leftmostReachSegment}.

\subsection{Some useful properties of the DFS-based parameters}
\label{section:usefulProperties}

In this section we state and prove some propositions that will be needed in the sequel. Throughout we assume that we work on a DFS tree $T$, and all DFS-based concepts refer to $T$.

\begin{lemma}
\label{lemma:Mp}
Let $c$ and $d$ be two vertices such that $c$ is an ancestor of $d$. Suppose that $M_p(c)$ is a descendant of $d$. Then $M_p(d)$ is an ancestor of $M_p(c)$.
\end{lemma}
\begin{proof}
Let $(x,y)$ be a back-edge in $B_p(c)$. Then $x$ is a descendant of $M_p(c)$, and therefore a descendant of $d$. Furthermore, $y$ is a proper ancestor of $p(c)$, and therefore a proper ancestor of $p(d)$. This shows that $(x,y)\in B_p(d)$. A first consequence of that is that $M_p(d)\neq\bot$. A second consequence is that $x$ is a descendant of $M_p(d)$. Therefore, due to the generality of $(x,y)\in B_p(c)$, we conclude that $M_p(c)$ is a descendant of $M_p(d)$.
\end{proof}

\begin{lemma}
\label{lemma:siblingsM}
Let $d$ be a vertex such that $\mathit{next}_{M_p}(d)\neq\bot$. Then, for every sibling $d'$ of $d$, we have $\mathit{next}_{M_p}(d')=\bot$.
\end{lemma}
\begin{proof}
Let us suppose, for the sake of contradiction, that there is a vertex $d'\neq d$ such that $p(d')=p(d)$ and $\mathit{next}_{M_p}(d')\neq\bot$. To avoid notational clutter, let $f:=\mathit{next}_{M_p}(d)$ and $f':=\mathit{next}_{M_p}(d')$. By definition, we have that $f$ is a proper ancestor of $d$, and $f'$ is a proper ancestor of $d'$. Thus, both $f$ and $f'$ are ancestors of $p(d)=p(d')$, and therefore they are related as ancestor and descendant. So let us assume, w.l.o.g., that $f'$ is an ancestor of $f$. Since $f'=\mathit{next}_{M_p}(d')$, there is a back-edge $(x,y)$ such that $x$ is a descendant of $M_p(d')$ and $y$ is a proper ancestor of $p(f')$. Thus, $x$ is a descendant of $d'$, and therefore a descendant of $p(d')=p(d)$, and therefore a descendant of $f$. Furthermore, since $f'$ is an ancestor of $f$, we have that $y$ is a proper ancestor of $p(f)$. This shows that $(x,y)\in B_p(f)$, and therefore $x$ is a descendant of $M_p(f)$. Since $f=\mathit{next}_{M_p}(d)$, we have $M_p(f)=M_p(d)$. Thus, $x$ is a descendant of $M_p(d)$, and therefore a descendant of $d$. But now we have that $x$ is a common descendant of $d$ and $d'$, in contradiction to the fact that $d$ and $d'$ are distinct vertices with the same parent.
\end{proof}

\begin{lemma}
\label{lemma:lowestl1}
Let $v$ be a vertex with $B_p(v)\neq\emptyset$. Then, Assumption~\ref{assumption2} implies that $L_p(v)$ on $T_\mathit{lowInc}$ is a descendant of $v$ with the lowest $l_1$ point.
\end{lemma}
\begin{proof}
Let us suppose, for the sake of contradiction, that there is a descendant $x$ of $v$ with $l_1(x)<l_1(L_p(v))$, and let us assume, w.l.o.g., that $x$ is the lowest descendant of $v$ on $T_\mathit{lowInc}$ with this property. Let $z$ be the nearest common ancestor of $x$ and $L_p(v)$. By Assumption~\ref{assumption2} we have that both $x$ and $L_p(v)$ are leaves (since there are back-edges that stem from them), and therefore they are not related as ancestor and descendant. Thus, there are distinct children $c$ and $c'$ of $z$, such that $x$ is a descendant of $c$, and $L_p(v)$ is a descendant of $c'$. Since both $x$ and $L_p(v)$ are descendants of $v$, we have that $z$ is also a descendant of $v$. Thus, every descendant of $z$ is a descendant of $v$. In particular, this implies that $x$ is one of the descendants of $z$ with the lowest $l_1$ point.

Now notice that $c$ is the first child of $z$ on $T_\mathit{lowInc}$. To see this, let us assume the contrary. Then, due to the sorting of the children lists on $T_\mathit{lowInc}$, we have that the first child $d$ of $z$ on $T_\mathit{lowInc}$ has $\mathit{low}(d)\leq\mathit{low}(c)$. Then, there is a descendant $x'$ of $d$ that provides the $\mathit{low}$ point of $d$, and thus we have $l_1(x')=\mathit{low}(d)$. But then we have $l_1(x')=\mathit{low}(d)\leq\mathit{low}(c)=l_1(x)$, which is impossible due to the minimality of $x$. (I.e., the problem is that $x'<x$ on $T_\mathit{lowInc}$, since $x'$ is a descendant of the first child of $z$, whereas $x$ is a descendant of a greater child of $z$ than the first.) Now, since $l_1(x)<l_1(L_p(v))$, and $x$ is a descendant of $v$, we have that $(x,l_1(x))$ is a back-edge in $B_p(v)$. But since $x$ is a descendant of the first child of $z$, whereas $L_p(v)$ is a descendant of a different child of $z$, we have $x<L_p(v)$, which contradicts the minimality in the definition of $L_p(v)$. Thus, we conclude that there is no descendant of $v$ whose $l_1$ point is lower than $l_1(L_p(v))$. 
\end{proof}


\section{Efficient computation of the DFS-based parameters}
\label{section:efficient_computation}

\subsection{Computing the $\mathit{high}^p_1$ and $\mathit{high}^p_2$ points}
\label{section:computeHigh}
In order to compute the $\mathit{high}^p_1$ and $\mathit{high}^p_2$ points, we use the same idea that was used in a previous work for computing the $\mathit{high}^p_1$ points \cite{DBLP:conf/isaac/GeorgiadisK20} (which were being referred to there as ``$\mathit{high}_p$"), but we have to slightly extend it so that it also provides the $\mathit{high}^p_2$ points as well.

So let us first describe the general idea, and then explain how it can be made to run efficiently. We initialize all $\mathit{high}^p_1$ and $\mathit{high}^p_2$ points as $\bot$. Now the idea is to process all vertices $y$ in a bottom-up fashion. Then, for every vertex $y$ that we process, we consider all back-edges of the form $(x,y)$. For every such back-edge, we ascend the path $T[x,y]$, and we process every vertex $z$ that we meet, until it satisfies $p(z)= y$ (in which case we terminate the processing of $(x,y)$). For every such vertex $z$, we set $\mathit{high}^p_1(z)\leftarrow y$, if $\mathit{high}^p_1(z)=\bot$. Otherwise, if $\mathit{high}^p_1(z)$ has been computed and it has $\mathit{high}^p_1(z)\neq y$, then we set $\mathit{high}^p_2(z)\leftarrow z$, if $\mathit{high}^p_2(z)=\bot$. This method is shown in Algorithm~\ref{algorithm:high_naive}. It has a greedy nature, and it should be clear that it correctly computes all $\mathit{high}^p_1$ and $\mathit{high}^p_2$: this is an easy consequence of the fact that we process all back-edges in decreasing order w.r.t. their lower endpoint.

The problem with Algorithm~\ref{algorithm:high_naive} is obvious: its running time can be as large as $\Omega(n^2)$. The reason for this is, that as we process the back-edges $(x,y)$, we may have to ascend the same segments of vertices an excessive amount of times, whereas the vertices that are contained in them have long had their $\mathit{high}^p_1$ and/or $\mathit{high}^p_2$ points computed. Thus, we must find a way to bypass at once the segments of vertices that have their $\mathit{high}$ points computed. But this can be done easily, precisely because the sets of vertices that have their $\mathit{high}$ points computed, at any point during the course of the algorithm, form a collection of subtrees of the DFS tree on which we work. Thus, we can rely on a disjoint-set union (DSU) data structure. Then, once we reach a vertex that has its $\mathit{high}$ point(s) computed, we can immediately ask for the root of the subtree that contains it, and has the property that all the vertices on it have their $\mathit{high}$ point(s) computed, and then go directly to that root. That's the general idea, but now we must explain the details precisely.

First, let us explain the usage of the DSU data structure. This maintains a partition of the vertex set, where initially all vertices are placed in singletons. We will maintain throughout the following two properties:

\begin{enumerate}[label={(\arabic*)}]
\item{Every set maintained by the DSU is a subtree of the DFS tree.}
\item{All vertices on a subtree maintained by the DSU have their $\mathit{high}^p_1$ and $\mathit{high}^p_2$ points correctly computed, except for the root, which has $\mathit{high}^p_2=\bot$ (and possibly $\mathit{high}^p_1=\bot$).}
\end{enumerate}

\begin{algorithm}[h!]
\caption{\textsf{A correct but inefficient method to compute the $\mathit{high}^p_1$ and $\mathit{high}^p_2$ points}}
\label{algorithm:high_naive}
\LinesNumbered
\DontPrintSemicolon

\ForEach{vertex $v$}{
  $\mathit{high}^p_1[v]\leftarrow\bot$\;
  $\mathit{high}^p_2[v]\leftarrow\bot$\;
}

\For{$y\leftarrow n$ to $y=1$}{
  \ForEach{back-edge $(x,y)$}{
    $z\leftarrow x$\;
    \While{$p(z)\neq y$}{
      \If{$\mathit{high}^p_1[z]=\bot$}{
        $\mathit{high}^p_1[z]\leftarrow y$\;
      }
      \ElseIf{$\mathit{high}^p_1[z]\neq y$ \textbf{and} $\mathit{high}^p_2[z]=\bot$}{
        $\mathit{high}^p_2[z]\leftarrow y$\;
      }
      $z\leftarrow p(z)$\;
    }
  }
}
\end{algorithm}

The operations supported by the DSU are the following\footnote{Actually, we will also need one extra, more technical, assumption, which is discussed in the next footnote.}:

\begin{itemize}
\item{$\mathtt{makeset}(v)$: create a singleton that consists of $v$.}
\item{$\mathtt{find}(v)$: return a vertex that is contained in the set maintained by the DSU that contains $v$.}
\item{$\mathtt{unite}(v,u)$: unite the sets maintained by the DSU that contain $v$ and $u$.}
\end{itemize}

Since we will maintain throughout the property that the sets maintained by the DSU are subtrees of the DFS tree (all of whose vertices, except for the root, have their $\mathit{high}$ points computed), we would prefer the operation $\mathtt{find}(v)$ to return the root of the subtree maintained by the DSU that contains $v$. To do this, we associate with every vertex $v$ a pointer $\mathit{representative}$, which is initialized as $v.\mathit{representative}\leftarrow v$. Then it is sufficient to make sure that the representative of the result of $\mathtt{find}(v)$ points precisely to the root of the set that contains $v$. Thus, we also have to take care to update the representative after a call to $\mathtt{unite}(v,u)$, since this call may change internally the value returned by $\mathtt{find}(v)$\footnote{Notice that, theoretically, a DSU data structure (as we have stated its guarantees) may change the answer to many $\mathtt{find}$ operations, after any call to the data structure -- either $\mathtt{find}$ or $\mathtt{unite}$. For example, two consequtive $\mathtt{find}(v)$ calls may provide different answers, although both of them are correct (i.e., both of them are elements of the set that contains $v$). This, of course, depends on the particular implementation of the DSU data structure that one uses. We will make the simplifying assumption (which is satisfied by the implementation in \cite{DBLP:journals/jcss/GabowT85} that we cite), that the answers provided by $\mathtt{find}$ change only after a call to $\mathtt{unite}$, and that only the answers to $\mathtt{find}$ on the new set that was formed change, and that they are the same on every vertex of that set. In other words, the assumption is the following. After any call $\mathtt{find}(v)$, the answer provided by any $\mathtt{find}(w)$ has not changed. (I.e., $\mathtt{find}(w)$ returns the same vertex either before or after the call to $\mathtt{find}(v)$, as long as no $\mathtt{unite}$ intervenes.) And if we call $\mathtt{unite}(v,u)$, then $\mathtt{find}(w)$ stays the same for any $w$ that does not belong to the set $S$ that now contains $u$ and $v$, whereas it may have changed for $w\in S$, but now all the answers provided by $\mathtt{find}$ on $S$ are the same.}. (This is precisely what is done in Lines~\ref{line:firstunite} and \ref{line:nextfind} of Algorithm~\ref{algorithm:high_linear}.)

\begin{algorithm}[h!]
\caption{\textsf{Compute the $\mathit{high}^p_1$ and $\mathit{high}^p_2$ points in linear time}}
\label{algorithm:high_linear}
\LinesNumbered
\DontPrintSemicolon

\ForEach{vertex $v$}{
  $\mathit{high}^p_1[v]\leftarrow\bot$\;
  $\mathit{high}^p_2[v]\leftarrow\bot$\;  
}

initialize the DSU data structure of \cite{DBLP:journals/jcss/GabowT85} on the DFS tree\;

\ForEach{vertex $v$}{
  $\mathtt{makeset}(v)$\;
  $v.\mathit{representative}\leftarrow v$\;
}

\For{$y\leftarrow n$ to $y=1$}{
  \ForEach{back-edge $(x,y)$}{
    $z\leftarrow \mathtt{find}(x).\mathit{representative}$\;
    \While{$p(z)\neq y$}{
    \label{line:while_high}
      \If{$\mathit{high}^p_1[z]=\bot$}{
        $\mathit{high}^p_1[z]\leftarrow y$\;
        $z\leftarrow\mathtt{find}(p(z)).\mathit{representative}$\;
        \textbf{continue}\;
      }
      \ElseIf{$\mathit{high}^p_1[z]= y$}{
        \textbf{break}\;
      }
      $\mathit{high}^p_2[z]\leftarrow y$\;
      $\mathit{next}\leftarrow\mathtt{find}(p(z)).\mathit{representative}$\;
      $\mathtt{unite}(z,p(z))$\;
      \label{line:firstunite}
      $\mathtt{find}(z).\mathit{representative}\leftarrow\mathit{next}$\;
      \label{line:nextfind}
      $z\leftarrow \mathit{next}$\;
    }
  }
}
\end{algorithm}

In order to provide a linear-time implementation of the method shown in Algoritmh~\ref{algorithm:high_naive}, we distinguish four cases for every vertex $z$ on the path $T[x,y]$, with $p(z)\neq y$, during the processing of a back-edge $(x,y)$: $(i)$ $\mathit{high}^p_1(z)=\bot$, $(ii)$ $\mathit{high}^p_1(z)=y$, $(iii)$ $\mathit{high}^p_1(z)\notin\{\bot,y\}$ and $\mathit{high}^p_2(z)=\bot$, and $(iv)$ $\mathit{high}^p_2(z)\neq\bot$. (Notice that cases $(i)$ to $(iv)$ are mutually exclusive, and exhaust all possibilities. Specifically, in cases $(i)$ and $(ii)$ we have $\mathit{high}^p_2(z)=\bot$, due to the processing of the back-edges in decreasing order w.r.t. their lower endpoint.) Some things are obvious here. In case $(i)$, we just have to set $\mathit{high}^p_1(z)\leftarrow y$, and then move on to the next vertex. In case $(iii)$, we have to set $\mathit{high}^p_2(z)\leftarrow y$, and then move on to the next vertex. In case $(iv)$, we can bypass $z$, and a whole segment of consecutive ancestors of $z$ that have $\mathit{high}^p_2\neq\bot$. The interesting thing is that, in case $(ii)$, it is sufficient to terminate the processing of $(x,y)$. This is because, we can see inductively that, at the time we reach a vertex $z$ with $\mathit{high}^p_1(z)=y$, during the processing of $(x,y)$, we have properly processed all ancestors of $z$, that include $z$, up to the child of $y$ in the direction of $x$ (or equivalently: to the direction of $z$). To be specific, since we have $\mathit{high}^p_1(z)=y$, we must have met another back-edge $(x',y)$ before, which also satisfies $z\in T[x',y]$, and therefore, since we always follow the same method, we have already processed all ancestors of $z$, up to the child of $y$ in the direction of $z$.

This essentially explains the idea behind Algorithm~\ref{algorithm:high_linear}. By keeping in mind that, at any point during the course of the algorithm, the sets maintained by the DSU form subtrees of the DFS tree, all of whose vertices have their $\mathit{high}^p_1$ and $\mathit{high}^p_2$ points computed, except the root, which still has $\mathit{high}^p_2=\bot$, we can easily see the correctness, and the fact that we perform $O(n+m)$ calls in total to the DSU operations. Since the calls to $\mathtt{unite}$ are predetermined by the DFS tree, we can use the implementation by Gabow and Tarjan~\cite{DBLP:journals/jcss/GabowT85}, to achieve linear time in total. The main idea to see the bound in the running time is to observe that, whenever we enter the \textbf{while} loop in Line~\ref{line:while_high} of Algorithm~\ref{algorithm:high_linear}, we have that either $z$ will have one of its $\mathit{high}^p_1$ and $\mathit{high}^p_2$ points converted from $\bot$ to $y$, or we have $\mathit{high}^p_1(z)=y$, in which case we will break the loop. Our result is summarized in the following:

\begin{proposition}
Algorithm~\ref{algorithm:high_linear} correctly computes the $\mathit{high}^p_1$ and $\mathit{high}^p_2$ points of all vertices in linear time.
\end{proposition}

\subsection{The leftmost and rightmost skipping points}
\label{section:skippingPoints}

\begin{problem}
\label{problem:leftmost_skipping}
\normalfont
Let $T$ be a DFS tree of a graph $G$. Suppose that for every vertex $v$ of $G$ we associate a vertex $z_v$ or $\bot$. Then the goal is to find, for every vertex $v$, the minimum (w.r.t. $T$) descendant $x$ of $v$ for which there is a back-edge of the form $(x,y)$ with $y<p(v)$ and $y\neq\ z_v$, or report that no such vertex $x$ exists for $v$. We call this $x$ \emph{the leftmost point of $v$ that skips $z_v$}. (See Figure~\ref{figure:skipping}.) Similarly, by considering the maximum such $x$, we define \emph{the rightmost point of $v$ that skips $z_v$}. 
\end{problem}
\begin{remark}
\normalfont
Notice that, if for every vertex $v$ we let $z_v=\bot$, then Problem~\ref{problem:leftmost_skipping} coincides with that of computing the leftmost and the rightmost points, $L_p(v)$ and $R_p(v)$, for every vertex $v$.
\end{remark}

\begin{figure}[h!]\centering
\includegraphics[trim={1.3cm 22cm 0 0cm}, clip=true, width=1.2\linewidth]{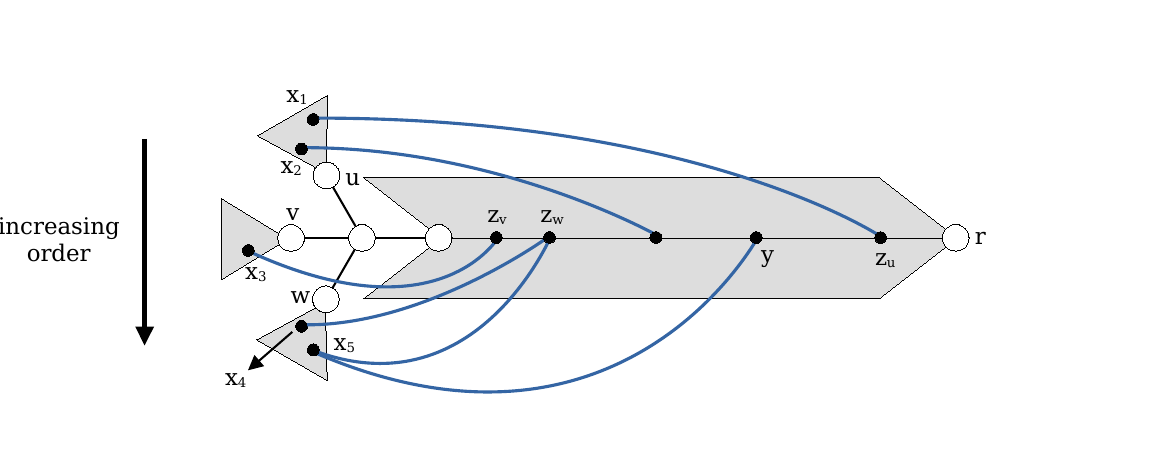}
\caption{\small{An illustration of an instance of Problem~\ref{problem:leftmost_skipping}, for the leftmost skipping points, which we denote here as $\widetilde{L}$. This figure depicts all back-edges from $B_p(u)$, $B_p(v)$ and $B_p(w)$. Then, we see that (for this particular instance of the skipping points problem), we have $\widetilde{L}(u)=x_2$, $\widetilde{L}(v)=\bot$, and $\widetilde{L}(w)=x_5$. Notice that, in particular, although there is a back-edge $(x_5,z_w)$, there is also a back-edge $(x_5,y)\in B_p(w)$ with $y\neq z_w$, and this is why $\widetilde{L}(w)=x_5$.}}\label{figure:skipping}
\end{figure}

The idea for computing the leftmost and the rightmost skipping points of all vertices is the same as that for computing the $\mathit{high}^p_1$ and $\mathit{high}^p_2$ points. The difference is that here we process the back-edges in increasing or decreasing order (depending on whether we want to compute the leftmost or the rightmost skipping points, respectively) w.r.t. their \emph{higher} endpoints. From now on, we will just refer to the ``skipping points" of the vertices, and whether we mean ``the leftmost skipping point" or ``the rightmost skipping point" depends on whether we process the back-edges in increasing or decreasing order, respectively, w.r.t. their higher endpoints.

Now, for every back-edge $(x,y)$ that we process, we assign greedily $x$ as the skipping point for as many vertices $v$ on $T[x,y]$ as possible, until we meet the vertex $v$ with $p(v)=y$, in which case the processing of $(x,y)$ terminates. In order for this greedy algorithm to be correct, we have to determine whether a vertex $v\in T[x,y]$ has indeed $x$ as its skipping point. Due to the processing of the back-edges in monotonic order w.r.t. their higher endpoints, notice that $x$ is indeed the skipping point of $v$ if and only if: $(1)$ the skipping point of $v$ has not already been computed, and $(2)$ $y\neq z_v$.

Thus, in order to get an efficient algorithm, the idea is to try to avoid, as much as possible, every $v$ on $T[x,y]$ with the property that either its skipping point has been computed, or $z_v=y$. In order to avoid completely meeting the vertices $v$ on $T[x,y]$ whose skipping point has already been computed, we use a DSU data structure as in Section~\ref{section:computeHigh} for the computation of the $\mathit{high}^p_1$ and $\mathit{high}^p_2$ points. The idea in using this data structure is to maintain subtrees of $T$ with the property that all vertices in them have their skipping point computed, except for the root, whose skipping point has not yet been computed. Thus, whenever we start to ascend the segment $T[x,y]$, we first find the root of the subtree that contains $x$, and each time we just go directly to such roots. Now, the first time we meet a vertex $v$ (during the processing of any back-edge), we definitely have that the skipping point of $v$ has not been computed, but it may be that $z_v=y$, where $y$ is the lower endpoint of the first back-edge during the processing of which we met $v$. (If $z_v\neq y$, then we have happily found the skipping point of $v$, i.e., this is the higher endpoint of the back-edge under process, and so we unite $v$ with the DSU subtree of its parent, so that we won't have to meet $v$ again.) If $z_v=y$, then we give a status \emph{marked} for $v$, and we move to the next vertex (i.e., to the root of the DSU subtree of $p(v)$). The usefulness of marking is the following. The next time we meet $v$ again, during the processing of a back-edge $(x',y')$, there are two possibilities: either $z_v=y'$, or not. In the second case, we have that the skipping point of $v$ is $x'$, and we are happily done with $v$. Otherwise, it is sufficient to terminate the processing of $(x',y')$. This is because, during the processing of the previous back-edge $(x,y)$, we have done all we could with the vertices on the segment $T[v,y]$: i.e., either we assigned them their correct skipping points, or we have marked them. But every marked vertex $u$ on $T[v,y]$ with $z_u=y$ can be avoided during the processing of $(x',y')$, because $y'=z_v=y$. Thus, it is correct to terminate the processing of a back-edge when we meet a marked vertex $v$ whose $z_v$ is the lower endpoint of that back-edge. The implementation of this procedure is shown in Algorithm~\ref{algorithm:skipping}. It should be clear that our discussion establishes the following:

\begin{proposition}
\label{proposition:skipping}
Suppose that for every vertex $v$ we associate a vertex $z_v$ (which may be $\bot$). Then, Algorithm~\ref{algorithm:skipping} correctly computes the leftmost skipping point of $v$ that skips $z_v$, for every vertex $v$, in linear time in total. If we want to compute the rightmost such points, then we can just replace Line~\ref{line:skippingFor} with ``\textbf{for} $x\leftarrow n$ to $x=2$". 
\end{proposition}

\begin{algorithm}[h!]
\caption{\textsf{Compute the leftmost skipping points in linear time}}
\label{algorithm:skipping}
\LinesNumbered
\DontPrintSemicolon

\ForEach{vertex $v$}{
  $\mathit{leftmost}[v]\leftarrow\bot$\;
}

initialize the DSU data structure of \cite{DBLP:journals/jcss/GabowT85} on the DFS tree\;

\ForEach{vertex $v$}{
  $\mathtt{makeset}(v)$\;
  $v.\mathit{representative}\leftarrow v$\;
  $v.\mathit{marked}\leftarrow$ \textbf{false}\;
}

\For{$x\leftarrow 2$ to $x=n$}{
\label{line:skippingFor}
  \ForEach{back-edge $(x,y)$}{
    $v\leftarrow \mathtt{find}(x).\mathit{representative}$\;
    \While{$p(v)\neq y$}{
      \If{$v.\mathit{marked}=$ \textbf{false}}{
        \If{$z_v\neq y$}{
          $\mathit{leftmost}[v]\leftarrow x$\;
          $\mathit{next}\leftarrow\mathtt{find}(p(v)).\mathit{representative}$\;
          $\mathtt{unite}(v,p(v))$\;
          $\mathtt{find}(v).\mathit{representative}\leftarrow\mathit{next}$\;
          $z\leftarrow\mathit{next}$\;
        }
        \Else{
          $v.\mathit{marked}\leftarrow$ \textbf{true}\;
          $v\leftarrow \mathtt{find}(p(v)).\mathit{representative}$\;
        }
      }
      \Else{
      \If{$z_v\neq y$}{
          $\mathit{leftmost}[v]\leftarrow x$\;
          $\mathit{next}\leftarrow\mathtt{find}(p(v)).\mathit{representative}$\;
          $\mathtt{unite}(v,p(v))$\;
          $\mathtt{find}(v).\mathit{representative}\leftarrow\mathit{next}$\;
          $z\leftarrow\mathit{next}$\;
        }
        \Else{
          \textbf{break}\;
        }
      }
    }
  }
}
\end{algorithm}

\subsection{Computing $\mathit{numLow}$ and $\mathit{numHigh}$}
\label{section:numlow}
Here we provide linear-time algorithms to compute $\mathit{numLow}(v)$ and $\mathit{numHigh}(v)$, for all vertices $v$. Recall that $\mathit{numLow}(v)$ is the number of back-edges $(x,y)\in B_p(v)$ with $y=\mathit{low}(v)$, and $\mathit{numHigh}(v)$ is the number of back-edges $(x,y)\in B_p(v)$ with $y=\mathit{high}_p(v)$. 

We can easily compute all $\mathit{numLow}(v)$ by processing the vertices in a bottom-up order, since we can rely on the recursive relation $\mathit{numLow}(v)=\mathit{numLow}(c_1)+\dots+\mathit{numLow}(c_k)+\#[\mbox{outgoing back-edges from }v\mbox{ of the form }(v,\mathit{low}(v))]$, where $c_1,\dots,c_k$ are the children of $v$ whose $\mathit{low}$ point is the same as that of $v$. (Of course, it may be that no such children exist.) The number $\#[\mbox{outgoing back-edges from }v\mbox{ of the form }(v,\mathit{low}(v))]$ is very easy to compute by processing the adjacency list of $v$, and it can be either $0$ or $1$ (since, for our purposes, it is sufficient to assume that we work on simple graphs).

The computation of $\mathit{numHigh}(v)$ is not so straightforward, because we cannot have such a convenient recursive relation as in the case of $\mathit{numLow}$.\footnote{The problem is that there might be a child $c$ of $v$ with $\mathit{high}_p(c)=p(v)$, and such that $T(c)$ provides back-edges of the form $(x,\mathit{high}_p(v))$. This would imply that $\mathit{high}^p_2(c)=\mathit{high}_p(v)$, and so we would need a counter for the back-edges from $B_p(c)$ of the form $(x,\mathit{high}^p_2(c))$, but in order to maintain that we would need a counter for the back-edges that reach the $\mathit{high}^p_3$ point, and so on.} However, we can still use a recursive procedure, just on a different tree. Specifically, for every vertex $z$, we consider the set $H(z)$ of all vertices $v$ with $\mathit{high}_p(v)=z$. We equip $H(z)$ with a forest structure as follows. If for a vertex $v\in H(z)$ there is no ancestor of $v$ (on $T$) that is also in $H(z)$, then we consider $v$ to be the root of a tree (of $H(z)$). Otherwise, we let the parent of $v$ (on $H(z)$) be the greatest proper ancestor of $v$ (on $T$) that is also in $H(z)$. 

Now we can use the forests $H(z)$, for all vertices $z$, in order to compute the values $\mathit{numHigh}(v)$, for all vertices $v$, as follows. (See also Algorithm~\ref{algorithm:numHigh}.) First, we initialize all $\mathit{numHigh}(v)\leftarrow 0$. Then, for every vertex $z$, we process the incoming back-edges to $z$ in decreasing order w.r.t. their higher endpoint. For every back-edge $(x,z)$ that we process, we find the greatest vertex $v\in H(z)$ such that $x$ is a descendant of $v$. If such a $v$ exists, then we increase the counter $\mathit{numHigh}(v)\leftarrow\mathit{numHigh}(v)+1$. Notice that, in order to find $v$, it is sufficient to process the vertices from $H(z)$ in decreasing order, starting from the last one that we processed. Let the last vertex from $H(z)$ that we processed be $v'$. Then, if $v'$ is an ancestor of $x$, we are done (we have $v=v')$. Otherwise, if $v'>x$, then we go to the next vertex from $H(z)$. We keep doing that, until we reach the greatest vertex $v'$ from $H(z)$ that is lower than, or equal to, $x$. Then, if $v'$ is an ancestor of $x$ we are done. Otherwise, we just move to the next incoming back-edge to $z$. Finally, after having processed a vertex $z$, we process every tree from $H(z)$ in a bottom-up fashion, and for every vertex $v$ that we meet we set $\mathit{numHigh}(v)\leftarrow\mathit{numHigh}(v)+\mathit{numHigh}(c_1)+\dots+\mathit{numHigh}(c_k)$, where $c_1,\dots,c_k$ are the children of $v$ on $H(z)$. The correctness and linear-time complexity of this procedure is established in the following:

\begin{algorithm}[h]
\caption{\textsf{Compute $\mathit{numHigh}(v)$, for all vertices $v$}}
\label{algorithm:numHigh}
\LinesNumbered
\DontPrintSemicolon

\ForEach{vertex $z$}{
  build the forest $H(z)$, that consists of all vertices $v$ with $\mathit{high}_p(v)=z$\;
  have the vertices in $H(z)$ sorted in decreasing order\;
  let $\mathit{currentVertex}(z)$ be the greatest vertex in $H(z)$\;
}

\ForEach{vertex $v$}{
  $\mathtt{numHigh}[v]\leftarrow 0$\;
}

\ForEach{vertex $z$}{
\label{line:processZforH(z)}
  \ForEach{back-edge $(x,z)$, with $x$ in decreasing order}{
    $v\leftarrow\mathit{currentVertex}(z)$\;
    \label{line1H(z)}
    \While{$v\neq\bot$ \textbf{and} $v>x$}{
      $v\leftarrow $ next vertex of $H(z)$\;
    }
    $\mathit{currentVertex}(z)\leftarrow v$\;
    \label{line2H(z)}
    \lIf{$v=\bot$}{\textbf{break}}
    \label{line:H(z)break}
    \lIf{$v$ is not an ancestor of $x$}{\textbf{continue}}
    $\mathtt{numHigh}[v]\leftarrow\mathtt{numHigh}[v]+1$\;
    \label{lineH(z)addunit}
  }
  \ForEach{vertex $c\in H(z)$, in decreasing order}{
  \label{line:forloopH(z)}
    \If{$c$ is not a root in $H(z)$}{
      let $v$ be the parent of $c$ on $H(z)$\;
      $\mathtt{numHigh}[v]\leftarrow\mathtt{numHigh}[v]+\mathtt{numHigh}[c]$\;
    }
  }
}

\end{algorithm}

\begin{proposition}
Algorithm~\ref{algorithm:numHigh} correctly computes $\mathit{numHigh}(v)$, for every vertex $v$. Furthermore, it has a linear-time implementation.
\end{proposition}
\begin{proof}
First, we shall explain how to build the forest $H(z)$, for every vertex $z$. To do that efficiently, we assume that $H(z)$ is sorted in decreasing order. In order to ensure that, we can sort all sets $H(z)$ simultaneously with bucket-sort. Thus, we process the vertices $v$ in decreasing order, and for every vertex $v$, we simply put it at the end of the list $H(\mathit{high}_p(v))$. (Notice that we use ``$H(z)$'' to denote both the list of vertices with $\mathit{high}_p=z$ sorted in decreasing order, and the forest of those vertices that we want to build; however, it will be clear which meaning is meant for $H(z)$ in everything that follows.)

Now the procedure for computing the $H(z)$ forests is shown in Algorithm~\ref{algorithm:H(z)forest}. This works by processing each vertex $z$ separately. For every vertex $z$, we process the vertices from $H(z)$ in decreasing order, and we maintain some vertices in a stack $S$ that satisfies the following two invariants, whenever we start processing a vertex $v$ from $H(z)$:

\begin{enumerate}[label={(\arabic*)}]
\item{Every vertex from $H(z)$ that has been taken out from $S$ has its parent in $H(z)$ correctly computed.}
\item{For every vertex $u$ still in $S$, no vertex from $H(z)$ that we have met so far is its parent in $H(z)$.}
\end{enumerate}

So now let us see what happens when it is the time to process a vertex $v$ from $H(z)$. Let $u$ be the top element in $S$ (which may be $\bot$). Then there are two possibilities: either $v$ is an ancestor of $u$ (on $T$), or it is not. In the first case, due to invariant $(2)$, and due to the fact that we process the vertices from $H(z)$ in decreasing order, we have that $v$ is the parent of $u$ in $H(z)$, because $v$ is indeed the greatest vertex in $H(z)$ that is an ancestor (on $T$) of $u$. Thus, it is proper to pop $u$ out of $S$, and invariant $(1)$ is maintained. In the other case, we have that $v$ is not an ancestor of $u$. Since we process the vertices from $H(z)$ in decreasing order, we have $v<u$. Thus, since $v$ is not an ancestor of $u$, we have that $u\geq v+\mathit{ND}(v)$. Furthermore, since we process the vertices from $H(z)$ in decreasing order, we have that every other element $u'$ in $S$ is greater than $u$, and therefore $u'\geq v+\mathit{ND}(v)$. Thus, no vertex from $S$ is a descendant of $v$ (on $T$), and therefore it is correct to stop the search in $S$ for children (on $H(z)$) of $v$. Thus, by finally putting $v$ in $S$, we have that invariants $(1)$ and $(2)$ are satisfied. In the end, after we have processed all vertices from $H(z)$, we have that, due to invariant $(2)$, all vertices still in $S$ must be roots in the forest $H(z)$. This shows that Algorithm~\ref{algorithm:H(z)forest} is correct, and it is easy to see that it runs in $O(n)$ time.

Now we will show that Algorithm~\ref{algorithm:numHigh} correctly computes $\mathit{numHigh}(v)$, for every vertex $v$. Notice that this algorithm processes each vertex $z$ separately (in Line~\ref{line:processZforH(z)}), and we claim that, by the time it finishes processing $z$, every vertex from $H(z)$ has its $\mathit{numHigh}$ correctly computed.

In order to understand Algorithm~\ref{algorithm:numHigh} better, notice that every back-edge of the form $(x,z)$ contributes ``$1$" to the $\mathit{numHigh}$ value of every vertex $v$ that is an ancestor of $x$ and has $\mathit{high}_p(v)=z$. Thus, all the vertices for whose $\mathit{numHigh}$ value the back-edge $(x,z)$ makes a contribution are related as ancestor and descendant (on $T$), and so they have the same relation on $H(z)$. Thus, it is sufficient to find the greatest $v$ on $H(z)$ that is an ancestor of $x$ (if it exists), and then add a unit to the $\mathit{numHigh}$ value of all the ancestors of $v$ on $H(z)$. Equivalently, it is sufficient to just add $1$ to the $\mathit{numHigh}$ value (under computation) of $v$, and then propagate this unit up to all the ancestors of $v$ on $H(z)$. This is precisely the purpose of the \textbf{For} loop in Line~\ref{line:forloopH(z)}. The purpose of Lines~\ref{line1H(z)} to \ref{line2H(z)} is to find the greatest ancestor (on $T$) of $x$ in $H(z)$. This works because we process simultaneously the incoming back-edges $(x,z)$ to $z$ in decreasing order w.r.t. $x$, and the vertices $v$ from $H(z)$ in decreasing order. Thus, if the current vertex $v$ of $H(z)$ is greater than $x$, we keep moving to the next element of $H(z)$, until we reach the greatest element $v$ of $H(z)$ with $v \leq x$. If $v=\bot$, then all remaining back-edges of the form $(x,z)$ have $x$ lower than every element of $H(z)$, and therefore no element of $H(z)$ is an ancestor of $x$. Thus, it is correct to terminate the processing of those back-edges in Line~\ref{line:H(z)break}. Otherwise, if $v$ is an ancestor of $x$, then it is the greatest ancestor (on $T$) of $x$ in $H(z)$, and therefore in Line~\ref{lineH(z)addunit} we apply the idea of adding a unit to $\mathit{numHigh}(v)$ that we described earlier. This explains the correctness of our procedure, and it is easy to see that it runs in linear time.

\end{proof}  

\begin{algorithm}[h]
\caption{\textsf{Build the forest $H(z)$, for all vertices $z$}}
\label{algorithm:H(z)forest}
\LinesNumbered
\DontPrintSemicolon

\ForEach{vertex $z$}{
  let $S$ be an empty stack for storing vertices\;
  \ForEach{vertex $v$ in $H(z)$, in decreasing order}{
    \While{$S.\mathit{top}()$ is a descendant of $v$ (on $T$)}{
      $u\leftarrow S.\mathit{pop}()$\;
      let the parent of $u$ in $H(z)$ be $v$\;
    }
    $S.\mathit{push}(v)$\;
  }  
  \ForEach{vertex $v$ in $S$}{
    make $v$ a root in $H(z)$\;
  }
}

\end{algorithm}

\subsection{The leftmost and rightmost points that reach a segment}
\label{section:leftmostReachSegment}
Consider the following problem. Given a triple of vertices $(z,u,v)$ such that $z$ is a descendant of $u$, and $u$ is a descendant of $v$, we ask for the minimum and the maximum vertices, $x$ and $y$, respectively, that are descendants of $z$, with the property that there are back-edges of the form $(x,w)$ and $(y,w')$, with $w,w'\in T[u,v]$. We denote the vertices $x$ and $y$ as $L(z,u,v)$ and $R(z,u,v)$, respectively, and we call them \emph{the leftmost and the rightmost points that start from $T(z)$ and reach the segment $T[u,v]$}. In this section we will show how to compute a batch of queries for such points in linear time. However, in order to do that, we demand that the collection of queries satisfies a kind of nestedness that is captured by the following definition.

\begin{definition}[Nested Queries]
\label{definition:nestedqueries}
\normalfont
Let $N$ be a positive integer, and let $(z_i,u_i,v_i)$, with $i\in\{1,\dots,N\}$, be a collection of triples of vertices with $z_i$ a descendant of $u_i$, and $u_i$ a descendant of $v_i$, for every $i\in\{1,\dots,N\}$. Then we say that the collection $\{(z_i,u_i,v_i)\mid i\in\{1,\dots,N\}\}$ satisfies the nested property if: $(1)$ for every $i,j\in\{1,\dots,N\}$ with $i\neq j$, either $T[u_i,v_i]\cap T[u_j,v_j]=\emptyset$, or one of $T[u_i,v_i]$ and $T[u_j,v_j]$ strictly contains the other, and $(2)$ for every $i,j\in\{1,\dots,N\}$ with $T[u_i,v_i]\subset T[u_j,v_j]$, either
\begin{enumerate}[label={(\roman*)}]
\item{$z_i$ is an ancestor of $z_j$, or}
\item{for every $k\in\{1,\dots,N\}$ such that $T[u_j,v_j]\subseteq T[u_k,v_k]$, there is no back-edge $(x,y)$ with $x\in T(z_k)$ and $y\in T[u_i,v_i]$.}
\end{enumerate}
\end{definition}

In other words, property $(1)$ in Definition~\ref{definition:nestedqueries} says that the collection $\{T[u_i,v_i]\mid i\in\{1,\dots,N\}\}$ of sets of vertices constitutes a laminar family. Property $(2)$ is a little bit more complicated, but in essence it says that, if we have established that $T[u_i,v_i]$, for some $i\in\{1,\dots,N\}$, is contained in $T[u_j,v_j]$, for some $j\in\{1,\dots,N\}$, and $z_i$ is not an ancestor of $z_j$, then, for every $k\in\{1,\dots,N\}$ such that $T[u_k,v_k]$ contains $T[u_j,v_j]$ (and therefore $T[u_i,v_i]$), we can be certain that $L(z_k,u_k,v_k)$ and $R(z_k,u_k,v_k)$ do not provide a back-edge with lower endpoint in $T[u_i,v_i]$. The usefulness of that property will appear in the proof of the following: 

\begin{proposition}
\label{proposition:leftmostinsegment}
Let $N$ be a positive integer, and let $(z_i,u_i,v_i)$, with $i\in\{1,\dots,N\}$ be a collection of triples of vertices that satisfies Definition~\ref{definition:nestedqueries}. Then, we can compute $L(z_i,u_i,v_i)$ and $R(z_i,u_i,v_i)$, for all $i\in\{1,\dots,N\}$, in $O(n+m)$ time in total.
\end{proposition}
\begin{proof}
First, we note that ``$N$" does not appear in the running time, because $N< 2n$, due to the laminarity of the collection $\mathcal{S}=\{T[u_i,v_i]\mid i\in\{1,\dots,N\}\}$ of sets of vertices, which is implied by $(1)$ of Definition~\ref{definition:nestedqueries}.

We will refer to every triple $(z_i,u_i,v_i)$, for $i\in\{1,\dots,N\}$, as a \emph{query}, and we will denote it as $Q_i$. A back-edge $(x,y)$ with the property $x\in T(z_i)$ and $y\in T[u_i,v_i]$ is called \emph{a candidate back-edge for $Q_i$}. Thus, $L(z_i,u_i,v_i)$ and $R(z_i,u_i,v_i)$ are given by the smallest and the greatest, respectively, higher endpoint of a candidate back-edge for $Q_i$.

As a consequence of the laminarity of $\mathcal{S}$, we have that, for every vertex $w$ that appears in a set from $\mathcal{S}$, there is a unique set $T[u_i,v_i]\in\mathcal{S}$ with minimum size such that $w\in T[u_i,v_i]$. (The ``size" here refers simply to the number of elements.) Then we call the triple $(z_i,u_i,v_i)$ the \emph{parent query} of $w$, and we denote it as $Q(w)$. If $w$ does not appear in a set from $\mathcal{S}$, then we let $Q(w):=\bot$.

As a further consequence of property $(1)$ of Definition~\ref{definition:nestedqueries}, we have that, for every $i\in\{1,\dots,N\}$, there is at most one  $j\in\{1,\dots,N\}$, such that $T[u_i,v_i]\subset T[u_j,v_j]$ with $T[u_j,v_j]$ having minimum size. If this $j$ indeed exists, and we also have that $z_i$ is an ancestor of $z_j$, then we call $Q_j$ \emph{the parent query} of $Q_i$. Otherwise, we consider $Q_i$ as a \emph{root query}.

Notice that this relation of parenthood that we have defined (between vertices and queries, and between the queries themselves), defines a collection of trees, which we call \emph{the forest of queries}. 
(One can imagine this forest as the usual forest that represents a laminar family of sets, where we have just cut off some parent edges.) In particular, the notion of a ``parent query" naturally extends to that of an \emph{ancestor query}, and therefore to a \emph{descendant query}. First we will show how to use this forest in order to compute all $L(z_i,u_i,v_i)$ and $R(z_i,u_i,v_i)$, for $i\in\{1,\dots,N\}$, in $O(n+m)$ time, and then we will show how to construct the forest itself in $O(n)$ time.

In order to compute all $L(z_i,u_i,v_i)$, for $i\in\{1,\dots,N\}$, we process all back-edges in increasing order w.r.t. their higher endpoint. For every back-edge $(x,y)$ that we process, the idea is to find greedily all the queries for which it is a candidate back-edge, and for which the answer is not yet computed. Thus, due to the processing in increasing order w.r.t. the higher endpoints, $(x,y)$ provides the answer to every such query (which now is marked as answered). In order to compute all $R(z_i,u_i,v_i)$, for $i\in\{1,\dots,N\}$, the idea is precisely the same, except that we process the back-edges in decreasing order w.r.t. their higher endpoint.

The forest of queries comes into play in the following way. Notice that $(x,y)$ is a candidate back-edge for a query $Q_i$, for $i\in\{1,\dots,N\}$, only if $y\in T[u_i,v_i]$. In particular, $(x,y)$ has the potential to be a candidate back-edge for $Q(y)$. (Let us assume that $Q(y)\neq\bot$, because otherwise $(x,y)$ is not a candidate back-edge for any query $Q_i$, with $i\in\{1,\dots,N\}$, because, by definition, $y$ is not contained in any $T[u_i,v_i]$, for $i\in\{1,\dots,N\}$.) Now we claim the following:\\

$(*)$ If $(x,y)$ is a candidate back-edge for a query $Q_i$, for $i\in\{1,\dots,N\}$, then $Q_i$ is an ancestor query of $Q(y)$, and $(x,y)$ is a candidate back-edge for every ancestor query of $Q(y)$ which is a descendant query of $Q_i$. (In other words, if $(x,y)$ is a candidate back-edge for a query, then it is definitely a candidate for $Q(y)$, and it is so only for a segment of queries on the tree that contains $Q(y)$.)\\

Notice why $(*)$ is useful: because it tells us that, in order to find all the queries for which $(x,y)$ is candidate back-edge, it is sufficient to start from $Q(y)$, and traverse all the ancestor queries of $Q(y)$, until we reach one for which $(x,y)$ is not a candidate back-edge, or the root query of the tree that contains $Q(y)$ (which may or may not have $(x,y)$ as a candidate back-edge).

Now we will prove $(*)$. First, let us fix some notation. Let $t$ be the index in $\{1,\dots,N\}$ with $Q_t=Q(y)$, let $\mathcal{T}$ be the query tree that contains $Q_t$, and let $s$ be the index in $\{1,\dots,N\}$ such that $Q_s$ is the root query of $\mathcal{T}$. Now, due to the laminarity of $\mathcal{S}$, notice that every query $Q_i$, for $i\in\{1,\dots,N\}$, that has $(x,y)$ as a candidate back-edge, must satisfy $T[u_t,v_t]\subseteq T[u_i,v_i]$. Notice that this property is satisfied by all ancestor queries of $Q_t$ (on $\mathcal{T}$), and that these are all the queries on $\mathcal{T}$ that satisfy this property, but there may be other queries too, that lie on other trees, that satisfy this property. So let us first prove that it is impossible that $(x,y)$ is a candidate back-edge for a query that does not lie on $\mathcal{T}$. Consider the root query $Q_s$ of $\mathcal{T}$. There are two possibilities here (that made $Q_s$ a root query): either $(a)$ there is no index $i\in\{1,\dots,N\}$ such that $T[u_s,v_s]\subset T[u_i,v_i]$, or $(b)$ the smallest $T[u_i,v_i]$ with $T[u_s,v_s]\subset T[u_i,v_i]$ has the property that $z_i$ is not a descendant of $z_s$. In case $(a)$, we can infer from property $(1)$ of Definition~\ref{definition:nestedqueries} that, for every $i\in\{1,\dots,N\}$ such that $y\in T[u_i,v_i]$, we have that $Q_i$ must lie on $\mathcal{T}[Q_t,Q_s]$. In case $(b)$, we can infer from property $(2)$ of Definition~\ref{definition:nestedqueries} that, for every $j\in\{1,\dots,N\}$ such that $T[u_s,v_s]\subset T[u_j,v_j]$, $(x,y)$ cannot be a candidate back-edge for $Q_j$ (because $y\in T[u_s,v_s]$). Notice that these are precisely all the queries with $y\in T[u_j,v_j]$ that are not contained on $\mathcal{T}$. This shows that all queries for which $(x,y)$ is a candidate back-edge lie on $\mathcal{T}$. Now let $i$ be an index in $\{1,\dots,N\}$ such that $Q_i$ has $(x,y)$ as a candidate back-edge. We have already shown that $Q_i$ is on $\mathcal{T}$, and that $T[u_t,v_t]\subseteq T[u_i,v_i]$. Thus, $Q_i$ is an ancestor query of $Q_t$. And now let $j$ be an index in $\{1,\dots,N\}$ such that $Q_j$ is an ancestor of $Q_t$ and a descendant of $Q_i$. Since $Q_j$ is an ancestor of $Q_t$ we have that $T[u_t,v_t]\subseteq T[u_j,v_j]$, and therefore $y\in T[u_j,v_j]$. And since $Q_j$ is a descendant of $Q_i$, we have that $z_i$ is a descendant of $z_j$. Since $(x,y)$ is a candidate back-edge for $Q_i$, we have that $x$ is a descendant of $z_i$, and therefore it is a descendant of $z_j$. This shows that $(x,y)$ is a candidate back-edge for $Q_j$, and thus we have established $(*)$.

Now, since we have demonstrated how to use $(*)$ in order to find all queries for which $x$ provides the answer, it remains to explain how to achieve linear time in total. Notice that the problem appears since we may have to traverse some queries on the forest of queries multiple times. However, we can avoid this problem by using precisely the same method that computes efficiently the $\mathit{high}_p$ points (see Section~\ref{section:computeHigh}). Specifically, we use a disjoint-set union (DSU) data structure to maintain sets of queries that have already been answered. Since we always traverse segments of the trees in order to provide the answer, the DSU maintains subtrees, and the sequence of unions is predetermined by the query forest. Thus, we can use the implementation by Gabow and Tarjan~\cite{DBLP:journals/jcss/GabowT85}, in order to achieve $O(n+m)$ time in total. 

Finally, we shall explain how to construct the forest of queries. Due to the laminarity of $\mathcal{S}$, it is easy to do that in a greedy manner, by processing the sets from $\mathcal{S}$ in increasing order w.r.t. their size. The procedure that we follow is shown in Algorithm~\ref{algorithm:query_forest}. So, the first step is to compute the sizes of the sets in $\mathcal{S}$. Since every set in $\mathcal{S}$ has the form of a segment $T[u,v]$, where $u$ is a descendant of $v$, we have that $|T[u,v]|=\mathit{depth}(u)-\mathit{depth}(v)+1$ (where the depth here is computed w.r.t. the DFS tree $T$). Then, we sort all segments from $\mathcal{S}$ in $O(n)$ time using bucket-sort. Now, for every $T[u,v]\in\mathcal{S}$ that we process, we do the following. We ascend the tree-path from $T[u,v]$, and for every vertex $w$ that we meet there are two possibilities: either $(a)$ $T[u,v]$ is the first segment from $\mathcal{S}$ that we meet such that $w\in T[u,v]$, or $(b)$ $w$ is an endpoint of a segment from $\mathcal{S}$ that we already met, and let $T[w,v']$ be the latest such segment that we met. (Notice that case $(b)$ is due to the laminarity of $\mathcal{S}$: i.e., since $w$ was met before, now that we meet it again it must necessarily be the highest endpoint of every previous segment that contained it.) 

Let $Q$ be the query that corresponds to $T[u,v]$. Now, in case $(a)$, we let $Q$ be the parent query of $w$. This is correct, due to the processing of the segments in increasing order w.r.t. their size. In case $(b)$, we have that the parent query of $w$ has already been computed. Moreover, due to the laminarity of the sets in $\mathcal{S}$, and due to their processing in increasing order w.r.t. their size, we have that every vertex in $T[w,v']$ has its parent query computed, and that for every set from $\mathcal{S}$ that is strictly included in $T[w,v']$ we have computed the parent query of its associated query (or determined that it is a root query). Thus, the segment $T[w,v']$ can be bypassed at once, and the next vertex for processing should be $p(v')$ (if it is still a descendant of $v$). Now let $Q=(z,u,v)$ and $Q'=(z',w,v')$ be the associated queries of $T[u,v]$ and $T[w,v']$, respectively. Then, if $z'$ is an ancestor of $z$, we have that $Q$ is the parent query of $Q'$. Otherwise, we mark $Q'$ as a root query. Since we process the segments from $\mathcal{S}$ in increasing order w.r.t. their size, we have that $T[u,v]$ is indeed the smallest segment from $\mathcal{S}$ that contains $T[w,v']$. Thus, it is not difficult to see that Algorithm~\ref{algorithm:query_forest} correctly computes the forest of queries in $O(n)$ time.
\end{proof} 

\begin{algorithm}[]
\caption{\textsf{Compute the forest of queries}}
\label{algorithm:query_forest}
\LinesNumbered
\DontPrintSemicolon

\SetKwInOut{Input}{input}
\SetKwInOut{Output}{output}

\Input{A set of triples $\mathcal{Q}=\{(z_i,u_i,v_i)\mid i\in\{1,\dots,N\}\}$ that satisfies Definition~\ref{definition:nestedqueries}.}
\Output{The forest of queries that corresponds to $\mathcal{Q}$, as defined in the proof of Proposition~\ref{proposition:leftmostinsegment}.}

\ForEach{$i\in\{1,\dots,N\}$}{
  let $Q_i:=(z_i,u_i,v_i)$\;
  $\mathit{size}(Q_i)\leftarrow\mathit{depth}(u_i)-\mathit{depth}(v_i)+1$\;
}

\ForEach{vertex $w$}{
  $\mathtt{parent\_query\_of}[w]\leftarrow\bot$\;
  $\mathtt{latest\_query}[w]\leftarrow\bot$\;
}
\ForEach{$Q\in\mathcal{Q}$}{
  $\mathtt{parent\_query\_of}[Q]\leftarrow\bot$\;
}

sort the elements from $\mathcal{Q}$ in increasing order w.r.t. their $\mathit{size}$\;

\For{$Q=(z,u,v)\in\mathcal{Q}$ in increasing order w.r.t. the $\mathit{size}$}{
  $w\leftarrow u$\;
  \While{$w\neq p(v)$}{
    \If{$\mathtt{parent\_query\_of}[w]=\bot$}{
      $\mathtt{parent\_query\_of}[w]\leftarrow Q$\;
      $\mathtt{latest\_query}[w]\leftarrow Q$ \tcp{this is actually needed only when $w=u$; otherwise, it makes no difference}
      $w\leftarrow p(w)$\;
    }
    \Else{
      let $Q':=\mathtt{latest\_query}[w]=(z',w,v')$\;
      \If{$z'$ is an ancestor of $z$}{
        $\mathtt{parent\_query\_of}[Q']\leftarrow Q$\;
      }
      \Else{
        mark $Q'$ as a root query\;
      }
      $\mathtt{latest\_query}[w]\leftarrow Q$ \tcp{the previous comment applies here too}
      $w\leftarrow p(v')$\;
    }    
  }
}

\ForEach{$Q\in\mathcal{Q}$}{
  \If{$\mathtt{parent\_query\_of}[Q]=\bot$}{
    mark $Q$ as a root query\;
  }
}

\end{algorithm}

It is important to note that the computation in Proposition~\ref{proposition:leftmostinsegment} can work on any permutation of the original DFS tree (and that is why we do not specify the DFS tree on which we work). Furthermore, notice that the nestedness of the queries that is captured in Definition~\ref{definition:nestedqueries} is an invariant across all permutations of a base DFS tree.

\begin{observation}
\label{observation:nestedqueries}
Let $\mathcal{Q}$ be a collection of triples of vertices that satisfy Definition~\ref{definition:nestedqueries} w.r.t. a DFS tree $T$. Then, $\mathcal{Q}$ satisfies Definition~\ref{definition:nestedqueries} w.r.t. any DFS tree $\widetilde{T}$ that is a permutation of $\mathcal{T}$.
\end{observation}
\begin{proof}
This is an easy consequence of the fact that the ancestry relation is the same for two DFS trees that are a permutation of each other.
\end{proof}

The reason that we need the above observation has to do with the structure of our arguments in what follows. Specifically, we first extract a collection of queries $\mathcal{Q}$ from a specific DFS tree $T$, on which we prove that Definition~\ref{definition:nestedqueries} is satisfied for $\mathcal{Q}$. But then, we need to have available the answers to those queries w.r.t. a possibly different DFS tree, but which is a permutation of $T$. (Notice that, although the satisfaction of Definition~\ref{definition:nestedqueries} is an invariant across all permutations of a DFS tree, the answers to the queries may be different, since they have to do with the higher endpoints of some back-edges.)

Now we will show how Proposition~\ref{proposition:leftmostinsegment} can help us in the computation of some parameters that we will need in later sections.

\begin{figure}[h!]\centering
\includegraphics[trim={2cm 23cm 0 0cm}, clip=true, width=1.1\linewidth]{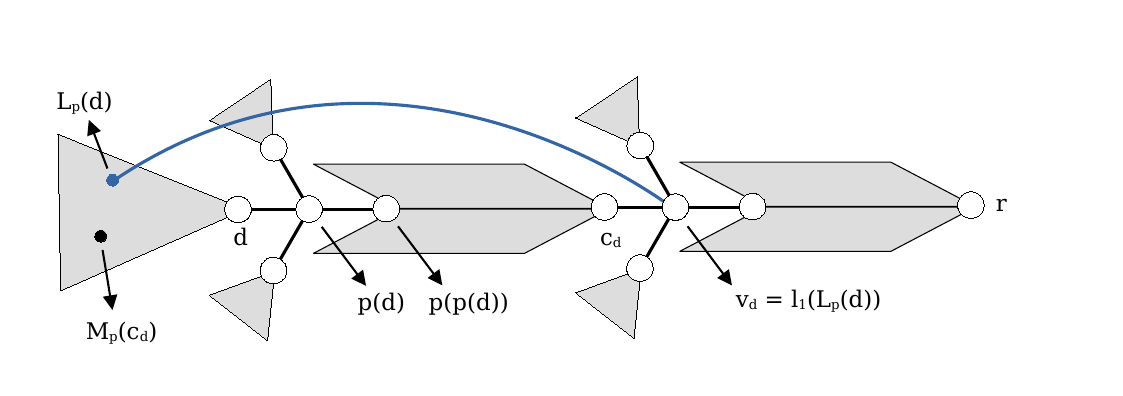}
\caption{\small{An example of a vertex $d$ that belongs to the set $U$ that appears in the statement of Lemma~\ref{lemma:segmentpointsL}. Thus, we have that $l_1(L_p(d))=v_d$, where $v_d$ is a proper ancestor of $p(p(d))$, and $M_p(c_d)$ is a descendant of $d$, where $c_d$ is the child of $v_d$ in the direction of $d$. The goal is to find the leftmost and the rightmost points that start from $T(d)$ and reach the segment $T[p(p(d)),c_d]$.}}\label{figure:segment1}
\end{figure}

\begin{lemma}
\label{lemma:segmentpointsL}
Let $T$ be a DFS tree, and let $\widetilde{T}$ be a permutation of $T$. Let $U$ be the collection of all vertices $d$ with the following property: $l_1(L_p(d))=v_d$ on $T$, where $v_d$ is a proper ancestor of $p(p(d))$, and $M_p(c_d)$ is a descendant of $d$, where $c_d$ is the child of $v_d$ in the direction of $d$. (See Figure~\ref{figure:segment1}.) Then, after a linear-time preprocessing, we can have available the points $L(d,p(p(d)),c_d)$ and $R(d,p(p(d)),c_d)$ on $\widetilde{T}$, for every $d\in U$. Furthermore, the lemma still holds if, instead of the ``$l_1(L_p(\cdot))$'' point in the definition of $U$, we use ``$l_1(R_p(\cdot))$''. 
\end{lemma}
\begin{proof}
First of all, notice that, for every $d\in U$, we have that the segment $T[p(p(d)),c_d]$ contains at least one vertex, since $v_d$ is a proper ancestor of $p(p(d))$, and therefore $c_d$ is an ancestor of $p(p(d))$.

Now, for every $d\in U$, we generate the triple $Q_d=(d,p(p(d)),c_d)$, and let $\mathcal{Q}$ be the collection of all those triples. Then, according to Proposition~\ref{proposition:leftmostinsegment} and Observation~\ref{observation:nestedqueries}, it is sufficient to establish that $\mathcal{Q}$ satisfies Definition~\ref{definition:nestedqueries}.

So let us suppose, for the sake of contradiction, that property $(1)$ of Definition~\ref{definition:nestedqueries} does not hold for $\mathcal{Q}$. This means that there are two distinct vertices $d$ and $d'$ in $U$ such that the segments $S_1=T[p(p(d)),c_d]$ and $S_2=T[p(p(d')),c_{d'}]$ intersect, but neither of them strictly includes the other. First, we distinguish the two cases $c_d=c_{d'}$ and $c_d\neq c_{d'}$. Let us consider the case $c_d=c_{d'}$. Then, since neither of $S_1$ and $S_2$ strictly includes the other, we have that $p(p(d))$ and $p(p(d'))$ are not related as ancestor and descendant. (Because, if e.g. $p(p(d))$ is an ancestor of $p(p(d'))$, then $S_1\subset S_2$.) This implies that $d$ and $d'$ are not related as ancestor and descendant either. But since $d,d'\in U$, we have that $M_p(c_d)=M_p(c_{d'})$ must be a descendant of both $d$ and $d'$, which is absurd. 

Thus, we have shown that $c_d\neq c_{d'}$. Then, we may assume w.l.o.g. that $c_d<c_{d'}$. Since $S_1$ and $S_2$ intersect, there is a vertex $w$ common in both $S_1$ and $S_2$. Then, $w$ is a descendant of both $c_d$ and $c_{d'}$. Thus, $c_d$ and $c_{d'}$ are related as ancestor and descendant, and therefore $c_d<c_{d'}$ implies that $c_d$ is a proper ancestor of $c_{d'}$. Furthermore, since $S_1$ and $S_2$ intersect, we have that $c_{d'}\in T[p(p(d)),c_d)$. Then, since $S_1$ does not contain $S_2$, we have that $p(p(d'))$ cannot be an ancestor of $p(p(d))$. Thus, we have two possibilities: either \textbf{(a)} $p(p(d'))$ is a proper descendant of a vertex on $T(p(p(d)),c_d)$, but not of $p(p(d))$, or \textbf{(b)} $p(p(d'))$ is a proper descendant of $p(p(d))$. (See Figure~\ref{figure:segment1ab}.) Let us consider case \textbf{(a)} first. Notice that in this case we have that $d$ and $d'$ are not related as ancestor and descendant. 
Since $d\in U$, we have that $M_p(c_d)$ is a descendant of $d$, which is a descendant of $p(p(d))$, which is a descendant of $c_{d'}$. Then, since $c_{d'}$ is a descendant of $c_d$, Lemma~\ref{lemma:Mp} implies that $M_p(c_{d'})$ is an ancestor of $M_p(c_d)$. But then we have that $M_p(c_d)$ is a common descendant of $d$ and $d'$, which is absurd.

Thus, we are left to consider case \textbf{(b)}. Here we will crucially utilize the fact that $l_1(L_p(d))=p(c_d)$ and $l_1(L_p(d'))=p(c_{d'})$ (which is implied from $d,d'\in U$). We have that $L_p(d)$ is a descendant of $d$, and therefore a descendant of $p(p(d))$, and therefore a descendant of $c_{d'}$ (since $c_{d'}\in T[p(p(d)),c_d)$). Since $l_1(L_p(d))=p(c_d)$, we have that $l_1(L_p(d))$ is a proper ancestor of $c_d$, and therefore a proper ancestor of $p(c_{d'})$ (since $c_d$ is a proper ancestor of $c_{d'}$). This shows that $(L_p(d),l_1(L_p(d)))\in B_p(c_{d'})$, and therefore $L_p(d)$ is a descendant of $M_p(c_{d'})$. Since $d'\in U$, we have that $M_p(c_{d'})$ is a descendant of $d'$, and therefore $L_p(d)$ is a descendant of $d'$. Then, since $L_p(d)$ is a common descendant of $d$ and $d'$, we have that $d$ and $d'$ are related as ancestor and descendant. Thus, since $p(p(d'))$ is a proper descendant of $p(p(d))$, we infer that $d'$ is a proper descendant of $d$. Notice that the back-edge $(L_p(d),l_1(L_p(d)))$ is also in $B_p(d')$, and therefore (by definition of $L_p(d')$) we have that $L_p(d')\leq L_p(d)$. Now, if $L_p(d')=L_p(d)$, then we have  $p(c_{d'})=l_1(L_p(d'))=l_1(L_p(d))=p(c_{d})$, which is impossible, because $p(c_{d'})$ and $p(c_d)$ are distinct (because $c_d$ is a proper ancestor of $c_{d'})$. Thus, we have $L_p(d')<L_p(d)$. Now, consider the back-edge $e=(L_p(d'),l_1(L_p(d')))$. We have that $L_p(d')$ is a descendant of $d'$, and therefore a descendant of $d$. Furthermore, we have that $l_1(L_p(d'))=p(c_{d'})$, and $p(c_{d'})$ is a proper ancestor of $c_{d'}$, which is a proper ancestor of $p(d)$ (since $c_{d'}\in T[p(p(d)),c_d)$). This shows that $e\in B_p(d)$. But then, by the definition of $L_p(d)$, we have $L_p(d)\leq L_p(d')$, a contradiction. 

Thus, we have established condition $(1)$ of Definition~\ref{definition:nestedqueries} for $\mathcal{Q}$. Now we will establish condition $(2)$ of Definition~\ref{definition:nestedqueries} for $\mathcal{Q}$. Thus, let $d$ and $d'$ be two distinct vertices from $U$ such that one of the segments $S_1=T[p(p(d)),c_d]$ and $S_2=T[p(p(d')),c_{d'}]$ strictly contains the other. Then, we may assume w.l.o.g. that $S_2\subset S_1$. This implies that $p(p(d'))\in S_1$, and $c_{d'}\in S_1$. We will show that condition $(i)$ of $(2)$ of Definition~\ref{definition:nestedqueries} holds for $(d,p(p(d)),c_d)$ and $(d',p(p(d')),c_{d'})$. 

So let us suppose, for the sake of contradiction, that $d'$ is not an ancestor of $d$. 
Since $p(p(d'))\in S_1$, we have that $p(p(d'))$ is an ancestor of $p(p(d))$, and therefore $d'$ cannot be a descendant of $d$. Now, since $c_{d'}\in S_1$, we have that $c_{d'}$ is an ancestor of $p(p(d))$, and therefore an ancestor of $d$, and therefore an ancestor of $M_p(c_d)$ (because, since $d\in U$, we have that $d$ is an ancestor of $M_p(c_d)$). Furthermore, since $c_{d'}\in S_1$, we have that $c_{d'}$ is a descendant of $c_d$. Thus, Lemma~\ref{lemma:Mp} implies that $M_p(c_{d'})$ is an ancestor of $M_p(c_d)$. Therefore, since $M_p(c_d)$ is a common descendant of $d$ and $M_p(c_{d'})$, we have that $d$ and $M_p(c_{d'})$ are related as ancestor and descendant. Since $d'\in U$, we have that $M_p(c_{d'})$ is a descendant of $d'$. Now, if $d$ is a descendant of $M_p(c_{d'})$, this implies that $d$ is a descendant of $d'$, which is impossible. Therefore, we have that $M_p(c_{d'})$ is a descendant of $d$. But then, since $d$ and $d'$ have $M_p(c_{d'})$ as a common descendant, they must be related as ancestor and descendant, which is also impossible. Thus, we have established condition $(2)$ of Definition~\ref{definition:nestedqueries} for $\mathcal{Q}$. 

\end{proof}

\begin{figure}[h!]\centering
\includegraphics[trim={0cm 23.5cm 0 0cm}, clip=true, width=1\linewidth]{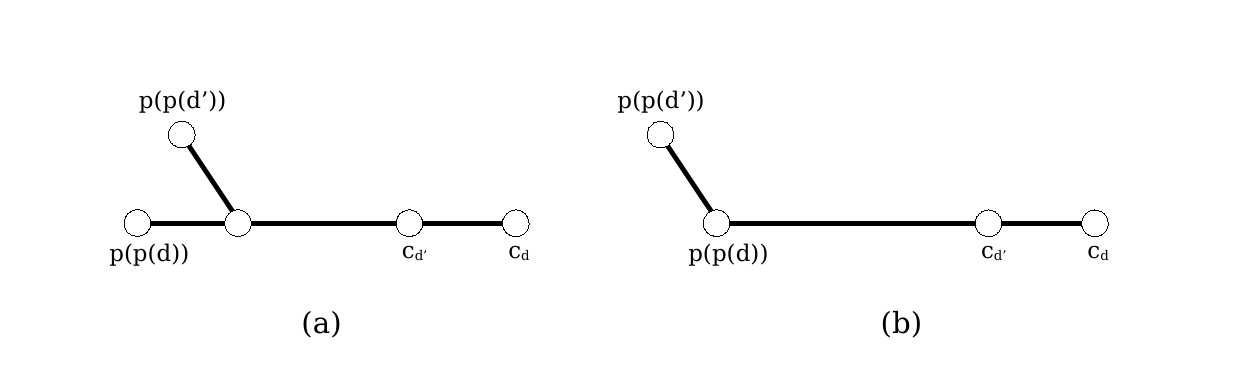}
\caption{\small{An illustration of the cases \textbf{(a)} and \textbf{(b)} that appear in the proof of Lemma~\ref{lemma:segmentpointsL}, demonstrating the ancestry relation between the vertices $c_d$, $c_{d'}$, $p(p(d))$, and $p(p(d'))$.}}\label{figure:segment1ab}
\end{figure}

\begin{figure}[h!]\centering
\includegraphics[trim={1.2cm 22cm 0 0.5cm}, clip=true, width=1.1\linewidth]{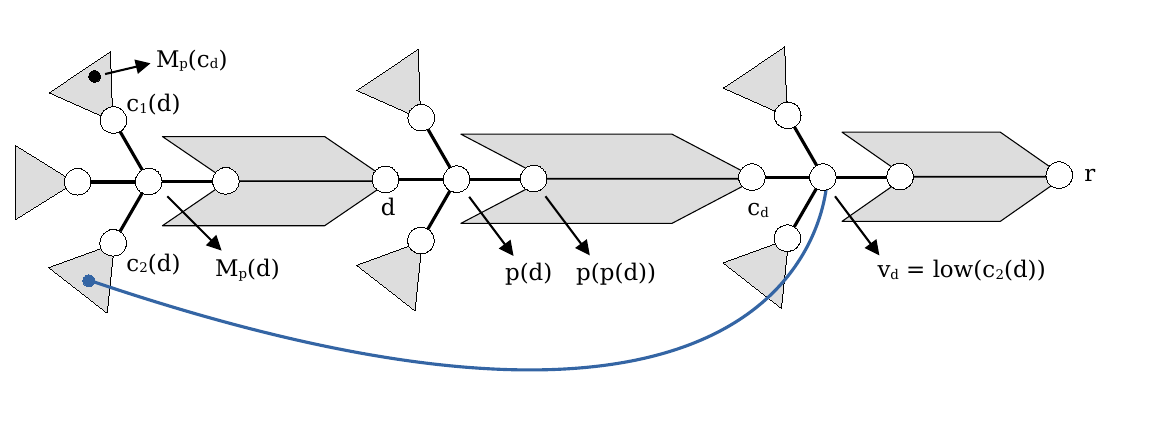}
\caption{\small{An example of a vertex $d$ that belongs to the set $U$ that appears in the statement of Lemma~\ref{lemma:segmentpointslow}. Thus, we have that $\mathit{low}(c_2(d))=v_d$, where $c_2(d)$ is a child of $M_p(d)$ that does not have the lowest $\mathit{low}$ point among the children of $M_p(d)$, and $v_d$ is a proper ancestor of $p(p(d))$. Furthermore, we have that $M_p(c_d)$ is a descendant of $c_1(d)$, where $c_d$ is the child of $v_d$ in the direction of $d$, and $c_1(d)$ is the child of $M_p(d)$ with the lowest $\mathit{low}$ point among the children of $M_p(d)$. The goal is to find the leftmost and the rightmost points that start from $T(d)$ and reach the segment $T[p(p(d)),c_d]$.}}\label{figure:segment2}
\end{figure}

\begin{lemma}
\label{lemma:segmentpointslow}
Let $T$ be a DFS tree, and let $U$ be the collection of all vertices $d$ with the following property (which is equivalent across all permutations of $T$): $M_p(d)$ has a unique child $c_1(d)$ with the lowest $\mathit{low}$ point among the children of $M_p(d)$, and there is another child $c_2(d)$ of $M_p(d)$ such that $\mathit{low}(c_2(d))=v_d$, where $v_d$ is a proper ancestor of $p(p(d))$, and $M_p(c_d)$ is a descendant of $c_1(d)$, where $c_d$ is the child of $v_d$ in the direction of $d$. (See Figure~\ref{figure:segment2}.) Then, after a linear-time preprocessing, we can have available the points $L(d,p(p(d)),c_d)$ and $R(d,p(p(d)),c_d)$ on $T$, for every $d\in U$.
\end{lemma}
\begin{proof}
The structure of this proof is the same as that of Lemma~\ref{lemma:segmentpointsL}. Thus, we first notice that, for every $d\in U$, the segment $T[p(p(d)),c_d]$ is not empty, because $v_d$ is a proper ancestor of $p(p(d))$, and therefore $c_d$ is an ancestor of $p(p(d))$. Now, for every $d\in U$, we generate the triple $Q_d=(d,p(p(d)),c_d)$, and let $\mathcal{Q}$ be the collection of all those triples. Then, as in the proof of Lemma~\ref{lemma:segmentpointsL}, it is sufficient to establish that $\mathcal{Q}$ satisfies Definition~\ref{definition:nestedqueries}. Recall that in the proof of Lemma~\ref{lemma:segmentpointsL} we relied heavily on the fact that, for every $d\in U$, we have that $M_p(c_d)$ is a descendant of $d$. The same fact holds here too, because $M_p(c_d)$ is a descendant of $c_1(d)$, which is a child of $M_p(d)$, which is a descendant of $d$. Thus, in what follows, it is sufficient to just recall which parts of the previous proof were a consequence of that fact, and then provide an appropriate argument in those places where the different conditions of the statement of the present lemma must be utilized.

First, let us suppose, for the sake of contradiction, that property $(1)$ of Definition~\ref{definition:nestedqueries} does not hold for $\mathcal{Q}$. This means that there are two distinct vertices $d$ and $d'$ in $U$ such that the segments $S_1=T[p(p(d)),c_d]$ and $S_2=T[p(p(d')),c_{d'}]$ intersect, but neither of them strictly includes the other. Recall that in the proof of Lemma~\ref{lemma:segmentpointsL}, here we made a distinction between the cases $c_d=c_{d'}$ and $c_d\neq c_{d'}$. The first case could be discarded as a consequence of the fact that $M_p(c_d)$ is a descendant of $d$, and $M_p(c_{d'})$ is a descendant of $d'$. Thus, it is sufficient to focus on the case that $c_d\neq c_{d'}$. And here, as previously, we assume, w.l.o.g., that $c_d<c_{d'}$, and we distinguish two possibilities: either \textbf{(a)} $p(p(d'))$ is a proper descendant of a vertex on $T(p(p(d)),c_d)$, but not of $p(p(d))$, or \textbf{(b)} $p(p(d'))$ is a proper descendant of $p(p(d))$. (See Figure~\ref{figure:segment1ab}.) In either case, we have $c_{d'}\in T[p(p(d)),c_d)$. Then, case \textbf{(a)} was discarded as a consequence of the fact that $M_p(c_d)$ is a descendant of $d$, and $M_p(c_{d'})$ is a descendant of $d'$.

Thus, we are left to consider case \textbf{(b)}. Since $d\in U$, we have that $M_p(c_d)$ is a descendant of $c_1(d)$. Now let $(x,y)$ be a back-edge where $x$ is a descendant of $c_1(d)$ and $y$ has the lowest possible $\mathit{low}$ point among all such back-edges. Then we have that $x$ is a descendant of $M_p(d)$ (since $c_1(d)$ is a child of $M_p(d)$), and therefore a descendant of $d$, and therefore a descendant of $p(p(d))$, and therefore a descendant of $c_{d'}$ (since $c_{d'}\in T[p(p(d)),c_d)$). Furthermore, due to the minimality of $y$, we have that $y$ is a proper ancestor of $p(c_d)$, and therefore a proper ancestor of $p(c_{d'})$ (since $c_d$ is a proper ancestor of $c_{d'}$). This shows that $(x,y)\in B_p(c_{d'})$. Now, since $d\in U$, we have that $\mathit{low}(c_2(d))=p(c_d)$. Thus, there is a back-edge $(x',y')$ with $x'$ a descendant of $c_2(d)$ and $y=p(c_d)$. Then, $x'$ is again a descendant of $c_{d'}$. Furthermore, $y$ is a proper ancestor of $p(c_{d'})$. Thus, $(x',y')\in B_p(c_{d'})$. Then, by the definition of $M_p(c_{d'})$, we have that $M_p(c_{d'})$ is an ancestor of both $x$ and $x'$. Therefore, since $x$ and $x'$ are descendants of two different children of $M_p(d)$, this implies that $M_p(c_{d'})$ is an ancestor of $M_p(d)$. Since $d'\in U$, we have that $M_p(c_{d'})$ is a descendant of $d'$. Therefore, $M_p(d)$ is a descendant of $d'$. 

Now, since $M_p(d)$ is a descendant of both $d$ and $d'$, we have that $d$ and $d'$ are related as ancestor and descendant. Therefore, since $p(p(d'))$ is a proper descendant of $p(p(d))$, we infer that $d'$ is a proper descendant of $d$. Now, since $d'\in U$, we have that $M_p(c_{d'})$ is a descendant of a child of $M_p(d')$. Since $M_p(c_{d'})$ is an ancestor of $M_p(d)$, this implies that $M_p(d')$ is a proper ancestor of $M_p(d)$. Since $c_1(d')$ is the child of $M_p(d')$ with the lowest $\mathit{low}$ point among the children of $M_p(d')$, this implies that $M_p(d)$ is a descendant of $c_1(d')$. To see this, let us suppose the contrary. Now, there is a back-edge $(z,w)$ where $z$ is a descendant of $c_1(d')$ and $w$ has the lowest possible $\mathit{low}$ point among all such back-edges. Then, $z$ is a descendant of $d'$, and therefore a descendant of $d$. Furthermore, $w$ must be lower than the lowest endpoint of all back-edges that stem from the child of $M_p(d')$ that is an ancestor of $M_p(d)$. But we have $\mathit{low}(c_1(d))<p(c_d)$, and therefore $w$ must be a proper ancestor of $p(c_d)$, and therefore a proper ancestor of $p(d)$ (since $c_d$ is an ancestor of $p(p(d))$). This means that $(z,w)\in B_p(d)$, and therefore $z$ is a descendant of $M_p(d)$, which is absurd (since $z$ and $M_p(d)$ are descendants of different children of $M_p(d')$). 

Thus, we have that $M_p(d)$ is a descendant of $c_1(d')$. Since $d'\in U$, we have that $\mathit{low}(c_2(d'))=p(c_{d'})$. This implies that there is a back-edge $(z,p(c_{d'}))$ such that $z$ is a descendant of $c_2(d')$. Then, we have that $z$ is a descendant of $d'$, and therefore a descendant of $d$. Furthermore, since $c_{d'}\in T[p(p(d)),c_d)$, we have that $c_{d'}$ is a proper ancestor of $p(d)$, and therefore $p(c_{d'})$ is a proper ancestor of $p(d)$. This shows that $(z,p(c_{d'}))\in B_p(d)$, and therefore $z$ is a descendant of $M_p(d)$. But since $M_p(d)$ is a descendant of $c_1(d')$, we have that $z$ is a descendant of $c_1(d')$, which is absurd (since $c_1(d')$ and $c_2(d')$ are two distinct children of $M_p(d')$). 

Thus, we have established property $(1)$ of Definition~\ref{definition:nestedqueries} for $\mathcal{Q}$. Now we can also establish property $(2)$ of Definition~\ref{definition:nestedqueries} for $\mathcal{Q}$, by following precisely the same argument as in the proof of Lemma~\ref{lemma:segmentpointsL}, since there we only used the fact that we had established property $(1)$, and that, for every $d\in U$, we have that $M_p(c_d)$ is a descendant of $d$.
\end{proof}

\begin{figure}[h!]\centering
\includegraphics[trim={1.7cm 22.5cm 0 0.5cm}, clip=true, width=1.1\linewidth]{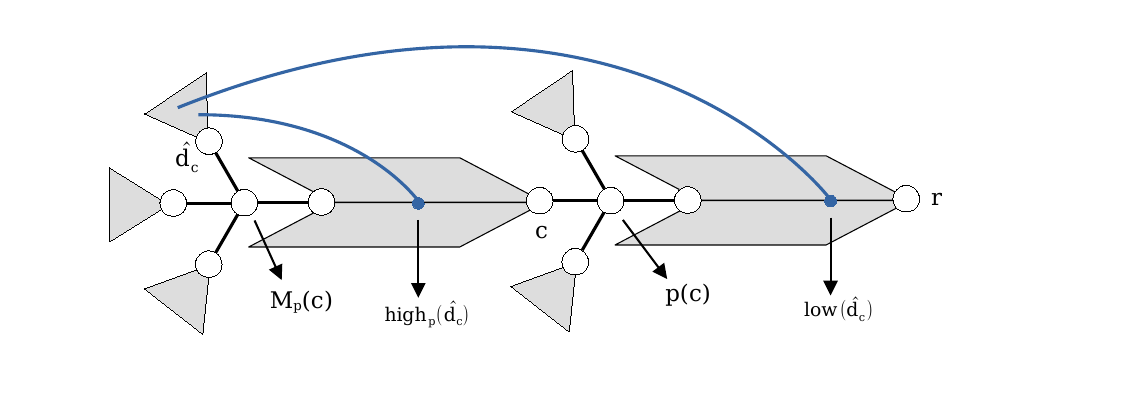}
\caption{\small{An example of a vertex $c$ that belongs to the set $W$ that appears in the statement of Lemma~\ref{lemma:segmentpointsspecial}. Thus, there is a unique child $\hat{d}_c$ of $M_p(c)$ such that $\mathit{high}_p(\hat{d}_c)\in T[p(M_p(c)),c]$ and $\mathit{low}(\hat{d}_c)<p(c)$. The goal is to find the leftmost and the rightmost points that start from $T(\hat{d}_c)$ and reach the segment $T[p(M_p(c)),c]$.}}\label{figure:segment3}
\end{figure}

\begin{lemma}
\label{lemma:segmentpointsspecial}
Let $T$ be a DFS tree, and let $W$ be the collection of all vertices $c$ with the following property (which is equivalent across all permutations of $T$): $M_p(c)$ is a proper descendant of $c$, and it has a unique child $\hat{d}_c$ with the property that $\mathit{high}_p(\hat{d}_c)\in T[p(M_p(c)),c]$ and $\mathit{low}(\hat{d}_c)<p(c)$. (See Figure~\ref{figure:segment3}.) Then, after a linear-time preprocessing, we can have available the points $L(\hat{d}_c,p(M_p(c)),c)$ and $R(\hat{d}_c,p(M_p(c)),c)$ on $T$, for every $c\in W$.
\end{lemma}
\begin{proof}
First of all, notice that, for every $c\in W$, the segment $T[p(M_p(c)),c]$ is not empty, because $M_p(c)$ is a proper descendant of $c$. Now let $\mathcal{Q}$ be the collection of triples of vertices $(\hat{d}_c,p(M_p(c)),c)$, for all $c\in W$. Then, according to Proposition~\ref{proposition:leftmostinsegment}, it is sufficient to establish that $\mathcal{Q}$ satisfies Definition~\ref{definition:nestedqueries}.

Let $c$ and $c'$ be two distinct vertices from $W$ such that the segments $S_1=T[p(M_p(c)),c]$ and $S_2=T[p(M_p(c')),c']$ intersect. This means that there is a vertex $w$ in $S_1\cap S_2$, and therefore $w$ is a common descendant of $c$ and $c'$, and therefore $c$ and $c'$ are related as ancestor and descendant. Thus, we may assume w.l.o.g. that $c'$ is a proper descendant of $c$. Now, since $S_1$ and $S_2$ intersect, this implies that $c'\in S_1$. Thus, $c'$ is an ancestor of $p(M_p(c))$, and therefore an ancestor of $M_p(c)$. Then, since $c'$ is a descendant of $c$, Lemma~\ref{lemma:Mp} implies that $M_p(c')$ is an ancestor of $M_p(c)$. Therefore, $p(M_p(c'))$ is an ancestor of $p(M_p(c))$. This implies that $S_2\subset S_1$. Thus, we have established that $\mathcal{Q}$ satisfies property $(1)$ of Definition~\ref{definition:nestedqueries}.

Now let $c$ and $c'$ be two distinct vertices from $W$ such that $T[p(M_p(c)),c]\subset T[p(M_p(c')),c']$. If $\hat{d}_{c}$ is an ancestor of $\hat{d}_{c'}$, then property $(i)$ of $(2)$ of Definition~\ref{definition:nestedqueries} is satisfied. So let us assume that this is not true, and let us consider a vertex $c''\in W$ (not necessarily distinct from $c'$) such that $T[p(M_p(c')),c']\subseteq T[p(M_p(c'')),c'']$. Now let us suppose, for the sake of contradiction, that there is a back-edge $(x,y)$ with $x\in T(\hat{d}_{c''})$ and $y\in T[p(M_p(c)),c]$.

Since $T[p(M_p(c)),c]\subset T[p(M_p(c')),c']$ and $T[p(M_p(c')),c']\subseteq T[p(M_p(c'')),c'']$, we infer that $T[p(M_p(c)),c]\subset T[p(M_p(c'')),c'']$. This implies that $c''$ is an ancestor of $c$, and that $c$ is an ancestor of $p(M_p(c''))$, and therefore an ancestor of $M_p(c'')$. Thus, Lemma~\ref{lemma:Mp} implies that $M_p(c'')$ is a descendant of $M_p(c)$. Then, since $x$ is a descendant of the child $\hat{d}_{c''}$ of $M_p(c'')$, we have that $x$ is a descendant of a child $d$ of $M_p(c)$ (and that $d$ is an ancestor of $\hat{d}_{c''}$). Then, due to the existence of the back-edge $(x,y)$, where $y\in T[p(M_p(c)),c]$, we infer that $\mathit{high}_p(d)\in T[p(M_p(c)),c]$. (This is a consequence of the maximality in the definition of $\mathit{high}_p(d)$, and of the fact that $\mathit{high}_p(d)$ must be a proper ancestor of $p(d)=M_p(c)$.) Furthermore, since $\mathit{low}(\hat{d}_{c''})<p(c'')$, we have $\mathit{low}(\hat{d}_{c''})<p(c)$. Therefore, since $T(\hat{d}_{c''})\subseteq T(d)$, we have $\mathit{low}(d)<p(c)$. Thus, since $\mathit{high}_p(d)\in T[p(M_p(c)),c]$ and $\mathit{low}(d)<p(c)$, we have $d=\hat{d}_c$ (due to the uniqueness in the definition of $\hat{d}_c$). 

With the same reasoning (applied for $c'$ and $c''$), we infer that there is a child $d'$ of $M_p(c')$ that is an ancestor of $x$, and that $d'=\hat{d}_{c'}$. Then, since $x$ is a common descendant of $\hat{d}_c$ and $\hat{d}_{c'}$, we have that $\hat{d}_c$ and $\hat{d}_{c'}$ are related as ancestor and descendant. Thus, since we work under the assumption that $\hat{d}_c$ is not an ancestor of $\hat{d}_{c'}$, we have that $\hat{d}_c$ is a proper descendant of $\hat{d}_{c'}$, and therefore $p(\hat{d}_c)=M_p(c)$ is a proper descendant of $p(\hat{d}_{c'})=M_p(c')$. But since $T[p(M_p(c)),c]\subset T[p(M_p(c')),c']$, using Lemma~\ref{lemma:Mp} as previously we have that $M_p(c')$ is a descendant of $M_p(c)$, a contradiction. This shows that $\mathcal{Q}$ satisfies property $(2)$ of Definition~\ref{definition:nestedqueries}.

\end{proof}

\begin{lemma}
\label{lemma:segmentpointsMp}
After a linear-time preprocessing on a DFS tree, we can construct an oracle that can answer any query of the following form in constant time: given two vertices $c$ and $d$ with $M_p(c)=M_p(d)$ and $c<p(d)$, return $L(d,p(p(d)),c)$ and $R(d,p(p(d)),c)$.
\end{lemma}
\begin{proof}
Let $U$ be the collection of all vertices $d$ with $M_p(d)\neq\bot$ and $\mathit{next}_{M_p}(d)<p(d)$. Then, for every $d\in U$, we generate a triple $(d,p(p(d)),\mathit{next}_{M_p}(d))$, and let $\mathcal{Q}$ be the collection of all those triples. Now we will use Proposition~\ref{proposition:leftmostinsegment}, in order to derive all points $L(d,p(p(d)),\mathit{next}_{M_p}(d))$ and $R(d,p(p(d)),\mathit{next}_{M_p}(d))$, for $d\in U$. Thus, it is sufficient to establish that $\mathcal{Q}$ satisfies Definition~\ref{definition:nestedqueries}.

Let $d$ and $d'$ be two distinct vertices in $U$ such that the segments $S_1=T[p(p(d)),\mathit{next}_{M_p}(d)]$ and $S_2=T[p(p(d')),\mathit{next}_{M_p}(d')]$ intersect. This means that there is a vertex $w$ in $S_1\cap S_2$. Then, $w$ is a common descendant of $\mathit{next}_{M_p}(d)$ and $\mathit{next}_{M_p}(d')$, and therefore $\mathit{next}_{M_p}(d)$ and $\mathit{next}_{M_p}(d')$ are related as ancestor and descendant. Thus, we may assume w.l.o.g. that $\mathit{next}_{M_p}(d)$ is a descendant of $\mathit{next}_{M_p}(d')$. Since $S_1$ and $S_2$ intersect, this implies that $\mathit{next}_{M_p}(d)\in T[p(p(d')),\mathit{next}_{M_p}(d')]$, and therefore $\mathit{next}_{M_p}(d)$ is an ancestor of $p(p(d'))$, and therefore an ancestor of $d'$. Since $M_p(\mathit{next}_{M_p}(d'))=M_p(d')$, we have that $M_p(\mathit{next}_{M_p}(d'))$ is a descendant of $d'$. Therefore, since $\mathit{next}_{M_p}(d)$ is an ancestor of $d'$, it is also an ancestor of $M_p(\mathit{next}_{M_p}(d'))$. Thus, since $\mathit{next}_{M_p}(d)$ is a descendant of $\mathit{next}_{M_p}(d')$, Lemma~\ref{lemma:Mp} implies that $M_p(\mathit{next}_{M_p}(d))$ is an ancestor of $M_p(\mathit{next}_{M_p}(d'))$. Then, since $M_p(d)=M_p(\mathit{next}_{M_p}(d))$, we have that $M_p(d)$ is an ancestor of $M_p(\mathit{next}_{M_p}(d'))=M_p(d')$, and therefore $d$ is an ancestor of $M_p(d')$. Then, since $M_p(d')$ is a common descendant of $d'$ and $d$, we have that $d'$ and $d$ are related as ancestor and descendant. 

So let us suppose, for the sake of contradiction, that $d$ is not a proper ancestor of $d'$. Then we have that $d$ is a descendant of $d'$. Since $\mathit{next}_{M_p}(d)\in T[p(p(d')),\mathit{next}_{M_p}(d')]$, we have that $\mathit{next}_{M_p}(d)$ is an ancestor of $p(p(d'))$, and therefore an ancestor of $d'$. Thus, since $M_p(\mathit{next}_{M_p}(d))=M_p(d)$ is a descendant of $d$, which is a descendant of $d'$, Lemma~\ref{lemma:Mp}, implies that $M_p(d')$ is an ancestor of $M_p(d)$. Then, since $M_p(d)$ is an ancestor of $M_p(d')$, we infer that $M_p(d)=M_p(d')$. But then, since $d$ is a descendant of $d'$ with $d\neq d'$, we have $\mathit{next}_{M_p}(d)\geq d'$, which implies that $\mathit{next}_{M_p}(d)$ cannot be an ancestor of $p(d')$, a contradiction. This shows that $d$ is a proper ancestor of $d'$, and therefore $p(p(d))$ is a proper ancestor of $p(p(d'))$, and therefore $S_1\subset S_2$. Thus, we have established that $\mathcal{Q}$ satisfies property $(1)$ of Definition~\ref{definition:nestedqueries}.


Now, in order to show that $\mathcal{Q}$ satisfies property $(2)$ of Definition~\ref{definition:nestedqueries}, let $d$ and $d'$ be two distinct vertices in $U$ such that the segments $S_1=T[p(p(d)),\mathit{next}_{M_p}(d)]$ and $S_2=T[p(p(d')),\mathit{next}_{M_p}(d')]$ satisfy $S_1\subset S_2$. This implies that $\mathit{next}_{M_p}(d)\in S_2$, and therefore $\mathit{next}_{M_p}(d)$ is a descendant of $\mathit{next}_{M_p}(d')$. Then we can work as previously, in order to conclude that $d$ is an ancestor of $d'$. Thus, property $(i)$ of $(2)$ of Definition~\ref{definition:nestedqueries} is satisfied for the triples $(d,p(p(d)),\mathit{next}_{M_p}(d))$ and $(d',p(p(d')),\mathit{next}_{M_p}(d'))$.

Thus, by Proposition~\ref{proposition:leftmostinsegment}, we can have available the points $L(d,p(p(d)),\mathit{next}_{M_p}(d))$ and $R(d,p(p(d)),\mathit{next}_{M_p}(d))$, for all $d\in U$, in linear time in total. 

Now let $U'$ be the set of all vertices $d$ with $\mathit{next}_{M_p}(d)\neq\bot$. We will need to collect all points of the form $L(M_p(d),p(d),\mathit{next}_{M_p}(d))$ and $R(M_p(d),p(d),\mathit{next}_{M_p}(d))$, for $d\in U'$. For this purpose, we define, for every vertex $d\in U'$, the two points $\widetilde{L}(d)$ and $\widetilde{R}(d)$, as the minimum and the maximum, respectively, descendant $x$ of $M_p(d)$, for which there exists a back-edge of the form $(x,p(d))$. Then, notice that, if $d\in U'$ with $\mathit{next}_{M_p}(d)=p(d)$, then $L(M_p(d),p(d),\mathit{next}_{M_p}(d))=\widetilde{L}(d)$ (and the same is true if  ``$L$" is replaced with ``$R$"). Furthermore, if $d\in U'$ with $\mathit{next}_{M_p}(d)<p(d)$, then $L(M_p(d),p(d),\mathit{next}_{M_p}(d))=\min\{\widetilde{L}(d), L(d,p(p(d)),\mathit{next}_{M_p}(d))\}$ (and the same is true if every ``$L$" is replaced with ``$R$", and ``$\min$" with ``$\max$"). In case this last claim is not entirely obvious, consider that every back-edge $(x,y)$ where $x$ is a descendant of $d$ and $y$ is an ancestor of $p(p(d))$, actually satisfies that $x$ is a descendant of $M_p(d)$. Thus, $L(M_p(d),p(d),\mathit{next}_{M_p}(d))\leq L(d,p(p(d)),\mathit{next}_{M_p}(d))$. (And $\widetilde{L}(d)$ considers the back-edges that start from a descendant of $M_p(d)$ and end precisely at $p(d)$.)

Now, computing all $\widetilde{L}(d)$ and $\widetilde{R}(d)$, for $d\in U'$, in total linear time, is very easy: just consider, for every $d\in U'$, all back-edges of the form $(x,p(d))$, where $x$ is a descendant of $M_p(d)$, and take the minimum and the maximum of all such $x$, respectively. This is obviously correct, and the linear-time bound is guaranteed by Lemma~\ref{lemma:siblingsM} (which ensures that, since $\mathit{next}_{M_p}(d)\neq\bot$ for $d\in U'$, we have that $d$ is the only child of $p(d)$ that is in $U'$). 

Now let $c$ and $d$ be two vertices with $M_p(c)=M_p(d)$ and $c<p(d)$. We will demonstrate the relation between $L(d,p(p(d)),c)$ and the above points we have discussed. (And the arguments for $R(d,p(p(d)),c)$ are similar.) Let $x_1,\dots,x_N$ be the sequence of vertices with $x_1=d$, $x_N=c$, and $x_{i+1}=\mathit{next}_{M_p}(x_i)$, for $i\in\{1,\dots,N-1\}$. (Notice that $N\geq 2$.) Now we define a set of values $L_1,\dots,L_{N-1}$ as follows. If $\mathit{next}_{M_p}(d)<p(d)$, we let $L_1=L(d,p(p(d)),x_2)$. Otherwise, (i.e., if $\mathit{next}_{M_p}(d)=p(d)$), we let $L_1=\bot$. Furthermore, we let $L_i=L(M_p(x_i),p(x_i),x_{i+1})$, for every $i\in\{2,\dots,N-1\}$. Now we claim that $L(d,p(p(d)),c)=\min\{L_1,\dots,L_{N-1}\}$ (where if $L_1=\bot$, then the $\min$ just ignores it). To see this, let $x=L(d,p(p(d)),c)$. This means that $x$ is the minimum descendant of $d$ such that there is a back-edge $(x,y)$ with $y\in T[p(p(d)),c]$. Thus, since $y<p(d)$, we have that $(x,y)\in B_p(d)$, and therefore $x$ is a descendant of $M_p(d)$. Furthermore, since $y\in T[p(p(d)),c]$, either there is an $i\in\{2,\dots,N-1\}$ such that $y\in T[p(x_i),x_{i+1}]$, or $y\in T[p(p(d)),\mathit{next}_{M_p}(d)]$ and $\mathit{next}_{M_p}(d)<p(d)$. This shows that $\min\{L_1,\dots,L_{N-1}\}\leq x$. Conversely, it should be clear that $x\leq L_i$ for every $i\in\{1,\dots,N-1\}$, and this establishes that $x=\min\{L_1,\dots,L_{N-1}\}$.

Thus, we make the following preprocessing in order to construct the oracle. For every vertex $z$, we construct the (possibly empty) list of all vertices $d$ with $M_p(d)=z$, and we have it sorted in decreasing order. We denote this list as $M_p^{-1}(z)$. Then, notice that the total size of all $M_p^{-1}(z)$, for all vertices $z$, is $O(n)$, and their construction can be completed in $O(n)$ time with bucket-sort. Now, for every $z$ with $|M_p^{-1}(z)|>1$, we create an array $\mathcal{A}(z)$ of size $|M_p^{-1}(z)|-1$, with the property that the $i$-th element of $\mathcal{A}(z)$ is $L(M_p(d),p(d),\mathit{next}_{M_p}(d))$, where $d$ is the $i$-th element of $M_p^{-1}(z)$. (Notice that we exclude an entry for the last element of $M_p^{-1}(z)$, because this has $\mathit{next}_{M_p}=\bot$.) Then, we initialize a range-minimum query (RMQ) data structure \cite{DBLP:conf/cpm/FischerH06} on $\mathcal{A}(z)$, that takes $O(|\mathcal{A}(z)|)$ time to be constructed, and can answer any query of the form ``what is the minimum entry of $\mathcal{A}(z)$ between the $i$-th and the $j$-th index?'', in constant time. To complete the description of the data structure, we also compute a pointer for every vertex $c$ to its index $\mathit{index}(c)$ in $M_p^{-1}(M_p(c))$. It is clear that all these ingredients can be computed in linear time. 

Now, if we are given two vertices $c$ and $d$ with $M_p(c)=M_p(d)$ and $c<p(d)$, we follow the method described previously, in order to compute $L(d,p(p(d)),c)$. Thus, if $c=\mathit{next}_{M_p}(d)$, then we simply return $L(d,p(p(d)),c)$, as this is already computed. Otherwise, we return the minimum of $L_1$ and $L_2$, which are computed as follows. If $\mathit{next}_{M_p}(d)\neq p(d)$, then $L_1=L(d,p(p(d)),\mathit{next}_{M_p}(d))$ (which is already available); otherwise, $L_1=\bot$. And $L_2$ is the answer to the RMQ query on $\mathcal{A}(M_p(d))$ between the indices $\mathit{index}(d)+1$ and $\mathit{index}(c)-1$. Thus, we can get the answer for $L(d,p(p(d)),c)$ in constant time.
\end{proof}

\section{The case where $u$ is an ancestor of both $v$ and $w$, but $v,w$ are not related as ancestor and descendant}
\label{section:vwnotrelated}
Let $F=\{u,v,w\}$ be the set of failed vertices. In the case where $u$ is an ancestor of both $v$ and $w$, but $v,w$ are not related as ancestor and descendant, we distinguish two cases, depending on whether $v$ and $w$ are descendants of different children of $u$, or of the same child of $u$.

\subsection{The case where $v$ and $w$ are descendants of different children of $u$}
\label{section:vwdescendantsofcc'}
Here we consider the case where $v$ and $w$ are descendants of two different children $c$ and $c'$ of $u$, respectively. (These children can be identified in constant time with a level-ancestor query as explained in Section~\ref{section:basicDFS}.) Thus, after the removal of $\{u,v,w\}$, by Assumption~\ref{assumption1} we have three internal components on $T$: $(A)$ the set of vertices of $T$ that are not descendants of $u$; $(B)$ the set of vertices of $T$ that are descendants of $c$ but not of $v$; and $(C)$ the set of vertices of $T$ that are descendants of $c'$ but not of $w$. (See Figure~\ref{figure:notrelated1}.) Then, our goal is to determine whether $B$ is connected with $A$, and whether $C$ is connected with $A$, either directly with a back-edge, or through hanging subtrees (of $v$ and $w$, respectively).

\begin{figure}[h!]\centering
\includegraphics[trim={0cm 15cm 0 0cm}, clip=true, width=1\linewidth]{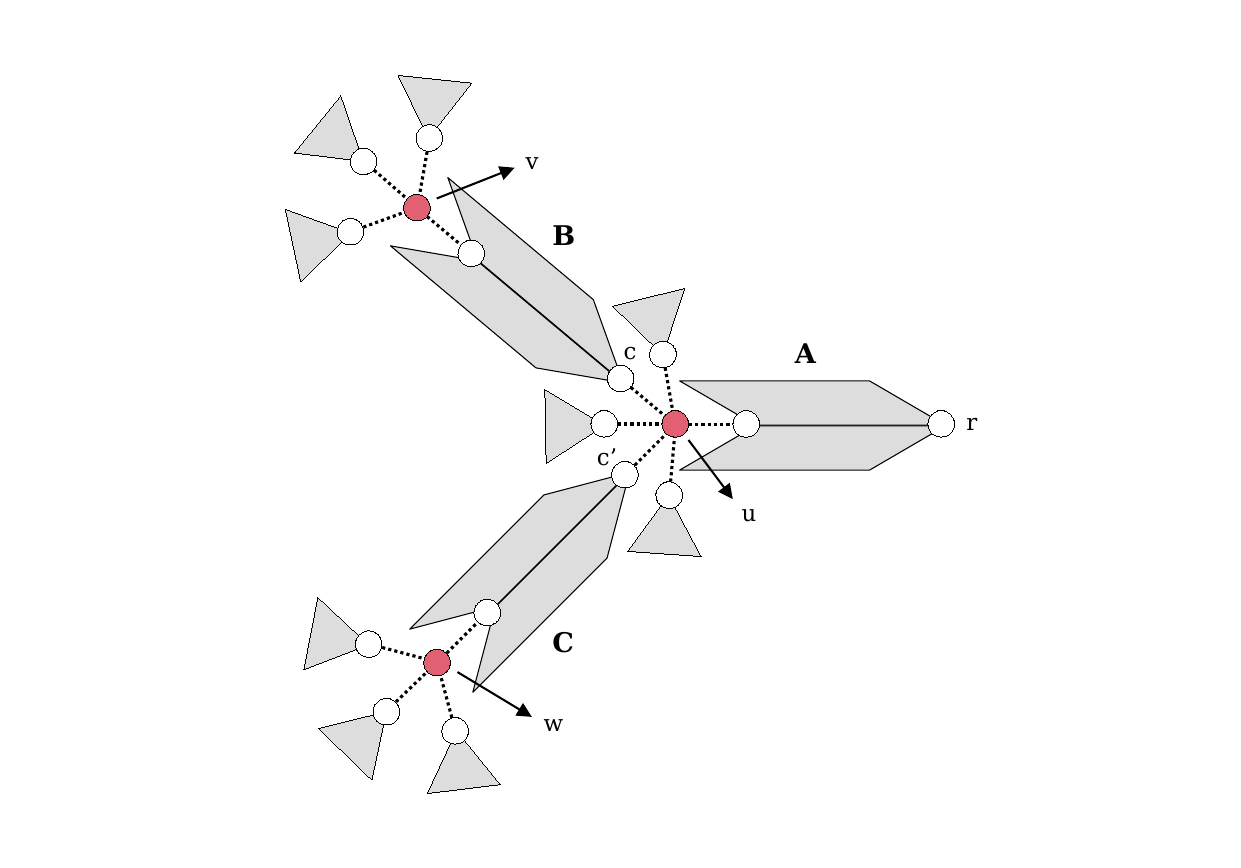}
\caption{\small{An illustration of the case where $u$ is an ancestor of $v$ and $w$, but $v$ and $w$ are descendants of different children of $u$. (Here, $v$ is a descendant of $c$, and $w$ is a descendant of $c'$.) The goal is to check whether the parts $A$, $B$ and $C$ remain connected. To do this, we rely on Lemma~\ref{lemma:lbelowwonT_high_MAIN}, and we check the connections $B$-$A$ and $C$-$A$ separately.}}\label{figure:notrelated1}
\end{figure}

We will show how to determine whether $B$ is connected with $A$, since we can follow a similar procedure to determine the connection of $C$ and $A$. Here we will use the criterion provided by Lemma~\ref{lemma:lbelowwonT_high_MAIN}. Thus, we will use the leftmost and the rightmost points $L_p(c)$ and $R_p(c)$, respectively, of $c$. Specifically, we consider both $L_p(c)$ and $R_p(c)$ on $T_\mathit{highDec}$. Now, if both $L_p(c)$ and $R_p(c)$ are $\bot$, then we conclude that $B$ is not connected with $A$. So let us consider the case that these are not $\bot$. Now, if either of those is in $B$, then we have that $B$ is connected with $A$ directly with a back-edge. Otherwise, we have that both $L_p(c)$ and $R_p(c)$ are descendants of $v$ (and there is no back-edge that connects $B$ and $A$ directly). Then, it is sufficient to consider the child $d$ of $v$ that is an ancestor of $L_p(c)$ (which can be identified with a level-ancestor query as explained in Section~\ref{section:basicDFS}), and then check whether $\mathit{high}_p(d)\in B$. Since we work on $T_\mathit{highDec}$, it is easy to see that there is a hanging subtree of $v$ that connects $B$ and $A$ if and only if $\mathit{high}_p(d)\in B$ (i.e., because, if such a hanging subtree exists, then $T(d)$ is definitely one of those). For a technical proof, we refer to Lemma~\ref{lemma:lbelowwonT_high_MAIN}.

\subsection{The case where $v$ and $w$ are descendants of the same child of $u$}
\label{section:vwdescendantsofc}
Here we consider the case where $v$ and $w$ are descendants of the same child $c$ of $u$ (which can be identified in constant time with a level-ancestor query as explained in Section~\ref{section:basicDFS}). Thus, after the removal of $\{u,v,w\}$, by Assumption~\ref{assumption1} there are two internal components on $T$: $(A)$ the set of vertices of $T$ that are not descendants of $u$, and $(B)$ the set of vertices of $T$ that are descendants of $c$ but not of $v$ or $w$. (See Figure~\ref{figure:notrelated2}.) Then, our goal is to determine whether $B$ is connected with $A$, either directly with a back-edge, or through hanging subtrees (of either $v$ or $w$).

\begin{figure}[h!]\centering
\includegraphics[trim={0cm 20cm 0 0cm}, clip=true, width=1\linewidth]{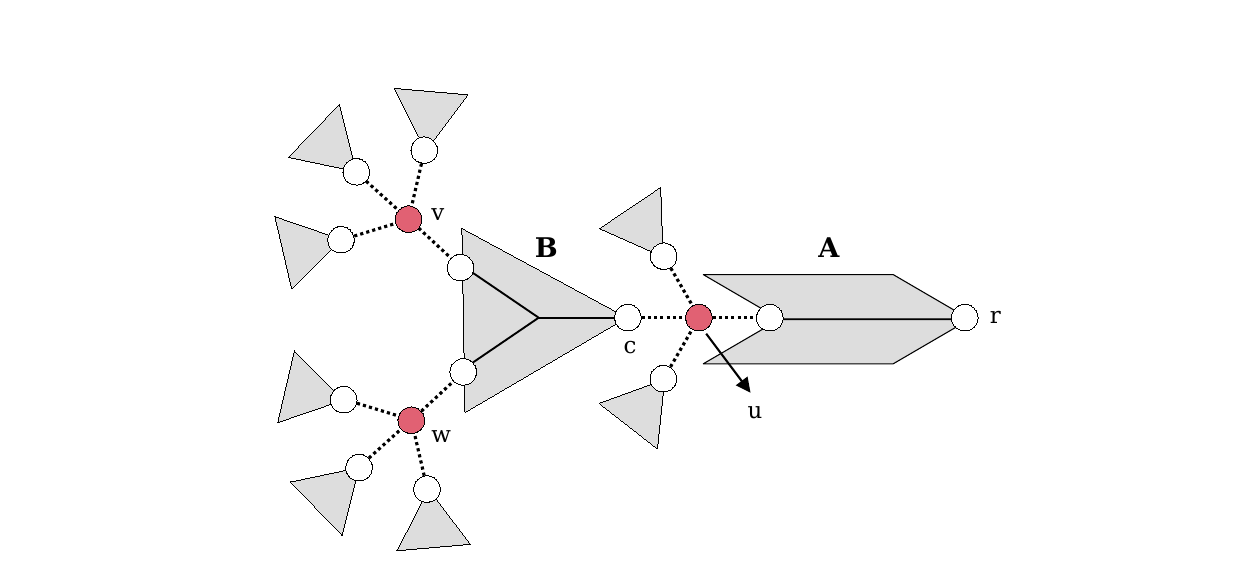}
\caption{\small{An illustration of the case where $v$ and $w$ are descendants of the same child $c$ of $u$, but they are not related as ancestor and descendant. The goal is to check whether the parts $A$ and $B$ remain connected. To do this, we consider the leftmost and rightmost points, $L_p(c)$ and $R_p(c)$, of $c$. If either of those is on $B$, then we are done. Otherwise, we proceed as described in the main text.}}\label{figure:notrelated2}
\end{figure}

Here, as previously, we will use the leftmost and the rightmost points $L_p(c)$ and $R_p(c)$, respectively, of $c$ (see Section~\ref{section:leftmostrightmost}). Specifically, we consider both $L_p(c)$ and $R_p(c)$ on $T_\mathit{highDec}$. If either of those is $\bot$, then $B$ is not connected with $A$. So let us consider the case where $L_p(c)$ and $R_p(c)$ are not $\bot$. Now, if either of those is in $B$, then we obviously have that $B$ is connected with $A$ directly with a back-edge. So let us consider the case that neither of them is in $B$. This implies that each of $L_p(c)$ and $R_p(c)$ is a descendant of either $v$ or $w$. 

Suppose first that both $L_p(c)$ and $R_p(c)$ are descendants of one of $v$ and $w$, and let us assume w.l.o.g. that both of them are descendants of $v$. Then, by Assumption~\ref{assumption2}, we have that $L_p(c)$ and $R_p(c)$ are proper descendants of $v$. Now, it is sufficient to consider the child $d$ of $v$ that is an ancestor of $L_p(c)$ (which can be identified with a level-ancestor query as explained in Section~\ref{section:basicDFS}), and then check whether $\mathit{high}_p(d)\in B$. Since we work on $T_\mathit{highDec}$, it is easy to see that there is a hanging subtree of $v$ that connects $B$ and $A$ if and only if $\mathit{high}_p(d)\in B$ (i.e., because, if such a hanging subtree exists, then $T(d)$ is definitely one of those). For a technical proof, we refer to the argument for Lemma~\ref{lemma:lbelowwonT_high_MAIN}.
 
So, finally, it remains to consider the case where one of $L_p(c)$ and $R_p(c)$ is a descendant of $v$, and the other is a descendant of $w$. Then, notice that Assumption~\ref{assumption2} implies that $L_p(c)$ and $R_p(c)$ are descendants of different children of $M_p(c)$ (i.e., because $L_p(c)$ cannot coincide with $M_p(c)$, because $L_p(c)$ provides back-edges, but $M_p(c)$ is not a leaf, because it has two distinct descendants). Now we will switch our consideration for a while to the tree $T_\mathit{lowInc}$. Since $L_p(c)$ and $R_p(c)$ are descendants of different children of $M_p(c)$, we have that $M_p(c)$ contains at least two children with $\mathit{low}$ point less than $p(c)=u$, and therefore the first two children of $M_p(c)$ on $T_\mathit{lowInc}$ have $\mathit{low}<u$. If $M_p(c)$ has a third child which also satisfies this property, then we can conclude that there is a back-edge that connects $B$ and $A$ directly. So let us consider the case that only the first two children $d$ and $d'$ of $M_p(c)$ have $\mathit{low}<u$. Then, one of those is the one that is an ancestor of $v$, and the other is the one that is an ancestor of $w$. So let us assume w.l.o.g. that $d$ is an ancestor of $L_p(c)$ and $d'$ is an ancestor of $R_p(c)$.

Now, since $d$ and $d'$ are two children of $M_p(c)$ that are uniquely determined by $c$ (i.e., they are the only two children of $M_p(c)$ with $\mathit{low}<p(c)$), we can afford to have computed, during the preprocessing phase, the leftmost and the rightmost points $L(c,d)$, $R(c,d)$, $L(c,d')$ and $R(c,d')$ (see Section~\ref{section:leftmostrightmost}, and, in particular, Proposition~\ref{proposition:L(v,d)}). Specifically, we compute those points on $T_\mathit{highDec}$. Then, we can work as previously, in order to determine whether $B$ and $A$ are connected with a back-edge directly, or through a hanging subtree. If either of those four points is on $B$, then $B$ is connected with $A$ directly with a back-edge. Otherwise, we consider the child $f$ of $v$ that is an ancestor of $L(c,d)$, and the child $f'$ of $w$ that is an ancestor of $L(c,d')$, and we check whether the $\mathit{high}_p$ point of at least one of those is on $B$. By following the same argument as in the proof of Lemma~\ref{lemma:lbelowwonT_high_MAIN}, we can see that there is a hanging subtree of $v$ or of $w$ that connects $B$ and $A$, if and only if $\mathit{high}_p(f)\in B$ or $\mathit{high}_p(f')\in B$, respectively.

\section{The case where $u$ is an ancestor of $v$, and $v$ is an ancestor of $w$}
\label{section:vwrelated}
In this section we deal with the most difficult case, which is that where $u$ is an ancestor of $v$, and $v$ is an ancestor of $w$. By Assumption~\ref{assumption1}, we have that $v$ is a proper descendant of a child $c$ of $u$, and $w$ is a proper descendant of a child $d$ of $v$. Furthermore, after the removal of $\{u,v,w\}$ there are three internal components on $T$: $(A)$ the set of vertices that are not descendants of $u$; $(B)$ the set of vertices that are descendants of $c$, but not of $v$; and $(C)$ the set of vertices that are descendants of $d$, but not of $w$. (See Figure~\ref{figure:ABC}.) Our purpose will be to determine the connection of the parts $A$, $B$ and $C$ through the existence of back-edges that connect them directly, or through hanging subtrees, by utilizing our DFS-based parameters. (Note: when we talk about the connection of $A$, $B$ and $C$, from now on we mean their connection in $G\setminus\{u,v,w\}$.) We found that it is convenient to distinguish the following six cases: $(1)$ the case where $M_p(c)=\bot$, $(2)$ the case where $M_p(c)\in B$, $(3)$ the case where $M_p(c)=v$, $(4)$ the case where $M_p(c)\in C$, $(5)$ the case where $M_p(c)=w$, and $(6)$ the case where $M_p(c)$ is a descendant of a child $d'$ of $w$. It should be clear that these cases are mutually exclusive and exhaust all possibilities for $M_p(c)$. We will consider those cases in turn, subdividing them further in a convenient manner as we proceed.

\begin{figure}[h!]\centering
\includegraphics[trim={0cm 23cm 0 0cm}, clip=true, width=1\linewidth]{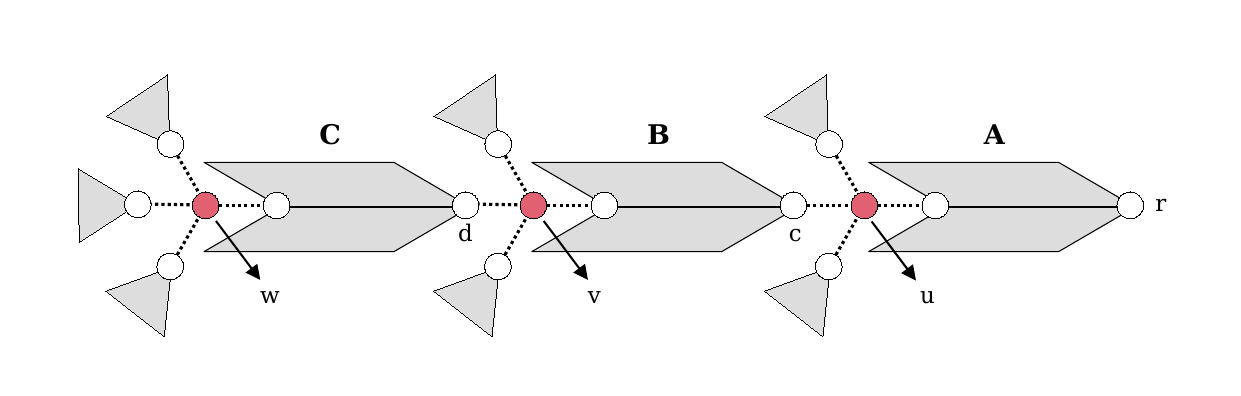}
\caption{\small{The case where $u$ is an ancestor of $v$, and $v$ is an ancestor of $w$. $c$ is the child of $u$ in the direction of $v$, and $d$ is the child of $v$ in the direction of $w$. The DFS tree is split into three internal components: $A=T(r)\setminus T(u)$, $B=T(c)\setminus T(v)$, and $C=T(d)\setminus T(w)$. The goal is to determine whether some of the parts $A$, $B$ and $C$, are connected directly with back-edges, or through the mediation of hanging subtrees. Our approach is to distinguish between the six distinct cases w.r.t. the location of $M_p(c)$: $(1)$ $M_p(c)=\bot$, $(2)$ $M_p(c)\in B$, $(3)$ $M_p(c)=v$, $(4)$ $M_p(c)\in C$, $(5)$ $M_p(c)=w$, and $(6)$ $M_p(c)$ is a proper descendant of $w$.}}\label{figure:ABC}
\end{figure} 

\subsection{The case where $M_p(c)=\bot$}
\label{section:M(c)=bot}
This is the easiest case for $M_p(c)$, and it provides an opportunity to demonstrate some of the techniques that we are going to use in the more demanding cases. Notice that the condition $M_p(c)=\bot$ implies that there is no back-edge $(x,y)$ with $x\in T(c)$ and $y\in A$, and therefore $A$ is isolated from the parts $B$ and $C$. Thus, we only have to determine whether $B$ is connected with $C$. To do this, we find it convenient to distinguish the following three cases: $(1)$ the case where $M_p(d)$ is a descendant of a child $d'$ of $w$, $(2)$ the case where $M_p(d)=w$, and $(3)$ the case where $M_p(d)\in C$. It should be clear that these cases are mutually exclusive and exhaust all possibilities for $M_p(d)$, except for the case $M_p(d)=\bot$. However, the case where $M_p(d)=\bot$ is trivial, because it implies that $B$ is not connected with $C$, because it means that there is no back-edge $(x,y)$ with $x\in T(d)$ and $y<p(d)$, and therefore no back-edge $(x,y)$ with $x\in T(d)$ and $y\in B$.

\subsubsection{The case where $M_p(d)$ is a descendant of a child of $w$}
In the case where $M_p(d)$ is a descendant of a child $d'$ of $w$ we have that there is no back-edge that connects $B$ and $C$ directly, and that the subtree with root $d'$ is the only hanging subtree that may connect $B$ and $C$. To see this, it is sufficient to observe that, since $M_p(d)\in T(d')$, we have that all back-edges $(x,y)$ with $x\in T(d)$ and $y<p(d)$ satisfy $x\in T(M_p(d))\subseteq T(d')$. 

Thus, we have that $B$ and $C$ are connected if and only if: $(1)$ there is a back-edge $(x,y)$ with $x\in T(d')$ and $y\in C$, \textbf{and} $(2)$ there is a back-edge $(x',y')$ with $x'\in T(d')$ and $y'\in B$. It is easy to see that the existence of a back-edge satisfying $(1)$ is equivalent to $\mathit{high}_p(d')\in C$. In order to check for the existence of a back-edge satisfying $(2)$, we first observe that $\mathit{low}(d')\geq u$ (which is a consequence of $M_p(c)=\bot$). Thus, it should be clear that the existence of a back-edge satisfying $(2)$ is equivalent to $\mathit{low}_1(d')\in B$ or $\mathit{low}_2(d')\in B$. (We need to check $\mathit{low}_2(d')\in B$ because we may have $\mathit{low}_1(d')=u$, in which case we must ask for the second lowest lower endpoint of a back-edge that stems from $T(d')$ and ends in a proper ancestor of $w$.)

\subsubsection{The case where $M_p(d)=w$}
\label{section:M(d)=w,M(c)=bot}
In this case, as in the previous one, we have that there is no back-edge that connects $B$ and $C$ directly. But, contrary to the previous case, there may exist several children of $w$ that are roots of subtrees that connect $B$ and $C$. However, it is sufficient (and necessary) to determine the existence of one such child. Notice that, as analyzed in the previous case, if there is a child $d'$ of $w$ that connects $B$ and $C$, then it has $(1)$ $\mathit{high}_p(d')\in C$ and $(2)$ either $\mathit{low}_1(d')\in B$ or $\mathit{low}_2(d')\in B$. Conversely, if a child $d'$ of $w$ satisfies $(1)$ and $(2)$, then it connects $B$ and $C$.

Now, in order to determine the existence of a child $d'$ of $w$ that satisfies $(1)$ and $(2)$, it is sufficient to answer the following question:\\

$(*)$ among the children of $w$ that have $\mathit{high}_p\in C$, is there at least one that has either $\mathit{low}_1\in B$ or $\mathit{low}_2\in B$?\\

To answer this question in constant time, it is sufficient to know the following three parameters: among the children of $w$ that have $\mathit{high}_p\in C$, what is $(i)$ the lowest $\mathit{low}_1$ point, $(ii)$ the second-lowest $\mathit{low}_1$ point, and $(iii)$ the lowest $\mathit{low}_2$ point. First we will show that these parameters are indeed sufficient in order to answer question $(*)$ (by checking whether at least one of them is in $B$), and then we will show how we can have them computed during a linear-time preprocessing.

First, since we are working under the assumption $M_p(c)=\bot$, notice that the $\mathit{low}$ point of any child of $w$ must be at least $u$. Now let us assume that there is a child $d'$ of $w$ that satisfies $(1)$ and $(2)$. Then, we have that either $\mathit{low}_1(d')\in B$ or $\mathit{low}_2(d')\in B$. First, suppose that $\mathit{low}_1(d')\in B$. We have that $(i)$ can be at least $u$. If $(i)$ is $u$, then the existence of $d'$ implies that $(ii)$ is in $B$. Otherwise, we have that $(i)$ is greater than $u$, and the existence of $d'$ implies that $(i)$ is in $B$. Now suppose that $\mathit{low}_1(d')\notin B$, and so $\mathit{low}_2(d')\in B$. Since the $\mathit{low}_1$ point of every child of $w$ is at least $u$, every child of $w$ with $\mathit{low}_2\neq\bot$ has its $\mathit{low}_2$ point greater than $u$. Furthermore, since $\mathit{low}_1(d')\notin B$ and $\mathit{low}_2(d')\in B$, we infer that $\mathit{low}_1(d')=u$. Thus, the existence of $d'$ implies that $(iii)$ is in $B$. Conversely, it is obvious that, if at least one of the parameters $(i)$, $(ii)$ and $(iii)$ is in $B$, then there exists a child $d'$ of $w$ with $(1)$ $\mathit{high}_p(d')\in C$ and $(2)$ either $\mathit{low}_1(d')\in B$ or $\mathit{low}_2(d')\in B$.

Now we will show how to compute the parameters $(i)$, $(ii)$ and $(iii)$ in the preprocessing phase. First, notice that such parameters are associated with every vertex $d$ that has $M_p(d)\neq\bot$. (To see this, observe, first, that $w=M_p(d)$, and second, that it is sufficient to replace the part ``$C$" that appears in the definition of those parameters with $T[p(M_p(d)),d]$, since the $\mathit{high}_p$ points under consideration are those of children of $w$, and therefore they are proper ancestors of $w$.) Then, notice that these parameters are an invariant across all permutations of the base DFS tree $T$, and that it is most convenient to work with $T_\mathit{highDec}$. Since the children of $w$ are sorted in decreasing order w.r.t. their $\mathit{high}_p$ point on $T_\mathit{highDec}$, we have that the set of children of $w$ with the property that their $\mathit{high}_p$ point is on $T[p(w),d]$ form an initial (possibly empty) segment of the list of the children $w$. Then, if we process the vertices $d$ in decreasing order, this initial segment can only become larger. Thus, it is sufficient to maintain three variables associated with $w$, that store the lowest $\mathit{low}_1$ point, the second lowest $\mathit{low}_1$ point, and the lowest $\mathit{low}_2$ point that were encountered so far. It is trivial to update those variables as we expand the initial segment of the list of children of $w$. For an illustration, the whole procedure is shown in Algorithm~\ref{algorithm:three_low_parameters}. It is easy to see that it has a linear-time implementation.  

\begin{algorithm}[h!]
\caption{\textsf{Compute the parameters $(i)$, $(ii)$ and $(iii)$ that are used in the case $M_p(c)=\bot$ and $M_p(d)=w$.}}
\label{algorithm:three_low_parameters}
\LinesNumbered
\DontPrintSemicolon

\ForEach{vertex $d$}{
  $\mathtt{lowest\_low1}[d]\leftarrow \bot$\;
  $\mathtt{second\_lowest\_low1}[d]\leftarrow \bot$\;
  $\mathtt{lowest\_low2}[d]\leftarrow \bot$\;
   \textcolor{cblue}{\tcp{these correspond to $(i)$, $(ii)$ and $(iii)$, respectively, as defined in the main text}}
}

\ForEach{vertex $w$}{
  $\mathit{currentChild}[w]\leftarrow $ first child of $w$ on $T_\mathit{highDec}$\;
  $\mathit{firstParam}[w]\leftarrow\bot$\;
  $\mathit{secondParam}[w]\leftarrow\bot$\;
  $\mathit{thirdParam}[w]\leftarrow\bot$\;
  
}

\For{$d\leftarrow n$ to $d=2$}{
  \If{$M_p(d)=\bot$ \textbf{or} $d=M_p(d)$}{\textbf{continue}\;}
  $w\leftarrow M_p(d)$\;
  $c\leftarrow\mathit{currentChild}[w]$\;
  \While{$\mathit{high}_p(c)\in T[p(w),d]$}{
  \label{line:whileofiiiiii}
    \If{$\mathit{firstParam}[w]=\bot$ \textbf{or} $\mathit{low}_1(c)<\mathit{firstParam}[w]$}{
      $\mathit{secondParam}[w]\leftarrow\mathit{firstParam}[w]$\;
      $\mathit{firstParam}[w]\leftarrow\mathit{low}_1(c)$\;
    }
    \ElseIf{(($\mathit{secondParam}[w]=\bot$ \textbf{or} $\mathit{low}_1(c)<\mathit{secondParam}[w]$) \textbf{and} $\mathit{low}_1(c)\neq\mathit{firstParam}[w]$)}{
      $\mathit{secondParam}[w]\leftarrow\mathit{low}_1(c)$\;
    }
    \If{$\mathit{low}_2(c)<\mathit{thirdParam}[w]$}{
      $\mathit{thirdParam}[w]\leftarrow\mathit{low}_2(c)$\;
    }
    $c\leftarrow$ next child of $w$ on $T_\mathit{highDec}$\;
  }
  $\mathit{currentChild}[w]\leftarrow c$\;
  $\mathtt{lowest\_low1}[d]\leftarrow\mathit{firstParam}[w]$\;
  $\mathtt{second\_lowest\_low1}[d]\leftarrow\mathit{secondParam}[w]$\;
  $\mathtt{lowest\_low2}[d]\leftarrow\mathit{thirdParam}[w]$\;
}

\end{algorithm}

\subsubsection{The case where $M_p(d)\in C$}
\label{section:M(d)inC,M(c)=bot}

In the case where $M_p(d)\in C$, we have that at least one of $L_p(d)$ and $R_p(d)$ is not a descendant of $w$, and therefore it is on $C$. Let us assume that $L_p(d)\in C$. (The case where $R_p(d)\in C$ is treated similarly.) Then, since $M_p(c)=\bot$, we have $l_1(L_p(d))\geq u$. If $l_1(L_p(d))>u$, then $l_1(L_p(d))\in B$ (because $l_1(L_p(d))$ must be a proper ancestor of $p(d)=v$). Thus, $B$ and $C$ are connected directly with the back-edge $(L_p(d),l_1(L_p(d)))$. So let us assume that $l_1(L_p(d))=u$. Then we can use the leftmost and the rightmost points, $\widetilde{L}$ and $\widetilde{R}$, respectively, of $d$, that skip $u$ on $T_\mathit{highDec}$. We can have those points computed during the preprocessing phase according to Proposition~\ref{proposition:skipping}. Specifically, during the preprocessing phase, we generate a query, for every vertex $d$, for the leftmost and the rightmost points of $d$ that skip $l_1(L_p(d))$, on the tree $T_\mathit{highDec}$. We collect all those queries, and we use Proposition~\ref{proposition:skipping} to compute all of them in linear time in total. (And we do the same w.r.t. the $l_1(R_p(\cdot))$ points, in order to prepare for dealing with the case that $R_p(d)\in C$.) 

Now, if both $\widetilde{L}$ and $\widetilde{R}$ are $\bot$, then we conclude that $B$ is not connected with $C$ (because all back-edges $(x,y)$ with $x$ a descendant of $d$ and $y$ a proper ancestor of $p(d)=v$ have $y=u$). If either of those is on $C$, then we have that $B$ is connected with $C$ directly with a back-edge. (E.g., if $\widetilde{L}\in C$, then either $(\widetilde{L},l_1(\widetilde{L}))$ or $(\widetilde{L},l_2(\widetilde{L}))$ connects $B$ and $C$, because we have that either $l_1(\widetilde{L})>u$, or $l_1(\widetilde{L})=u$ and $l_2(\widetilde{L})>u$ and $l_2(\widetilde{L})<p(d)=v$.) So let us assume that neither of $\widetilde{L}$ and $\widetilde{R}$ is on $C$. This implies that both $\widetilde{L}$ and $\widetilde{R}$ are descendants of $w$. Due to Assumption~\ref{assumption2}, this implies that both of them are proper descendants of $w$. So let $d'$ be the child of $w$ that is an ancestor of $\widetilde{L}$. (We note that $d'$ can be computed in constant time using a level-ancestor query as explained in Section~\ref{section:basicDFS}.) Then it is sufficient to check whether $\mathit{high}_p(d')\in C$, because we have that $C$ is connected with $B$ if and only if $\mathit{high}_p(d')\in C$. To see this, consider first that, since both $\widetilde{L}$ and $\widetilde{R}$ are descendants of $w$, there are no back-edges from $B$ to $C$. (I.e., because every back-edge $(x,y)$ with $x\in B$ and $y<p(d)$ must satisfy that $y=u$.) Thus, if $C$ and $B$ are connected, this can only be through hanging subtrees of $w$. Then, by applying almost the same argument that was used in the proof of Lemma~\ref{lemma:lbelowwonT_high_MAIN} (where we replace vertex ``$c$" with ``$d$", ``$v$" with ``$w$", and ``$d$" with ``$d'$", and the set of vertices ``$A$" with ``$B$", and ``$B$" with ``$C$"), we can see that, since we are working on $T_\mathit{highDec}$, there is a hanging subtree of $w$ that connects $C$ and $B$ if and only if $T(d')$ has this property. The only difference with Lemma~\ref{lemma:lbelowwonT_high_MAIN} is that here we can use only those descendants $x$ of $d$ from which stem back-edges of the form $(x,y)\in B_p(d)$ with $y\neq u$, and this is why we use $\widetilde{L}$ instead of $L_p(d)$.

\subsection{The case where $M_p(c)\in B$}
\label{section:M(c)inB}
In the case where $M_p(c)\in B$, we have that $B$ is connected with $A$ directly with a back-edge. To see this, consider the following. Since $M_p(c)\in B$, we have that $M_p(c)$ is not a descendant of $v$. Thus, there is a back-edge $(x,y)\in B_p(c)$ such that $x$ is a descendant of $c$, but not a descendant of $v$. Therefore, $x\in B$. Furthermore, we have $y\in A$, because $y$ is a proper ancestor of $p(c)=u$. Thus, we only need to determine whether $C$ is connected with either $B$ or $A$, either directly or through hanging subtrees (whose roots are children of $w$). As in the previous case (where $M_p(c)=\bot$), here we further distinguish three cases, which are treated almost identically as previously: either $(1)$ $M_p(d)$ is a descendant of a child $d'$ of $w$, or $(2)$ $M_p(d)=w$, or $(3)$ $M_p(d)\in C$. Notice that these cases are mutually exclusive and exhaust all possibilities for $M_p(d)$, except the case $M_p(d)=\bot$. However, the case where $M_p(d)=\bot$ is trivial, because we can simply conclude that $C$ is isolated from $B$ and $A$, because this means that there is no back-edge $(x,y)$ with $x\in T(d)$ and $y<p(d)=v$.

\subsubsection{The case where $M_p(d)$ is a descendant of a child of $w$}
In the case where $M_p(d)$ is a descendant of a child $d'$ of $w$, we have that there is no back-edge that directly connects $C$ with either $B$ or $A$, and that $T(d')$ is the only hanging subtree that may connect $C$ with $B$, or $C$ with $A$. Thus, we only have to check $(1)$ whether there is a back-edge $(x,y)$ with $x\in T(d')$ and $y\in C$, and $(2)$ whether there is a back-edge $(x',y')$ with $x'\in T(d')$ and either $y'\in B$ or $y'\in A$. Then we have that $C$ is connected with $B$ and $A$ if and only if $(1)$ and $(2)$ are true. It should be clear that $(1)$ is equivalent to $\mathit{high}_p(d')\in C$. For $(2)$ we use $\mathit{high}_p(d)$ and $\mathit{low}_p(d')$. It should be clear that there is a back-edge $(x',y')$ with $x'\in T(d')$ and $y'\in B$ if and only if $\mathit{high}_p(d)\in B$, and there is a back-edge $(x',y')$ with $x'\in T(d')$ and $y'\in A$ if and only if $\mathit{low}(d')\in A$. (See Figure~\ref{figure:McInB}.)

\begin{figure}[h!]\centering
\includegraphics[trim={0cm 22cm 0 0cm}, clip=true, width=1\linewidth]{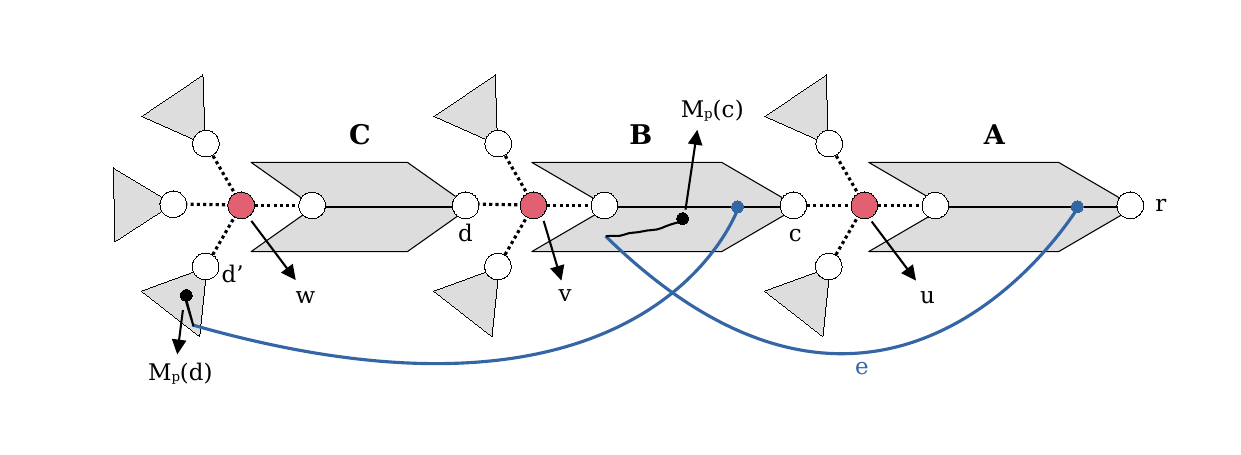}
\caption{\small{Here $M_p(c)\in B$, and $M_p(d)$ is a descendant of a child $d'$ of $w$. The condition $M_p(c)\in B$ implies that there is a back-edge that connects $B$ and $A$ directly (e.g., $e$). Then, we have to determine whether $T(d')$ is connected with $C$ (which is equivalent to $\mathit{high}_p(d')\in C$), and whether $T(d')$ is connected with either $B$ or $A$. Since $M_p(d)\in T(d')$, we have that $T(d')$ is connected with a back-edge with $B$ if and only if $\mathit{high}_p(d)\in B$.}}\label{figure:McInB}
\end{figure} 

\subsubsection{The case where $M_p(d)=w$}
In this case, as in the previous one, there is no back-edge that directly connects $C$ with either $B$ or $A$, but there may exist several hanging subtrees of $w$ that  connect $C$ with $B$, or $C$ with $A$. However, it is sufficient to determine the existence of only one hanging subtree that either connects $C$ with $B$, or $C$ with $A$. To do this, we use the parameters $(i)$, $(ii)$ and $(iii)$ that we also used in Section~\ref{section:M(d)=w,M(c)=bot}, in the case where $M_p(c)=\bot$ and $M_p(d)=w$, and that we can have available after a linear-time preprocessing. Recall that $(i)$ is the lowest $\mathit{low}_1$ point of the children of $w$ that have $\mathit{high}_p\in C$, $(ii)$ is the second-lowest such $\mathit{low}_1$ point, and $(iii)$ is the lowest $\mathit{low}_2$ point of those children of $w$. We will show that it is sufficient to check whether one of $(i)$, $(ii)$, or $(iii)$ is in $B$ or in $A$, in order to determine whether $C$ is connected with either $B$ or $A$ through a hanging subtree of $w$. To be precise, we have that $C$ is connected with $B$ or with $A$ through the subtree of a child of $w$ if and only if either $(i)$, or $(ii)$ or $(iii)$ is in $B$ or in $A$. Notice that the ``$\Leftarrow$" direction of this equivalence is trivial, and thus it is sufficient to demonstrate the ``$\Rightarrow$" direction.

So let us suppose that either $(a)$ there is a child $d'$ of $w$ such that $T(d')$ connects $C$ and $A$, or $(b)$ there is a child $d'$ of $w$ such that $T(d')$ connects $C$ and $B$. In either case, it should be clear that $\mathit{high}_p(d')\in C$. First, let us assume that $(a)$ is true. Then, it should be clear that $\mathit{low}_1(d')\in A$. Therefore, $(i)$ is also in $A$. So let us suppose that $(a)$ is not true, and therefore $(b)$ is true, and let $d'$ be a child of $w$ that satisfies $(b)$ with minimum $\mathit{low}_1$ point. Since $T(d')$ connects $C$ and $B$, this means that there is a back-edge $(x,y)$ with $x\in T(d')$ and $y\in B$. Notice that $\mathit{low}_1(d'')\notin A$ for any child $d''$ of $w$ that has $\mathit{high}_p(d'')\in C$, because otherwise $(a)$ would be true. Thus, we have that $(i)$ is either $u$ or greater than $u$. In particular, either $\mathit{low}_1(d')=u$, or $\mathit{low}_1(d')>u$. If $\mathit{low}_1(d')=u$, then surely we must have $\mathit{low}_2(d')\in B$, because $T(d')$ connects $C$ and $B$. Since $(i)\geq u$, in this case we have $(i)=u$. This implies that $(iii)$ is in $B$. Otherwise, (i.e., if $\mathit{low}_1(d')>u$), then we must have that either $(i)$ or $(ii)$ is in $C$, due to the minimality condition of $d'$.

\subsubsection{The case where $M_p(d)\in C$}
The argument in this case is almost identical to that in Section~\ref{section:M(d)inC,M(c)=bot}, where $M_p(c)=\bot$ and $M_p(d)\in C$. Thus, initially we will use either the leftmost or the rightmost point $x$ of $d$ (whichever is not a descendant of $w$), and check whether $l_1(x)\in A$, or $l_1(x)=u$ and $l_2(x)\in B$, or $l_1(x)\in B$. If neither of those is true, then the only back-edge from $B_p(d)$ provided by $x$ is $(x,u)$, and so we will use the leftmost and the rightmost points of $d$ that skip $u$ on $T_\mathit{highDec}$. If either of those is in $C$, then $C$ is connected directly with a back-edge either with $B$ or with $A$. (If those points do not exist, then $C$ is isolated from $B$ and $A$.) Otherwise, we consider the child $d'$ of $w$ that is an ancestor of the leftmost such skipping point, and we check whether $\mathit{high}_p(d')\in C$ (which, in this case, is true if and only if $C$ is connected either with $B$ or with $A$).

\subsection{The case where $M_p(c)=v$}
\label{section:Mp(c)=v}
In the case where $M_p(c)=v$, we have that $B$ is not connected with $A$ directly with a back-edge, but $B$ and $A$ may be connected through the mediation of children of $v$. Specifically, it may be that:\\

$(*)$ there is a child $d'$ of $v$, with $d'\neq d$, with the property that $\mathit{high}_p(d')\in B$ and $\mathit{low}(d')\in A$. (See Figure~\ref{figure:M(c)=v}$(a)$.)\\

\begin{figure}[h!]\centering
\includegraphics[trim={0cm 13.2cm 0 0cm}, clip=true, width=1\linewidth]{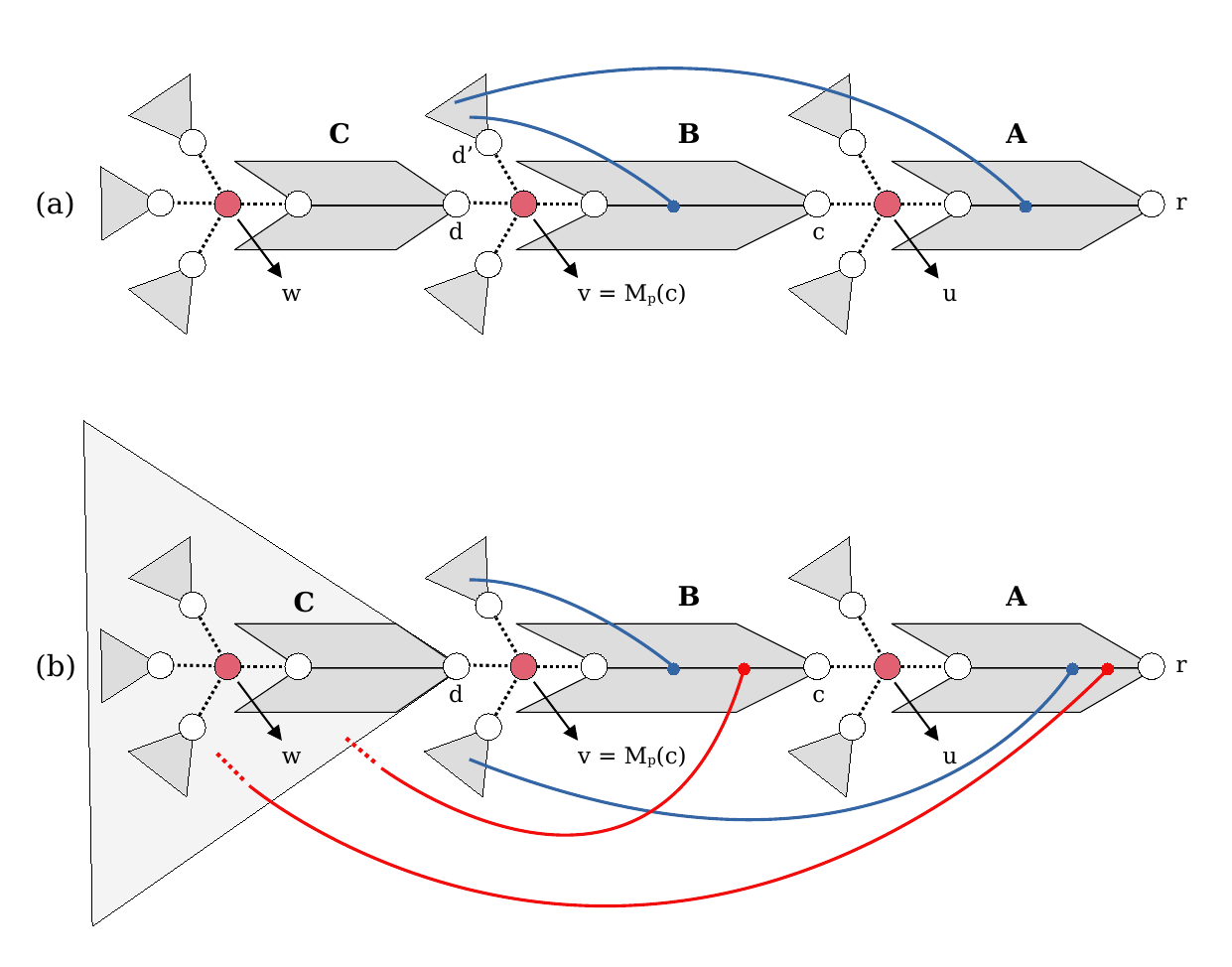}
\caption{\small{An illustration of two subcases in the case where $M_p(c)=v$. In $(a)$ there is a child $d'$ of $M_p(c)$, with $d'\neq d$, such that $\mathit{high}_p(d')\in B$ and $\mathit{low}(d')\in A$. In this case we have that $A$ and $B$ are connected through the hanging subtree $T(d')$ of $v$, and then we can establish the connection between $C$ and either $A$ or $B$ using precisely the same techniques as in Section~\ref{section:M(c)inB} (for the case where $M_p(c)\in B$). In case $(b)$, we have that $d$ is the only child of $M_p(c)$ with $\mathit{high}_p(d)\in B$ and $\mathit{low}(d)\in A$. This is the most difficult subcase (case $(3)$ in this section) that may arise in the case where $M_p(c)=v$, because the connection of the parts $A$, $B$ and $C$ hinges on back-edges that stem from $T(d)$, but $T(d)$ is possibly broken into many components, because it contains $w$.}}\label{figure:M(c)=v}
\end{figure}

In this case, we have that $A$ is connected with $B$ through the hanging subtree $T(d')$, and the connection between $C$ and either $B$ or $A$ is established precisely as in Section~\ref{section:M(c)inB}, where $M_p(c)\in B$.

We can easily check the condition $(*)$ if we have computed, during the preprocessing phase, two children $d_1$ and $d_2$ of $v$ with the property that their $\mathit{high}_p$ point is in $B$, and their $\mathit{low}_1$ point is as low as possible, and let us assume w.l.o.g. that $\mathit{low}_1(d_1)\leq\mathit{low}_1(d_2)$. (Notice: $d_1$ and $d_2$ are associated with $c$ as two children of $M_p(c)$ that have their $\mathit{low}$ point as low as possible, while their $\mathit{high}_p$ is in $T[p(M_p(c)),c]$. In particular, it may be that neither of those children exists, or that only one exists, or that one of them is unique in having the lowest $\mathit{low}$ point with $\mathit{high}_p\in T[p(M_p(c)),c]$, while the other has a greater $\mathit{low}$ point, but still the second lowest.) It is easy to see why those two children are sufficient: we have that $(*)$ is true if and only if either $\mathit{low}_1(d_1)\in A$ and $d_1\neq d$, or $\mathit{low}_1(d_2)\in A$. The computation of $d_1$ and $d_2$ during the preprocessing phase is easy if we work on $T_\mathit{highDec}$ (e.g., see how we have computed the parameters $(i)$, $(ii)$ and $(iii)$ in Section~\ref{section:M(d)=w,M(c)=bot}).

So let us assume that $(*)$ is not true. This means that, for every child $d'$ of $v$ with $d'\neq d$, we have that either $\mathit{high}_p(d')\notin B$, or $\mathit{low}(d')\notin A$. Thus, if $B$ is to be connected with $A$, this can only happen through the mediation of $T(d)$. 
Here we distinguish five cases for $d$: either $(1)$ $M_p(d)=\bot$, or $(2)$ $\mathit{high}_p(d)\in B$ and $\mathit{low}(d)\geq u$, or $(3)$ $\mathit{high}_p(d)\in B$ and $\mathit{low}(d)\in A$, or $(4)$ $\mathit{high}_p(d)=\mathit{low}(d)=u$, or $(5)$ $\mathit{high}_p(d)\leq u $ and $\mathit{low}(d)\in A$. It should be clear that these cases are mutually exclusive and exhaust all possibilities for $d$. We will explain briefly how to treat cases $(1)$, $(2)$, $(4)$, and $(5)$, because they are either trivial, or we can handle them precisely with the techniques that were presented in Section~\ref{section:M(c)=bot} (i.e., in the case where $M_p(c)=\bot$). However, case $(3)$ proves to be more complicated, for reasons that we will explain shortly.

In case $(1)$, we simply have that the parts $A$, $B$ and $C$ are isolated from each other, because there is no back-edge in $B_p(d)$ (and thus $C$ is isolated from $B$ and $A$), and there is no other hanging subtree of $v$ that could connect $B$ and $A$. In case $(4)$, we have that every back-edge $(x,y)\in B_p(d)$ satisfies $y=u$. Thus, we infer again that the parts $A$, $B$ and $C$ are isolated from each other. In case $(2)$, we have that $A$ is isolated from $B$ and $C$, and it remains to determine whether $B$ is connected with $C$ (either directly with a back-edge, or through a hanging subtree of $w$). To do this, we can work as in Section~\ref{section:M(c)=bot} (i.e., the case where $M_p(c)=\bot$), by distinguishing the three cases for $M_p(d)$, and using precisely the same techniques. Case $(5)$ is treated similarly as case $(2)$: the only difference is that here $B$ is isolated from $A$ and $C$, and it remains to determine whether $C$ is connected with $A$.

Case $(3)$ is more complicated, because the conditions $\mathit{high}_p(d)\in B$ and $\mathit{low}(d)\in A$ imply that $A$ may be connected with $B$ through the mediation of $T(d)$, but this is not necessarily the case, because $T(d)$ may be broken up into several pieces, since it contains $w$. (See Figure~\ref{figure:M(c)=v}$(b)$.) First we will provide a criterion in order to determine whether $C$ is connected with $B$ either directly with a back-edge, or through a hanging subtree of $w$. (Notice, however, that if $C$ and $B$ are not connected with one of those two ways, they may still be connected through the mediation of $A$.) Then we will provide criteria in order to determine the connection of $A$ and $B$, and $A$ and $C$, through either of those two ways. In any case, we crucially rely on the fact that $d$ is the only child of $M_p(c)$ with the property that $\mathit{high}_p(d)\in B$ and $\mathit{low}(d)\in A$. In turns out that this is enough information to allow for sufficient preparation for this case during the preprocessing phase. 

By Lemma~\ref{lemma:segmentpointsspecial}, we can assume that we have access to the leftmost and the rightmost descendants of $d$ with the property that they provide a back-edge whose lower endpoint is in $B$. Specifically, we compute those points on $T_\mathit{highDec}$, and let us call them $\widetilde{L}$ and $\widetilde{R}$, respectively. If they are $\bot$, then we conclude that $C$ and $B$ are not connected directly with a back-edge or through a hanging subtree of $w$ (and thus they may be connected only through the mediation of $A$). So let us suppose that $\widetilde{L}$ and $\widetilde{R}$ are not $\bot$. Now, if either of those is in $C$, then we obviously have that $C$ and $B$ are connected directly with a back-edge. Otherwise, we have that $\widetilde{L}$ and $\widetilde{R}$ are both descendants of $w$. In this case, it is sufficient to find the child $d'$ of $w$ that is an ancestor of $\widetilde{L}$ (using a level-ancestor query, as explained in Section~\ref{section:basicDFS}), and check whether $\mathit{high}_p(d')\in C$. Then, since we are working on the tree $T_\mathit{highDec}$, we have that there is a hanging subtree of $w$ that connects $C$ and $B$ if and only if $\mathit{high}_p(d')\in C$.

Thus, we have demonstrated how to check in constant time whether $C$ is connected with $B$ either directly with a back-edge, or through a hanging subtree of $w$.
Now we distinguish two cases, depending on whether $(a)$ $B$ is connected with $C$ either directly with a back-edge or through a hanging subtree of $w$, or $(b)$ there is no such connection between $B$ and $C$.

In case $(a)$, we only have to determine whether $A$ is connected either with $B$ or with $C$, directly or through a hanging subtree of $w$. Here we rely on the fact that $d$ is completely determined by $c$ (i.e., it is the only child of $M_p(c)$ with $\mathit{high}_p(d)\in T[p(M_p(c)),c]$ and $\mathit{low}(d)<p(c)$). Thus, we can use Lemma~\ref{lemma:segmentpointsspecial} in order to compute (more precisely: to assume that we have computed, during the preprocessing phase) the leftmost and the rightmost descendant of $d$ with the property that they provide leaping back-edges from $B_p(c)$. Specifically, we compute those points on $T_\mathit{lowInc}$, and let us call them $L(c,d)$ and $R(c,d)$, respectively. (Notice that neither of them is $\bot$, because $\mathit{low}(d)\in A$.) Now, it may be that one of them is in $C$. This implies that $C$ is connected with $A$ directly with a back-edge, and thus the parts $A$, $B$ and $C$ are connected. If this is not the case, then we have that both $L(c,d)$ and $R(c,d)$ are descendants of $w$. In this case there are two possibilities: either $(a1)$ both $L(c,d)$ and $R(c,d)$ are descendants of the same child $d'$ of $w$, or $(a2)$ $L(c,d)$ and $R(c,d)$ are descendants of different children of $w$. In case $(a1)$, we can determine $d'$ by using a level-ancestor query as explained in Section~\ref{section:basicDFS}. Then it is sufficient to check $\mathit{high}^p_1(d')$ and $\mathit{high}^p_2(d')$. Specifically, we consider the following scenaria. If $\mathit{high}^p_1(d')\in C$, then $C$ is connected with $A$ through $T(d')$. Otherwise, if $\mathit{high}^p_1(d')=v$ and $\mathit{high}^p_2(d')\in B$, or if $\mathit{high}^p_1(d')\in B$, then $B$ is connected with $A$ through $T(d')$. Otherwise, we conclude that $A$ is isolated from $B$ and $C$. 

Now let us consider case $(a2)$. In this case, we can have $w$ associated with $c$ during the preprocessing phase, as $w$ is the nearest common ancestor of $L(c,d)$ and $R(c,d)$ (and we assume that all the necessary nearest common ancestors that we need have been computed off-line during the preprocessing phase, using e.g. \cite{DBLP:journals/siamcomp/BuchsbaumGKRTW08}). Now let $d'$ and $d''$ be the children of $w$ that are ancestors of $L(c,d)$ and $R(c,d)$, respectively (which can be determined using level-ancestor queries, as explained in Section~\ref{section:basicDFS}). Then, since we work on $T_\mathit{lowInc}$, notice that $d'$ is the first child of $w$, and the segment of children of $w$ from $d'$ to $d''$ are all the children of $w$ that have $\mathit{low}$ point lower than $u$ (and therefore, from every hanging subtree induced by such a child stems at least one back-edge to $A$). Then, for this segment, it is sufficient to have computed: $(i)$ the greatest $\mathit{high}^p_1$ point, $(ii)$ the second-greatest $\mathit{high}^p_1$ point, and $(iii)$ the greatest $\mathit{high}^p_2$ point. (Notice that the parameters $(i)$, $(ii)$ and $(iii)$ are analogous to those defined in Section~\ref{section:M(d)=w,M(c)=bot}, for the case $M_p(c)=\bot$ and $M_p(d)=w$, and can be computed in linear time during the preprocessing phase using a procedure similar to Algorithm~\ref{algorithm:three_low_parameters}.) Then it is easy to see that $A$ is connected with either $B$ or $C$ through a hanging subtree of $w$ if and only if at least one of $(i)$, $(ii)$ and $(iii)$ is in either $B$ or $C$.

In case $(b)$, we have to determine separately whether $A$ is connected with $B$, and whether $A$ is connected with $C$, either directly with a back-edge (for $A$-$C$) or through a hanging subtree of $w$ (either for $A$-$B$ or for $A$-$C$). In order to check the connection of $A$ and $C$, we use the same idea with the $L(c,d)$ and $R(c,d)$ points that we discussed in the previous paragraph. However, if this criterion establishes a connection between $A$ and $C$, then it cannot be used in order to determine whether $A$ is connected with $B$ in such a way. (This is because there may exist a hanging subtree $T(d')$ of $w$ from which stem back-edges to both $B$ and $A$, but the $\mathit{high}^p_1$ and $\mathit{high}^p_2$ points of $d'$ are descendants of $v$, and thus we will not be able to use $\mathit{high}^p_1(d')$ or $\mathit{high}^p_2(d')$ to certify the connection of $B$ with $T(d')$.) In any case, in order to determine the connection between $A$ and $B$ (through a hanging subtree of $w$), we can use the $\widetilde{L}$ and $\widetilde{R}$ points defined above, but this time computed on $T_\mathit{lowInc}$ (using Lemma~\ref{lemma:segmentpointsspecial}). (Notice that $\widetilde{L}$ and $\widetilde{R}$ cannot be $\bot$, because $\mathit{high}_p(d)\in B$.) Since $B$ is not connected with $C$ directly with a back-edge, we have that both of those points are descendants of $w$. In this case, it is sufficient to find the child $d'$ of $w$ that is an ancestor of $\widetilde{L}$ (using a level-ancestor query as explained in Section~\ref{section:basicDFS}), and check whether $\mathit{low}(d')\in A$. Then, since we are working on the tree $T_\mathit{lowInc}$, we have that there is a hanging subtree of $w$ that connects $B$ and $A$ if and only if $\mathit{low}(d')\in A$. 

Thus, to summarize the algorithm for handling case $(3)$. First, we check whether $C$ and $B$ are connected either directly with a back-edge, or through a hanging subtree of $w$. This is done by using the $\widetilde{L}$ and $\widetilde{R}$ points computed on $T_\mathit{highDec}$. If we are able to establish such a connection, then we are in case $(a)$ described above, and it remains to check whether $A$ is connected with either $B$ or $C$. This is done by using the $L(c,d)$ and $R(c,d)$ points computed on $T_\mathit{lowInc}$. Otherwise, we are in case $(b)$, and we have to check separately whether $A$ is connected with $B$, and whether $A$ is connected with $C$. In order to determine the connection of $A$ and $C$, we use the $L(c,d)$ and $R(c,d)$ points on $T_\mathit{highDec}$. And in order to determine the connection of $A$ and $B$, we use the $\widetilde{L}$ point on $T_\mathit{lowInc}$.

\subsection{The case where $M_p(c)$ is a descendant of a child $d'$ of $w$}
\label{section:Mp(c)Descw}
In this case, $M_p(c)$ is a descendant of a child $d'$ of $w$. Thus, if the parts $C$ and $B$ are connected with $A$, this can only be through the mediation of $T(d')$. However, it may also be that $C$ is connected with $B$ (either directly with a back-edge, or through the mediation of another hanging subtree of $w$), but $A$ is isolated from $C$ and $B$. In order to determine the connection between the parts $A$, $B$ and $C$, we find it again convenient to distinguish three cases for $M_p(d)$: either $(1)$ $M_p(d)$ is a descendant of $d'$, or $(2)$ $M_p(d)=w$, or $(3)$ $M_p(d)\in C$. It should be clear that those cases are mutually exclusive, but we have to explain a little bit why they exhaust all possibilities for $M_p(d)$. First, notice that, in order to have an exhaustive list, we should also consider the cases: $(4)$ $M_p(d)$ is a descendant of a child $d''$ of $w$ with $d''\neq d'$, and $(5)$ $M_p(d)=\bot$. However, those cases are impossible: since $c$ is an ancestor of $d$ and $M_p(c)$ is a descendant of $d$, Lemma~\ref{lemma:Mp} implies that $M_p(d)$ (exists and) is an ancestor of $M_p(c)$. Thus, $M_p(d)\neq\bot$, and $M_p(d)$ cannot be a descendant of a child of $w$ different than $d'$.

\subsubsection{The case where $M_p(d)$ is a descendant of $d'$}
In the case where $M_p(d)$ is a descendant of $d'$ we have that there is no back-edge that connects $C$ directly with $B$, and that $T(d')$ is the only hanging subtree that may connect the parts $A$, $B$ and $C$. Thus, it is sufficient to check for $(1)$ the existence of a back-edge $(x,y)$ with $x\in T(d')$ and $y\in C$, and $(2)$ the existence of a back-edge $(x,y)$ with $x\in T(d')$ and $y\in B$. (In any case, we know that there is a back-edge $(x,y)$ with $x\in T(d')$ and $y\in A$, since $M_p(c)$ is a descendant of $d'$.) It is easy to see that the existence of a back-edge with property $(1)$ is equivalent to $\mathit{high}_p(d')\in C$, and the existence of a back-edge with property $(2)$ is equivalent to $\mathit{high}_p(d)\in B$. Thus, we can determine in constant time whether $C$ is connected with $A$ through $T(d')$, and whether $B$ is connected with $A$ through $T(d')$. 

\subsubsection{The case where $M_p(d)=w$}
\label{section:Mp(c)descDandMp(d)=w}
In the case where $M_p(d)=w$ we have that there is no back-edge that connects $C$ directly with $B$, but there may be several hanging subtrees of $w$ that connect $C$ and $B$, and only $T(d')$ has the potential to connect $A$ with either $B$ or $C$.

First, we will explain how to determine and handle the case where $T(d')$ does not connect $A$ with $B$ or $C$. Notice that this is equivalent to one of the following: either $(1)$ $\mathit{high}^p_1(d')\in A$, or $(2)$ $\mathit{high}^p_1(d')=u$, or $(3)$ $\mathit{high}^p_1(d')=v$ and $\mathit{high}^p_2(d')\leq u$. In this case, we have that $A$ is isolated from $B$ and $C$, and it remains to determine the existence of a hanging subtree of $w$ that connects $B$ and $C$. To do this, we follow precisely the same procedure as in Section~\ref{section:M(d)=w,M(c)=bot} (for the case where $M_p(c)=\bot$ and $M_p(d)=w$).

Otherwise, we have that $T(d')$ either connects $A$ with $B$, or $A$ with $C$ (or both). First, let us consider the case that $T(d')$ connects $A$ with $B$, but not $A$ with $C$. It is easy to see that this is equivalent to $\mathit{high}^p_1(d')\in B$, or $\mathit{high}^p_1(d')=v$ and $\mathit{high}^p_2(d')\in B$. In this case, it remains to determine the existence of a hanging subtree of $w$ that connects $B$ and $C$. To do this, we can follow again precisely the same procedure as in Section~\ref{section:M(d)=w,M(c)=bot} (for the case where $M_p(c)=\bot$ and $M_p(d)=w$).

Finally, the most difficult case appears when it is definitely true that $T(d')$ connects $A$ and $C$ (which is equivalent to $\mathit{high}_p(d')\in C$). This case is more complicated, because now it may be that $T(d')$ also connects $A$ and $B$, and it happens to be the only hanging subtree of $w$ with this property. It is not straightforward to determine this, and we can only show how to do it indirectly, through a counting argument. We will be enabled to use this argument only if we have established first that no other hanging subtree of $w$ connects $B$ and $C$, and so $T(d')$ is the only candidate to do this. 

First, in order to determine whether there is a hanging subtree of $w$ different from $T(d')$ that connects $B$ and $C$, we can use essentially the same argument as in Section~\ref{section:M(d)=w,M(c)=bot}, but with a necessary change. Recall that the search for such a subtree in Section~\ref{section:M(d)=w,M(c)=bot} uses three parameters: $(i)$ the lowest $\mathit{low}_1$ point among the children of $w$ that have $\mathit{high}_p\in C$, $(ii)$ the second-lowest such $\mathit{low}_1$ point, and $(iii)$ the lowest such $\mathit{low}_2$ point. Here we can also use those parameters, if we have excluded $d'$ from the children of $w$ (which also has $\mathit{high}_p(d')\in C$). What enables us to use these new parameters is that $d'$ is a distinctly identifiable child of $w$: it is the unique one that has the lowest $\mathit{low}_1$ point among the children of $w$. Thus, we can modify the parameters $(i)$, $(ii)$ and $(iii)$, as $(i')$, $(ii')$ and $(iii')$, where the only difference in their definition is that they exclude from the children of $w$ under consideration the unique one with the lowest $\mathit{low}_1$ point. (Thus, $(i')$ e.g. can be defined as: the lowest $\mathit{low}_1$ point among the children of $w$, except the one that has the lowest $\mathit{low}_1$ point, that have $\mathit{high}_p\in C$.) We can have those parameters computed in the preprocessing phase in linear time in total for all vertices (where we focus only on those $d$ whose $M_p(d)$ has a unique child with lowest $\mathit{low}_1$ point), with an algorithm almost identical to the one in Section~\ref{section:M(d)=w,M(c)=bot} that computes $(i)$, $(ii)$ and $(iii)$. Thus, we can simply use Algorithm~\ref{algorithm:three_low_parameters}, where we just add the condition ``\textbf{if} $c$ is the unique child of $w$ with the lowest $\mathit{low}$ point, \textbf{then continue}", when we enter the \textbf{while} loop in Line~\ref{line:whileofiiiiii}. Then, with an argument almost identical to the one in Section~\ref{section:M(d)=w,M(c)=bot}, we can see that it is sufficient to check whether any of $(i')$, or $(ii')$, or $(iii')$ is in $B$, in order to determine whether there is a hanging subtree of $w$, distinct from $T(d')$, that connects $C$ and $B$. In this case, notice that we are done, since we can conclude that all parts $A$, $B$ and $C$ are connected (since $A$ was easily established to be connected with $C$).

Thus, it remains to assume that there is no hanging subtree of $w$, distinct from $T(d')$, that connects $C$ and $B$, and now we have to check whether $T(d')$ connects $C$ and $B$. As we said earlier, here we will use a counting argument. The idea is to determine the number of back-edges in $B_p(d)$ that stem from the children of $w$ that are distinct from $d'$, and then also use $|B_p(d)|$ and $|B_p(c)|$, in order to determine whether there is at least one back-edge $(x,y)\in B_p(d)$ with $x\in T(d')$ and $y\in B$. (Notice that the existence of such a back-edge establishes that $T(d')$ is connected with $B$, and we already know that $T(d')$ is connected with $C$, since $\mathit{high}_p(d')\in C$.)

Now let us see this counting argument in detail. First, from the fact that no hanging subtree of $w$, distinct from $T(d')$, connects $C$ and $B$, we can infer that: 
 for every child $d''$ of $w$, with $d''\neq d'$, that has $\mathit{high}_p(d'')\in C$ and provides a back-edge $(x,y)\in B_p(d)$ with $x\in T(d'')$, we have $y=\mathit{low}(d'')=u$. (This is because $y$ must be lower than $v$, but not in $B$, and it must also be higher than $p(u)$ (since $d''\neq d'$). Thus, $y=u$ is the only available option.) 

Thus, we can distinguish three types of children of $w$, distinct from $d'$, that provide back-edges from $B_p(d)$: 

\begin{itemize}
\item[]{\textbf{Type-1:} the children $d''$ of $w$ with $\mathit{high}_p(d'')<v$.}
\item[]{\textbf{Type-2:} the children $d''$ of $w$ with $\mathit{high}_p(d'')=v$ and $\mathit{low}(d'')<v$.}
\item[]{\textbf{Type-3:} the children $d''$ of $w$ with $\mathit{high}_p(d'')\in C$ and $\mathit{low}(d'')=u$.}
\end{itemize}

Notice that these three types do not overlap, and that, based on the previous paragraph, they exhaust all possibilities for the children of $w$ that are distinct from $d'$ and provide, from the subtree that they induce, a back-edge from $B_p(d)$. To be more precise, we have the following: if there is a back-edge $(x,y)\in B_p(d)$ such that $x\notin T(d')$, then $x$ is a descendant of a child of Type-1, -2, or -3, and if it is a descendant of a child of Type-3 then we have $y=u$.

Now we will show how to count the number of back-edges from $B_p(d)$ that stem from children of $w$ distinct from $d'$. Furthermore, we will also need the sum of the lower endpoints of those back-edges. These computations can be performed in constant time, provided that we have computed some parameters for $d$ during the preprocessing phase. Specifically, we will need the following items:

\begin{itemize}
\item[]{\textbf{item1a:} the sum $|B_p(d_1)|+\dots+|B_p(d_t)|$, where $d_1,\dots,d_t$ are the children of $M_p(d)$ with $\mathit{high}_p(d_i)\leq p(d)$, for $i\in\{1,\dots,t\}$.}
\item[]{\textbf{item1b:} the sum $\mathit{sumY}(d_1)+\dots+\mathit{sumY}(d_t)$, where $d_1,\dots,d_t$ are the children of $M_p(d)$ with $\mathit{high}_p(d_i)\leq p(d)$, for $i\in\{1,\dots,t\}$.}
\item[]{\textbf{item2:} the number of back-edges $(x,y)$ where $y=p(d)$ and $x$ is a descendant of a child of $M_p(d)$ distinct from the unique one that has the lowest $\mathit{low}$ point.}
\item[]{\textbf{item3:} the sum $\mathit{numLow}(e_1)+\dots+\mathit{numLow}(e_s)$, where $e_1,\dots,e_s$ are the children of $M_p(d)$ with $\mathit{high}_p(e_i)>p(d)$, for $i\in\{1,\dots,s\}$, excluding the unique child of $M_p(d)$ that has the lowest $\mathit{low}$ point.}
\end{itemize}

First, let us see how these items can help us in order to determine the existence of a back-edge $(x,y)\in B_p(d)$ with $x\in T(d')$ and $y\in B$. Let $N$ denote the number of back-edges $(x,y)\in B_p(d)$ where $x$ is a descendant of a child of $w$ distinct from $d'$, and let $S$ denote the sum of the lower endpoints of those back-edges. Now we will determine $N$ and $S$ using the above items. Here we rely on the typology of the children of $w$ that can provide back-edges from $B_p(d)$. For every child $d''$ of Type-1, we clearly have $B_p(d'')\subset B_p(d)$. Furthermore, for every child $d''$ of Type-2, we have that all back-edges from $B_p(d'')$ are included in $B_p(d)$, except those whose lower endpoint is $v$. Finally, based on our above discussion, every child $d''$ of Type-3 provides $\mathit{numLow}(d'')$ edges to $B_p(d)$. Thus, it is easy to see that $N=\mathit{item1a}-\mathit{item2}+\mathit{item3}$. Similarly, we can see that $S=\mathit{item1b}-(\mathit{item2}\cdot v)+(\mathit{item3}\cdot u)$. Now, notice that $T(d')$ can provide two types of back-edges from $B_p(d)$: $(a)$ those that are included in $B_p(c)$, and $(b)$ those that are not included in $B_p(c)$. Thus, we are interested in the back-edges of type $(b)$, and specifically in those that do not have the form $(x,u)$. Now, the number of back-edges of type $(b)$ is obviously $N'=|B_p(d)|-N-|B_p(c)|$, and the sum of the lower endpoints of those back-edges is $S'=\mathit{sumY}(d)-S-\mathit{sumY}(c)$. If all these back-edges have the form $(x,u)$, with $x\in T(d')$, then we must have $S'=N'\cdot u$. Otherwise, we have $S'>N'\cdot u$ (because we have excluded from $B_p(d)$ the back-edges whose lower endpoint is lower than $u$, since these are precisely those in $B_p(c)$). Thus, this inequality is equivalent to the existence of a back-edge $(x,y)\in B_p(d)$ with $x\in T(d')$ and $y\in B$.


Finally, it is easy to see how we can have the parameters $\mathit{item1a}$, $\mathit{item1b}$, $\mathit{item2}$, and $\mathit{item3}$, computed in linear time during the preprocessing phase. For $\mathit{item1a}$ and $\mathit{item1b}$, we can simply work on $T_\mathit{highDec}$ and process the vertices in increasing order, in order to be able to determine, for every vertex $d$, the segment of the children list of $M_p(d)$ that consists of the children that have their $\mathit{high}_p$ point $\leq p(d)$ (if there are any). (Notice that the same idea is applied in Algorithm~\ref{algorithm:three_low_parameters}, although there the processing of the vertices is done in decreasing order.) If this segment gets expanded during the processing of $d$, then we add the extra $|B_p(\cdot)|$ and $\mathit{sumY}(\cdot)$ values, of the new children that were discovered, to two values associated with $M_p(d)$ in order to maintain $\mathit{item1a}$ and $\mathit{item1b}$. For $\mathit{item3}$, we work similarly, but we process the vertices $d$ in decreasing order (precisely as was done in Algorithm~\ref{algorithm:three_low_parameters}), in order to compute the initial segment of the children list of $M_p(d)$, that consists of those children with $\mathit{high}_p\in T[p(M_p(d)),d]$. For $\mathit{item2}$, we simply process the vertices $d$ in increasing order, and the incoming back-edges to the vertices $p(d)$ also in increasing order w.r.t. their higher endpoint. Specifically, for every vertex $d$, we process the incoming back-edges to $p(d)$. As long as the current incoming back-edge $(x,p(d))$ to $p(d)$ satisfies that $x$ is a descendant of $d$, we check whether $x$ is a proper descendant of $M_p(d)$, but not of the child of $M_p(d)$ that has the lowest $\mathit{low}$ point (if such a child of $M_p(d)$ exists). Then we move to the next incoming back-edge to $p(d)$, until the higher endpoint of the current back-edge is no longer a descendant of $d$ (or until we have exhausted the list of the incoming back-edges to $p(d)$). It should be clear that this idea is correct and can be implemented in linear time (specifically, the sorting of the lists of the incoming back-edges in increasing order w.r.t. their higher endpoint can be performed in linear time with bucket-sort).

\subsubsection{The case where $M_p(d)\in C$}
\label{section:M(c)inDM(d)inC}

In the case where $M_p(d)\in C$ we have that $C$ and $B$ may be connected directly with a back-edge, or through a hanging subtree of $w$. Since $M_p(d)\in C$, we have that at least one of $L_p(d)$ and $R_p(d)$ is in $C$. Let us assume that $L_p(d)\in C$. (The case where $R_p(d)\in C$ is treated similarly.) Since $M_p(c)$ is a descendant of $w$, we have $l_1(L_p(d))\geq u$. (Otherwise, $M_p(c)$ would be an ancestor of $L_p(d)$, which is impossible.) If $l_1(L_p(d))>u$, then $C$ and $B$ are connected directly with the back-edge $(L_p(d),l_1(L_p(d)))$. Otherwise, (i.e., if $l_1(L_p(d))=u$), we have to use  Lemma~\ref{lemma:segmentpointsL}, in order to get the leftmost and the rightmost descendant of $d$, $\widetilde{L}$ and $\widetilde{R}$, respectively, with the property that it provides a back-edge to $B$. Specifically, we have $\widetilde{L}$ and $\widetilde{R}$  computed on $T_\mathit{highDec}$. Now, if $\widetilde{L}$ and $\widetilde{R}$ are $\bot$, then $C$ is not connected with $B$ directly with a back-edge or through a hanging subtree of $w$ (and therefore, in this particular case, not at all). If either of $\widetilde{L}$ and $\widetilde{R}$ is in $C$, then there is a back-edge that directly connects $C$ and $B$. Otherwise, both $\widetilde{L}$ and $\widetilde{R}$ are proper descendants of $w$ (due to Assumption~\ref{assumption2}). Then, it is easy to see that $C$ and $B$ are connected through a hanging subtree of $w$ if and only if $\mathit{high}_p(d'')\in C$, where $d''$ is the child of $w$ that is an ancestor of $\widetilde{L}$. 

Now it remains to determine whether $A$ is connected with either $B$ or $C$ (through the hanging subtree $T(d')$ of $w$). It is easy to see the following. If $\mathit{high}_p(d')\leq u$, or $\mathit{high}_1^p(d')=v$ and $\mathit{high}_2^p(d')\leq u$, then $A$ is isolated from $B$ and $C$. If $\mathit{high}_p(d')\in B$, or $\mathit{high}_1^p(d')=v$ and $\mathit{high}_2^p(d')\in B$, then $T(d')$ can only connect $A$ and $B$. And if $\mathit{high}_p(d')\in C$, then $T(d')$ definitely connects $C$ and $A$. In this case, there is no need to check whether $T(d')$ connects $A$ and $B$, because either the test of the previous paragraph has determined that $C$ and $B$ are not connected through a hanging subtree of $w$, or it has determined that they are connected in such a way, and therefore we conclude that all parts $A$, $B$ and $C$ remain connected. 
 
\subsection{The case where $M_p(c)=w$}
\label{section:MpC=w}
In the case where $M_p(c)=w$ we have that $A$ can be connected with either $B$ or $C$ only through hanging subtrees of $w$. Here we will need a criterion to determine whether $A$ is connected with $B$ or with $C$ through at least one hanging subtree of $w$. Notice that such a hanging subtree $T(d')$ (where $d'$ is a child of $w$) must have $\mathit{low}(d')\in A$. Thus, in the preprocessing phase, we must gather some information for all children $d'$ of $w$ with $\mathit{low}(d')\in A$ collectively. Specifically, we compute the following parameters among the children of $w$ that have their $\mathit{low}$ point in $A$: $(i)$ the greatest $\mathit{high}^p_1$ point, $(ii)$ the second-greatest $\mathit{high}^p_1$ point, and $(iii)$ the greatest $\mathit{high}^p_2$ point. Notice that the parameters $(i)$, $(ii)$ and $(iii)$ here are analogous to those in Section~\ref{section:M(d)=w,M(c)=bot}, and they are easily computed with a similar algorithm (i.e., with an algorithm analogous to Algorithm~\ref{algorithm:three_low_parameters}), by working on the tree $T_\mathit{lowInc}$.

Since $M_p(c)$ is a descendant of $d$, and $d$ is a descendant of $c$, by Lemma~\ref{lemma:Mp} we have that $M_p(d)$ is an ancestor of $M_p(c)$. Thus, it is convenient to distinguish two cases: either $M_p(d)=w$, or $M_p(d)\in C$.

\subsubsection{The case where $M_p(d)=w$}
\label{section:caseM(c)=wM(d)=w}
In the case where $M_p(d)=w$ we have that $C$ may be connected with $B$ through hanging subtrees of $w$ (or through the mediation of $A$, but not directly with a back-edge). Thus, we need to determine whether there is a child $d'$ of $w$ with $\mathit{high}_p(d')\in C$ with the property that there exists a back-edge $(x,y)$ with $x\in T(d')$ and $y\in B$. Since $M_p(d)=M_p(c)$ and $d$ is a descendant of $c$, here it is very convenient to use the leftmost vertex $z$ on $T_\mathit{highDec}$ with the property that $z$ is a descendant of $M_p(d)$ and there is a back-edge $(z,y)$ such that $y\in T[p(p(d)),c]\subseteq B$. (We can get $z$ in constant time using the oracle described in Lemma~\ref{lemma:segmentpointsMp}.) If this point does not exist, then we infer that $C$ is not connected with $B$ through a hanging subtree of $w$. Otherwise, we take the child $d'$ of $w$ that is an ancestor of $z$ (using a level-ancestor query as explained in Section~\ref{section:basicDFS}), and we check whether $\mathit{high}_p(d')\in C$. It is easy to see that this condition is equivalent to there being a hanging subtree of $w$ that connects $C$ and $B$ (i.e., because, if there is one such hanging subtree, then $T(d')$ is definitely one of them). 

Now we have two cases to consider: either $C$ is connected with $B$ through a hanging subtree of $w$, or this is not true. (Notice, however, that the second case does not necessarily imply that $C$ is disconnected from $B$, because $C$ and $B$ may still be connected through the mediation of $A$.) Let us consider the first case first. In this case, it is sufficient to check whether there is a hanging subtree of $w$ that connects $A$ and $C$, or $A$ and $B$. Using the parameters $(i)$, $(ii)$ and $(iii)$, it is easy to see the following two facts:

\begin{itemize}
\item[(1)]{There is a hanging subtree of $w$ that connects $A$ and $C$ if and only if $(i)$ is in $C$.}
\item[(2)]{If $(1)$ is not true, then there is a hanging subtree of $w$ that connects $A$ and $B$ if and only if either $(i)$, or $(ii)$, or $(iii)$, is in $B$.}
\end{itemize}

Notice the relation between the above two criteria: $(2)$ can determine the connection of $A$ and $B$ through a hanging subtree of $w$, but only if $(1)$ is not true. (Otherwise, if $(1)$ is true, then the ``$\Rightarrow$'' direction of the equivalence in $(2)$ is no longer valid.) Thus, if we have established that $B$ and $C$ are connected, then those two criteria are sufficient in order to determine the connection between $A$, $B$, and $C$.

So let us finally consider the case where there is no hanging subtree of $w$ that connects $C$ and $B$. Thus, we have to determine separately whether $A$ is connected with $C$ through a hanging subtree of $w$, and whether $A$ is connected with $B$ through a hanging subtree of $w$. Using the criterion $(1)$, we can check the connection of $A$ and $C$. If this criterion fails, then $C$ is isolated from $A$ and $B$, and we can use $(2)$ in order to check whether $A$ is connected with $B$. 

However, if we have established through $(1)$ that $A$ is connected with $C$, then we can no longer use $(2)$ in order to check whether $A$ is connected with $B$ as well. Thus, we have to use a different criterion for that. Here we use the leftmost vertex $z$ on $T_\mathit{lowInc}$ with the property that $z$ is a descendant of $d$ and there is a back-edge $(z,y)$ such that $y\in B$ (using the oracle described in Lemma~\ref{lemma:segmentpointsMp}). (Notice that this point, if it exists, must be a descendant of $M_p(d)$.) If this point does not exist, then we infer that $A$ is not connected with $B$. Otherwise, we take the child $d'$ of $w$ that is an ancestor of $z$ (using a level-ancestor query as explained in Section~\ref{section:basicDFS}), and we check whether $\mathit{low}(d')\in A$. It is easy to see that this condition is equivalent to there being a hanging subtree of $w$ that connects $A$ and $B$ (i.e., because, if there is one such hanging subtree, then $T(d')$ is definitely one of them). 

\subsubsection{The case where $M_p(d)\in C$} 

In the case where $M_p(d)\in C$, we first check whether $B$ is connected with $C$, either directly with a back-edge or through a hanging subtree of $w$. In order to do this, we work precisely as in Section~\ref{section:M(c)inDM(d)inC} (for the case where $M_p(c)$ is a descendant of $w$, and $M_p(d)\in C$). Recall that the argument in that section utilized the fact that $M_p(d)\in C$, and that $M_p(c)$ is a descendant of $d$.

Now it remains to check whether $A$ is connected with $B$, and whether $A$ is connected with $C$, through hanging subtrees of $w$. Here we apply the criteria provided by facts $(1)$ and $(2)$ of Section~\ref{section:caseM(c)=wM(d)=w}. (We note that these facts do not depend on $M_p(d)=w$, which was assumed in that section, but only on $M_p(c)=w$.) Thus, we first check whether $(i)\in C$, because this is equivalent to $A$ being connected with $C$ through a hanging subtree of $w$. (The parameters $(i)$, $(ii)$ and $(iii)$ were defined in the first paragraph of Section~\ref{section:MpC=w}.) If we determine that $A$ is not connected with $C$ through a hanging subtree of $w$, then we can use criterion $(2)$ in order to determine whether $A$ is connected with $B$ through a hanging subtree of $w$, by checking whether either $(i)$, or $(ii)$, or $(iii)$ is in $B$. In this case, we are done. (Because we have gathered all the information that is needed in order to determine the connection of $A$, $B$ and $C$.) 

However, if the first criterion has established the connection of $A$ and $C$ through a hanging subtree of $w$, then we can no longer apply criterion $(2)$ in order to determine whether $A$ is connected with $B$ through a hanging subtree of $w$. Of course, there is no need to do that if we have already established that $B$ is connected with $C$, either directly with a back-edge, or through a hanging subtree of $w$. However, if we have determined that $B$ is not connected with $C$ in one of those two ways, then it might still be that $B$ is connected with $C$ through the mediation of $A$. Specifically, it might be that there is a hanging subtree of $w$ that connects $A$ and $C$, and a different hanging subtree of $w$ that connects $A$ and $B$. To determine this case, we utilize the information that $B$ is not connected with $C$ directly with a back-edge. Specifically, as we saw in Section~\ref{section:M(c)inDM(d)inC}, this implies that either $l_1(L_p(d))=u$ or $l_1(R_p(d))=u$. So let us assume that $l_1(L_p(d))=u$, since the other case is treated similarly. Then we use Lemma~\ref{lemma:segmentpointsL}, in order to get the leftmost and the rightmost descendant of $d$, $\widetilde{L}$ and $\widetilde{R}$, respectively, with the property that it provides a back-edge to $B$. Specifically, we have $\widetilde{L}$ and $\widetilde{R}$ computed on $T_\mathit{lowInc}$. Since $C$ is not connected with $B$ directly with a back-edge, we have that either both those points are $\bot$, or they are descendants of $w$. In the first case, we conclude that $B$ is isolated from $A$ and $C$. Otherwise, since we work on $T_\mathit{lowInc}$, it is easy to see that there is a hanging subtree of $w$ that connects $A$ and $B$ if and only if $\mathit{low}(d')<u$, where $d'$ is the child of $w$ that is an ancestor of $\widetilde{L}$.

\subsection{The case where $M_p(c)\in C$}
\label{section:MpCinC}
In the case where $M_p(c)\in C$ we have that $M_p(c)$ is a descendant of $d$. Therefore, since $d$ is a descendant of $c$, Lemma~\ref{lemma:Mp} implies that $M_p(d)$ is an ancestor of $M_p(c)$. In this case, we also have that at least one of $L_p(c)$ and $R_p(c)$ is not a descendant of $w$ (otherwise, $M_p(c)$ would be a descendant of $w$, in contradiction to $M_p(c)\in C$). Thus, one of $L_p(c)$ and $R_p(c)$ is in $C$, and therefore $C$ is connected with $A$ directly with a back-edge. Thus, it remains to determine whether $B$ is connected with either $A$ or $C$. Since $M_p(c)\in C$, notice that there are only three ways in which $B$ can be connected with $A$ and $C$: either there is a back-edge from $C$ to $B$, or there is a hanging subtree of $w$ that connects $C$ and $B$, or $B$ and $A$.  

Initially, we distinguish between the case where $M_p(d)=M_p(c)$, and the case where $M_p(d)$ is a proper ancestor of $M_p(c)$. So let us suppose first that $M_p(d)=M_p(c)$. Since we are interested in back-edges from $B_p(d)$ whose lower endpoint is in $B$, here it is very convenient to use the leftmost and the rightmost descendants of $d$ with the property that there is a back-edge from them to $B$. Since $M_p(d)=M_p(c)$, we can get those points from the oracle described in Lemma~\ref{lemma:segmentpointsMp}. More precisely, we first ask for those points on $T_\mathit{highDec}$, and let us call them $\widetilde{L}_\mathit{high}$ and $\widetilde{R}_\mathit{high}$, respectively. If these are $\bot$, then $B$ is isolated from $A$ and $C$. Otherwise, it may be that one of them is in $C$, in which case $C$ is connected with $B$ directly with a back-edge. And finally, if both of them are descendants of $w$, then we have that $\widetilde{L}_\mathit{high}$ is a descendant of a child $d'$ of $w$ (which we can find using a level-ancestor query, as explained in Section~\ref{section:basicDFS}). In this case, since we work on the tree $T_\mathit{highDec}$, it is sufficient to ask whether $\mathit{high}_p(d')\in C$. In this case, we have that $C$ is connected with $B$ through the hanging subtree $T(d')$. Otherwise, we conclude that there is no hanging subtree of $w$ that connects $C$ and $B$ (because all other children of $w$ from whose induced hanging subtrees stem back-edges that reach $B$, have lower $\mathit{high}_p$ point than that of $d'$, and therefore their $\mathit{high}_p$ must be an ancestor of $v$). Then, we use again the leftmost and the rightmost descendants of $d$ with the property that there is a back-edge from them to $B$ (by using the oracle described in Lemma~\ref{lemma:segmentpointsMp}), but this time we ask for them on $T_\mathit{lowInc}$, and let us call them $\widetilde{L}_\mathit{low}$ and $\widetilde{R}_\mathit{low}$, respectively. Now, we have that both $\widetilde{L}_\mathit{low}$ and $\widetilde{R}_\mathit{low}$ are descendants of $w$ (because, since we have established this for $\widetilde{L}_\mathit{high}$ and $\widetilde{R}_\mathit{high}$, this is an invariant for the leftmost and rightmost points that start from $T(d)$ and end in $B$, across all permutations of the base DFS tree). Then, we consider the child $d''$ of $w$ which is an ancestor of $\widetilde{L}_\mathit{low}$ (which we can determine using a level-ancestor query, as explained in Section~\ref{section:basicDFS}). And now it is sufficient to ask whether $\mathit{low}(d'')\in A$. Notice that there is a hanging subtree of $w$ that connects $B$ and $A$ if and only if $T(d'')$ has this property.

So let us consider the case where $M_p(d)$ is a proper ancestor of $M_p(c)$. For this case, it is convenient to work on the tree $T_\mathit{lowInc}$. Then, by Lemma~\ref{lemma:lowestl1}, we have that $L_p(d)$ is a descendant of $d$ with the lowest $l_1$ point. Therefore, since $M_p(c)$ is a descendant of $d$, we have that $L_p(d)$ is a descendant of $M_p(c)$ (with $l_1(L_p(d))<p(c)=u$). Since $M_p(d)$ is a proper ancestor of $M_p(c)$, we have that $R_p(d)$ cannot also be a descendant of $M_p(c)$. Thus, $l_1(R_p(d))\geq u$. Now we distinguish two cases, depending on whether $l_1(R_p(d))=u$ or $l_1(R_p(d))>u$. In the case where $l_1(R_p(d))=u$, we can use Lemma~\ref{lemma:segmentpointsL}, in order to get the leftmost and the rightmost descendants of $d$ with the property that there is a back-edge from them to $B$. More precisely, we first use those points computed on $T_\mathit{highDec}$, and then (if needed) on $T_\mathit{lowInc}$, and we work as previously (i.e., in the case where $M_p(c)=M_p(d)$), in order to determine whether $C$ is connected with $B$ directly with a back-edge, or whether there is a  hanging subtree of $w$ that connects $C$ and $B$, or $B$ and $A$, or whether $B$ is isolated from $A$ and $C$.  


Thus, we are left to consider the case where $l_1(R_p(d))>u$. (Notice that this implies that $l_1(R_p(d))\in B$.) If $R_p(d)\in C$, then we are done: we conclude that $C$ and $B$ are connected directly with a back-edge, e.g., $(R_p(d),l_1(R_p(d)))$. So let us assume that $R_p(d)\notin C$. This implies that $R_p(d)$ is a descendant of $w$.

Since $L_p(d)\neq R_p(d)$, by Assumption~\ref{assumption2} we have that both $L_p(d)$ and $R_p(d)$ are proper descendants of $M_p(d)$, and therefore they are descendants of different children of $M_p(d)$ (since $M_p(d)=\mathit{nca}\{L_p(d),R_p(d)\}$). Then, since $L_p(d)$ is a descendant of $M_p(c)$, and since $R_p(d)$ is a descendant of $w$, and since $M_p(c)$ and $w$ are proper descendants of $M_p(d)$, we have that $M_p(c)$ and $w$ are descendants of different children of $M_p(d)$. 
%
%
Since we are working on $T_\mathit{lowInc}$, we have that $M_p(c)$ is a descendant of the first child $c'$ of $M_p(d)$ (because $c'$ is the child of $M_p(d)$ with the lowest $\mathit{low}$ point among those children). Now we are asking the following question: what is the $\mathit{low}$ point of the next sibling $c''$ of $c'$? Notice that, since $M_p(c)$ is a descendant of $c'$, we must have that $\mathit{low}(c'')\geq u$. (Otherwise, $M_p(c)$ would be an ancestor of both $c'$ and $c''$, and therefore an ancestor of $M_p(d)$, and therefore we would have $M_p(c)=M_p(d)$.) Here we distinguish two cases: either $\mathit{low}(c'')=u$, or $\mathit{low}(c'')>u$. In the case where $\mathit{low}(c'')=u$, we can use Lemma~\ref{lemma:segmentpointslow}, in order to get the leftmost and the rightmost descendants of $d$ with the property that there is a back-edge from them to $B$. Thus, in this case we can work as previously.

So let us suppose that $\mathit{low}(c'')>u$. If $c''$ is not an ancestor of $w$, then we are done: we conclude that $C$ and $B$ are connected directly with a back-edge (because there is one that starts from $T(c'')$ and ends at $\mathit{low}(c'')$). Otherwise, we have that $c'$ and $c''$ are the only children of $M_p(d)$ with low enough $\mathit{low}$ point to provide back-edges from $B_p(d)$. (This is because $c''$ is the child of $M_p(d)$ which is an ancestor of $R_p(d)$.) In this case, we have to determine whether $(a)$ there is a back-edge $(x,y)$ with $x\in T(c')$ and $y\in B$, or $(b)$ there is a back-edge $(x,y)$ with $x\in T(c'')\setminus T(w)$ and $y\in B$, or $(c)$ there is a hanging subtree of $w$ that connects $C$ and $B$. 
First we will check $(b)$ and $(c)$ together, and if we determine that such connections do not exist, then we can check $(a)$ using a counting argument.

Since $c''$ is the second child of $M_p(d)$ on $T_\mathit{lowInc}$, we have that $c''$ is determined by $d$, and so we can use Proposition~\ref{proposition:L(v,d)} in order to get (or more precisely: in order to ensure that we have constant-time access to) the leftmost and the rightmost descendants of $c''$ that provide back-edges from $B_p(d)$. Specifically, we will compute those points on $T_\mathit{highDec}$, and let us call them $\widetilde{L}$ and $\widetilde{R}$, respectively. (I.e., $\widetilde{L}=L(d,c'')$ and $\widetilde{R}=R(d,c'')$.) If either of $\widetilde{L}$ and $\widetilde{R}$ is in $C$ (i.e., not a descendant of $w$), then we are done: this establishes that $C$ is connected with $B$ directly with a back-edge (starting from one of $\widetilde{L}$ and $\widetilde{R}$). Otherwise, we have that both $\widetilde{L}$ and $\widetilde{R}$ are descendants of $w$. In this case, it is sufficient to determine the child $d'$ of $w$ that is an ancestor of $\widetilde{L}$ (using a level-ancestor query as explained in Section~\ref{section:basicDFS}), and check whether $\mathit{high}_p(d')\in C$. Due to the definition of $\widetilde{L}$ on $T_\mathit{highDec}$, notice that there is a hanging subtree of $w$ that connects $C$ and $B$ if and only if $\mathit{high}_p(d')\in C$. This concludes the check for the cases $(b)$ and $(c)$.

Now suppose that we have established that the connections in cases $(b)$ and $(c)$ do not exist, using the test of the previous paragraph. This means that $\widetilde{L}$ and $\widetilde{R}$ are descendants of $w$, and every child $d'$ of $w$ with $\mathit{low}(d')<p(d)$ has $\mathit{high}_p(d')\leq p(d)$. We can use this information in order to determine whether there is a back-edge $(x,y)$ with $x\in T(c')$ and $y\in B$. (Notice that such a back-edge is in $B_p(d)$.) We distinguish four cases for a back-edge $(x,y)\in B_p(d)$:

\begin{enumerate}[label={(\roman*)}]
\item{$x\in T(c')$ and $y\in B$.}
\item{$x\in T(c')$ and $y=u$.}
\item{$x\in T(c')$ and $y<u$.}
\item{$x\in T(c'')$ and $y\in B$.}
\end{enumerate}

Notice that cases $(i)$ to $(iv)$ provide a partition of $B_p(d)$. Also, notice that the back-edges of type $(iii)$ are precisely those from $B_p(c)$. Then, with a counting argument we can determine the existence of a back-edge of type $(i)$. So let $N$ be the number of back-edges of type $(iv)$, and let $S$ be the sum of the lower endpoints of those back-edges. Then, the parameters $N'=|B_p(d)|-|B_p(c)|-N$ and $S'=\mathit{sumY}(d)-\mathit{sumY}(c)-S$ give us the number, and the sum of the lower endpoints, of the back-edges of types $(i)$ and $(ii)$. Then, since the lower endpoint of every back-edge of type $(ii)$ is $u$ (which is lower than any vertex in $B$), we have that there is a back-edge of type $(i)$ if and only if $S'>N'\cdot u$. 

Thus, it remains to show how to compute $N$ and $S$ in constant time, and, in particular, what are the parameters that we must have computed during the linear-time preprocessing for that purpose. First, we distinguish two cases, depending on whether $\widetilde{L}$ and $\widetilde{R}$ are descendants of the same child $d'$ of $w$ or not. (This can be checked, and $d'$ can be determined, in constant time, using a level-ancestor query as explained in Section~\ref{section:basicDFS}.) 

If both $\widetilde{L}$ and $\widetilde{R}$ are descendants of the same child $d'$ of $w$, then the important question is whether $\mathit{high}_p(d')<v$ or $\mathit{high}_p(d')=v$. (These are the only two possible cases.) In the case where $\mathit{high}_p(d')<v$, we simply have that $N=|B_p(d')|$, and $S=\mathit{sumY}(d')$. In the case where $\mathit{high}_p(d')=v$, we have that $N=|B_p(d')|-\mathit{numHigh}(d')$, and $S=\mathit{sumY}(d')-(\mathit{numHigh}(d')\cdot v)$. 

Finally, let us consider the case where $\widetilde{L}$ and $\widetilde{R}$ are descendants of different children of $w$. In this case, we have to know the following parameters:

\begin{itemize}
\item[]{\textbf{item1a:} the sum $|B_p(d_1)|+\dots+|B_p(d_t)|$, where $d_1,\dots,d_t$ are the children of $w$ with $\mathit{high}_p(d_i)\leq p(d)$, for $i\in\{1,\dots,t\}$.}
\item[]{\textbf{item1b:} the sum $\mathit{sumY}(d_1)+\dots+\mathit{sumY}(d_t)$, where $d_1,\dots,d_t$ are the children of $w$ with $\mathit{high}_p(d_i)\leq p(d)$, for $i\in\{1,\dots,t\}$.}
\item[]{\textbf{item2:} the number of back-edges $(x,y)$ where $y=p(d)$ and $x$ is a descendant of a child $d_i$ of $w$ with $\mathit{high}_p(d_i)\leq p(d)$.}
\end{itemize}

Then, it is easy to see that $N=\mathit{item1a}-\mathit{item2}$, and $S=\mathit{item1b}-(\mathit{item2}\cdot v)$.
It is not difficult to see how we can have those parameters computed in linear time during the preprocessing phase. The important observation is that $w$ is determined by $d$ as: the nearest common ancestor of the leftmost and the rightmost descendants of $d$ that provide a back-edge from $B_p(d)$ and are also descendants of the child of $M_p(d)$ that has the second lowest $\mathit{low}$ point. Thus, all vertices $d$ for which those leftmost and rightmost points are defined, are ancestors of their associated $w$. Then, we can use a similar algorithm to compute those items as the one described in Section~\ref{section:Mp(c)descDandMp(d)=w}, where we used similar items to handle the case where $M_p(c)$ is a descendant of a child of $w$ and $M_p(d)=w$.

\bibliography{test}

\end{document}